\documentclass[acmsmall,dvipsnames]{acmart}
\makeatletter
\def\@acmplainindent{0pt}
\def\@acmdefinitionindent{0pt}
\def\@proofindent{\noindent}
\makeatother

\usepackage{svgcolor}
\usepackage{local}
\usepackage{wrapfig}

\usepackage{ifthen}
\newboolean{isExtendedVersion}
\setboolean{isExtendedVersion}{true} % set to false for the condensed version

\newcommand{\appendixref}[1]{%
  \ifthenelse{\boolean{isExtendedVersion}}%
    {\Cref{#1}}% If it's the extended version, reference within the document
    {the extended version of this paper {\cite{Moeller2024}}}% If it's the condensed version, cite the extended document
}

\setcopyright{rightsretained}
\acmJournal{PACMPL}
\acmYear{2024}
\copyrightyear{2024}
\acmSubmissionID{pldi24main-p646-p}
\acmVolume{8}
\acmNumber{PLDI}
\acmArticle{224}
\acmMonth{6}
\acmDOI{10.1145/3656454}
\received{2023-11-16}
\received[accepted]{2024-03-31}

\begin{document}

\author[M. Moeller]{Mark Moeller}
\authornote{Equal contribution.}
\affiliation{%
  \institution{Cornell University}%
  \city{Ithaca, NY}
  \country{USA}%
}
\orcid{0009-0002-9512-565X}
\email{mdm367@cornell.edu}
\author[J. Jacobs]{Jules Jacobs}
\authornotemark[2]
\affiliation{%
  \institution{Cornell University}%
  \city{Ithaca, NY}
  \country{USA}%
}
\orcid{0000-0003-1976-3182}
\email{jj758@cornell.edu}
\author[O. Savary Belanger]{Olivier Savary Belanger}
\affiliation{%
  \institution{Galois}%
  \city{Portland, OR}
  \country{USA}%
}
\orcid{0009-0008-2809-6047}
\email{olivier@galois.com}
\author[D. Darais]{David Darais}
\affiliation{%
  \institution{Galois}%
  \city{Portland, OR}
  \country{USA}%
}
\orcid{0000-0003-2314-0287}
\email{darais@galois.com}
\author[C. Schlesinger]{Cole Schlesinger}
\affiliation{%
  \institution{Galois}%
  \city{Portland, OR}
  \country{USA}%
}
\orcid{0009-0004-9350-3041}
\email{coles@galois.com}
\author[S. Smolka]{Steffen Smolka}
\affiliation{%
  \institution{Google}%
  \city{Mountain View, CA}
  \country{USA}%
}
\orcid{0000-0001-8625-5490}
\email{smolkaj@google.com}
\author[N. Foster]{Nate Foster}
\affiliation{%
  \institution{Cornell University}%
  \city{Ithaca, NY}
  \country{USA}%
}
\orcid{0000-0002-6557-684X}
\email{jnfoster@cs.cornell.edu}
\author[A. Silva]{Alexandra Silva}
\affiliation{%
  \institution{Cornell University}%
  \city{Ithaca, NY}
  \country{USA}%
}
\orcid{0000-0001-5014-9784}
\email{alexandra.silva@cornell.edu}

\setcitestyle{numbers}

\title{\KATch: A Fast Symbolic Verifier for NetKAT}

\begin{abstract}
  We develop new data structures and algorithms for checking verification queries in NetKAT, a domain-specific language for specifying the behavior of network data planes. Our results extend the techniques obtained in prior work on symbolic automata and provide a framework for building efficient and scalable verification tools. We present \KATch, an implementation of these ideas in Scala, featuring an extended set of NetKAT operators that are useful for expressing network-wide specifications, and a verification engine that constructs a bisimulation or generates a counter-example showing that none exists. We evaluate the performance of our implementation on real-world and synthetic benchmarks, verifying properties such as reachability and slice isolation, typically returning a result in well under a second, which is orders of magnitude faster than previous approaches. Our advancements underscore NetKAT's potential as a practical, declarative language for network specification and verification.
\end{abstract}

\titlenote{This research was developed with funding from the Defense Advanced Research Projects Agency (DARPA) and the Office of Naval Research (ONR). The views, opinions, and/or findings contained in this material are those of the authors and should not be interpreted as representing the official views or policies of the Department of the Navy, DARPA, or the U.S. Government. The U.S. Government is authorized to reproduce and distribute reprints for Government purposes notwithstanding any copyright notation hereon. DISTRIBUTION A. Approved for public release; distribution unlimited.}

\maketitle

\framebox{\parbox{\linewidth}{\textbf{Erratum. } An earlier version of this paper contained a mistake in the SPP canonicalization procedure (\Cref{app:sppsc}), and a corresponding mistake in its correctness theorem statement and proof (\Cref{lem:sppscsem}), neither of which appeared in the main text. There is a description of the mistake and the correction in \Cref{app:sppsc}, and \Cref{lem:sppscsem} and its proof are corrected to match.}}

\section{Introduction}

In the automata-theoretic approach to verification, programs and specifications are each encoded as automata, and verification tasks are reduced to standard questions in formal language theory---e.g., membership, emptiness, containment, etc.~\cite{Vardi1986}. The approach became popular in the mid-1980s, driven by use of temporal logics and model checking in hardware verification, and it remains a powerful tool today. In particular, automata provide natural models of phenomena like transitive closure, which arises often in programs but is impossible to express in pure first-order logic.

NetKAT, a domain-specific language for specifying and verifying the behavior of network data planes, is a modern success story for the automata-theoretic approach.  NetKAT programs denote sets of traces (or ``histories'') where a trace is the list of packets ``seen'' at each hop in the network. The NetKAT language specifies these traces using a regular expression-like syntax, as follows:
\[
  p,q \Coloneqq \bot \mid \top \mid f \test v \mid f \testNE v \mid f \mut v \mid \text{dup} \mid p + q \mid p \cdot q \mid p^\star
\]

Unlike ordinary regular expressions, which are stateless, NetKAT programs manipulate state in packets. Accordingly, \NetKAT's atoms are not letters, but \emph{actions} that either drop or forward the current packet ($\bot$ or $\top$), modify a header field ($f \mut\, v$), test a header field against a value ($f \test v$, $f \testNE v$), or append the current packet to the trace ($\text{dup}$).
The regular expression operations copy and forward a packet to two different programs ($p + q$), sequence two programs ($p \cdot q$), or loop a program ($p^\star \equiv \top + p + p\!\cdot\!p + p\!\cdot\!p\!\cdot\!p + \cdots$).

A NetKAT program thus gives a declarative specification of the network's global behavior in terms of sets of traces. Specifically, we can model the location of a packet in the network using a special header field (usually named $\text{sw}$ for ``switch''), that we can update ($\text{sw} \mut \ell$) to logically move the packet from its current location to a new location $\ell$. Given a declarative description of the intended behavior, the NetKAT compiler generates a set of local forwarding tables for individual switches that together realize the global behavior~\cite{Foster2011,Smolka2015}.

Moreover, NetKAT not only plays a role as an implementation language, but also as a language for expressing verification queries, analogous to verification tools powered by SMT solvers~\cite{Leino2014,Barnett2005,Torlak2013}. In particular, as NetKAT includes a union operator (``+''), containment can be reduced to program equivalence---i.e., $p \sqsubseteq q$ if and only if $p + q \equiv q$. Hence, unlike other contemporary tools, which rely on bespoke encodings and algorithms for checking network properties like reachability, slice isolation, etc.~\cite{Kazemian2012,Khurshid2012,Yang2016,Fogel2015}, NetKAT allows a wide range of practical verification questions to be answered using the same foundational mechanism, namely program equivalence. For example,

\begin{itemize}
    \item{\textit{Network Reachability:}} Are there any packets that
      can go from location $121$ to $543$ in the network specified by the NetKAT program
      %Note from Alexandra: I made net in sf font because in the text it was not clear it was an expression
      $\textsf{net}$? Formally: $(\text{sw} \mut 121) \cdot \textsf{net}^\star
      \cdot (\text{sw} = 543) \qequiv \zero$
  \item{\textit{Slice Isolation:}} Are slices $\textsf{net}_1$ and
    $\textsf{net}_2$ logically disjoint, even though their
    implementation on shared infrastructure uses the same devices? Formally:
    $\textsf{net}_1^\star + \textsf{net}_2^\star \qequiv (\textsf{net}_1 + \textsf{net}_2)^\star$
\end{itemize}

\NetKAT's $\text{dup}$ primitive is a key construct that enables reducing a wide range of network-specific properties to program equivalence. Recall that $\text{dup}$ appends the headers of the current packet to the trace. Hence, to verify properties involving the network's internal behavior, we add intermediate packets to the trace using $\text{dup}$, making them relevant for program equivalence. Conversely, if we only care about the input-output behavior of the entire network, we can omit $\text{dup}$ from the specification, so intermediate packets are not considered by the check for program equivalence.

The story so far is appealing, but there is a major fly in the ointment. To decide program equivalence, the automata-theoretic approach relies on the translation from NetKAT programs to automata. But standard NetKAT automata have an enormous space of potential transitions---i.e., the ``alphabet'' has a ``character'' for every possible packet. So using textbook algorithms for building automata and checking equivalence would clearly be impractical. In fact, NetKAT equivalence is PSPACE-complete~\cite{Anderson2014}. Instead, what's needed are symbolic techniques for encoding the state space and transition structure of NetKAT automata that avoid combinatorial blowup in the common case.

Prior work on symbolic automata provides a potential solution to the problems that arise when working with large (or even infinite)
alphabets. The idea is to describe the transitions in the automata using logical formulas~\cite{DAntoni2014,DAntoni2017} or Binary Decision Diagrams (BDDs)~\cite{Pous2015} rather than concrete characters. Algorithms such as membership, containment, and equivalence can then be reformulated to work with these symbolic representations. Symbolic automata have been successfully applied to practical problems such as building input sanitizers for Unicode, which has tens of thousands of characters~\cite{Hooimeijer2011}. However, \NetKAT's richer semantics precludes the direct adoption of standard notions of symbolic automata. In particular, \NetKAT's transitions describe not only predicates but also transformations on the current packet. As Pous writes, it ``seems feasible'' to generalize his work to NetKAT, but
``not straightforward''~\cite{Pous2015}.

In this paper, we close this gap and develop symbolic techniques for checking program equivalence for NetKAT. In doing so, we address three key challenges:

\paragraph{Challenge 1: Expressive but Compact Symbolic Representations}
The state space and transitions of NetKAT automata are very large due to the way the ``alphabet'' is built from the space of all possible packets. Orthogonally, these characters also encode packet transformations, as reflected in \NetKAT's packet-processing semantics. It is crucial that the symbolic representations for NetKAT programs be compact and admit efficient equivalence checks.

\paragraph{Challenge 2: Extended Logical Operators}
In the tradition of regular expressions, NetKAT only includes sequential composition ($e_1 \cdot e_2$) and union operators ($e_1 +
e_2$)---the latter computes the union of the sets of traces described by $e_1$ and $e_2$. However, when verifying network-wide properties it is often useful to have operators for intersection ($e_1 \cap e_2$), difference ($e_1 \setminus e_2$), and symmetric difference ($e_1
\oplus e_2$). Having native support for these operators at the level of syntax and in equivalence checking algorithms allows the reduction of all verification queries to an emptiness check---e.g., $A \equiv B$
reduces to $A \oplus B \equiv \zero$, and $A \sqsubseteq B$ reduces to $A \setminus B \equiv \zero$.

\paragraph{Challenge 3: Support for Counter-Examples}
When an equivalence check fails, it can be hard to understand the cause of the failure, and which changes might be needed to resolve the problem. It is therefore important to be able to construct \emph{symbolic counter-examples} that precisely capture the input packets and traces that cause equivalence to fail.

\medskip In addressing the above challenges, we make the following technical contributions:

\paragraph{\textbf{Contribution 1: Efficient symbolic representations} (\Cref{sec:symbolic}).}
We design a symbolic data structure called Symbolic Packet Programs (\SPPn{}s) for representing both sets of packets and transformations on packets in a symbolic manner. \SPPn{}s generalize classic BDDs and are asymptotically more efficient than the representations used in prior work on NetKAT. In particular, \SPPn{}s support efficient sequential composition ($e_1 \cdot e_2$) and can be made \emph{canonical} which, with hash consing, allows us to check equivalence of \SPPn{}s in constant time.

\paragraph{\textbf{Contribution 2: Symbolic Brzozowski derivatives} (\Cref{sec:symaut}).}
We introduce a new form of symbolic Brzozowski derivative to produce automata for NetKAT programs.
Our approach naturally and efficiently supports the extended ``negative'' logical operators mentioned above, in contrast to prior approaches that only support ``positive'' operators like union.

\paragraph{\textbf{Contribution 3: Symbolic bisimilarity checking} (\Cref{sec:fwdbwd}).}
Building on the foundation provided by \SPPn{}s and deterministic NetKAT automata, we develop symbolic algorithms for checking bisimilarity. We present a symbolic version of the standard algorithm that searches in the forward direction through the state space of the automata under consideration. We also present a novel backward algorithm that computes symbolic counter-examples to equivalence---i.e., the precise set of input packets for which equivalence fails.

\paragraph{\textbf{Contribution 4: \KATch implementation} (\Cref{sec:impl}) \textbf{and evaluation} (\Cref{sec:eval}).}
Finally, we present a new verification tool, called \KATch, implemented in Scala. \KATch{} implements symbolic bisimilarity checking, including an extended set of extended logical operators, as well as symbolic counter-examples.  We evaluate the performance of \KATch{} on a variety of real-world topologies, and show that it efficiently answers a variety of verification queries and scales to much larger networks than prior work. Due to its use of our efficient data structures, \KATch{} is several orders of magnitude faster than prior implementations of NetKAT on these realistic examples. Moreover, \KATch{} is faster than prior work on synthetic combinatorial benchmarks by arbitrary large factors.

\section{\NetKAT Expressions and Automata}
\label{sec:netkat}

This section reviews basic definitions for \NetKAT to set the stage for the new contributions presented in the subsequent sections.
\NetKAT is a language for specifying the packet-forwarding behavior of network data planes~\cite{Anderson2014,Foster2015,Smolka2015}.\footnote{Networks have a control plane, which computes paths through the topology using distributed routing protocols (or a software-defined networking controller), and a data plane, which implements paths using high-speed hardware and software pipelines. NetKAT models the behavior of the latter.}  The syntax and semantics of \NetKAT is presented in \Cref{fig:synsem}, and is based on Kozen's Kleene Algebra with Tests (KAT) \cite{Kozen1996}.

\begin{figure}
    \renewcommand{\arraystretch}{1.1}
    \begin{tabular}{lcc}
        \textbf{Syntax} & \textbf{Description} & \textbf{Semantics} \\[1mm]
        \hline \\[-0.9em]
        $p,q \Coloneqq $ & & $\Sem{p}(\alpha\!:\!\Pk) : \pow{\Pk^+} =$ \\[1mm]
        \qquad$\mid\quad \bot$ & \emph{False} & $\emptyset$ \\
        \qquad$\mid\quad \top$ & \emph{True} & $\{ \alpha \}$ \\
        \qquad$\mid\quad f \test v$ & \emph{Test equals} & $\text{ if } \alpha_f = v \text{ then }  \{ \alpha \} \text{ else } \emptyset$ \\
        \qquad$\mid\quad f \testNE v$ & \emph{Test not equals} & $\text{ if } \alpha_f \neq v \text{ then } \{ \alpha \} \text{ else } \emptyset$ \\
        \qquad$\mid\quad f \mut\, v$ & \emph{Modification} & $\{ \alpha[f \mut\, v] \}$ \\
        \qquad$\mid\quad \text{dup}$ & \emph{Duplication} & $\{ \alpha\alpha \}$ \\
        \qquad$\mid\quad p + q$ & \emph{Union} & $\Sem{p}(\alpha) \cup \Sem{q}(\alpha)$ \\
        \qquad$\mid\quad p \cdot q$ & \emph{Sequencing} & $\{ \mathbf{a}\mathbf{b} \mid \mathbf{a}\beta \in \Sem{p}(\alpha),\ \mathbf{b} \in \Sem{q}(\beta) \}$ \\
        \qquad$\mid\quad p^\star$ & \emph{Iteration} & $\bigcup\limits_{n\geq 0} \Sem{p^n}$ \\[3mm]
        \hline \\[-0.8em]
        {Values} $v \Coloneqq 0 \mid 1 \mid \ldots \mid n$ &
        {Fields} $f \Coloneqq f_1 \mid \ldots \mid f_k$ &
        {Packets} $\pk \Coloneqq \{f_1 = v_1, \ldots, f_k = v_k\}$
    \end{tabular}
    \renewcommand{\arraystretch}{1.0}
    \caption{NetKAT syntax and semantics.}
    \label{fig:synsem}
\end{figure}

To a first approximation, \NetKAT can be thought of as a simple, imperative language that operates over packets, where a \emph{packet}
is a finite record assigning values to fields. The basic primitives in \NetKAT are packet tests ($f \test v$, $f \testNE v$) and packet modifications ($f\mut v$). Program expressions are then compositionally built from tests and packet modifications, using union ($+$), sequencing ($\cdot$), and iteration ($\star$).  Conditionals and loops can be encoded in the standard way:
\[
\textbf{if } b\textbf{ then } p  \textbf{ else } q \triangleq b \cdot p + \neg b \cdot q \qquad\qquad
\textbf{while } b\textbf{ do } p   \triangleq (b \cdot p)^\star \cdot \neg b.
\]
In a network, conditionals can be used to model the behavior of the forwarding tables on individual switches while iteration can be used to model the iterated processing performed by the network as a whole;
the original paper on NetKAT provides further details~\cite{Anderson2014}. Note that assignments and tests in \NetKAT are always against constant values.  This is a key restriction that makes equivalence decidable and aligns with the capabilities of data plane hardware. The $\text{dup}$ primitive makes a copy of the current packet and appends it to the trace, which only ever grows as the packet goes through the network. This primitive is crucial for expressing network-wide properties, as it allows the semantics to
``see'' the intermediate packets processed on internal switches.

\NetKAT's formal semantics, given in \Cref{fig:synsem}, defines the action of a program on an input packet $\alpha$, producing a set of output traces.  The semantics for $\top$ gives a single trace containing the single input packet $\alpha$, whereas $\bot$ produces no traces.  Tests ($f \test v$ and $f \testNE v$) produce a singleton trace or no traces, depending on whether the test succeeds or fails. Modifications ($f \mut v$) produce a singleton trace with the modified packet.  Duplication ($\text{dup}$) produces a single trace with two copies of the input packet.  Union ($p + q$) produces the union of the traces produced by $p$ and $q$.  Sequential composition ($p \cdot q$)
produces the concatenation of the traces produced by $p$ and $q$, where the last packet in output traces of $p$ is used as the input to $q$. Finally, iteration ($p^\star$) produces the union of the traces produced by $p$ iterated zero or more times.
Consider the following example:
\begin{align*}
    (x \test 0 \cdot x \mut 1 \cdot \text{dup}\ +\ x \test 1 \cdot x \mut 0 \cdot \text{dup})^\star
\end{align*}
This program repeatedly flips the value of $x$ between $0$ and $1$, and traces the packet at each step. On input packet $x\test 0$, it produces the following traces:
\begin{align*}
    \{ [x\test0], [x\test 0, x\test 1], [x\test 0, x\test 1, x\test 0], [x\test 0, x\test 1, x\test 0, x\test 1], \ldots \}
\end{align*}
Consider what happens if we change the program to $(x \test 0 \cdot x \mut 1 \cdot \text{dup}\ +\ x \mut 0 \cdot \dup)^\star$.
Unlike the previous example, where the tests were disjoint, this program can generate multiple outputs for a given input packet, as the right branch of the union can always be taken, leading to sequences $x\test 0, x\test 0, \ldots$. On the other hand, if the $x \mut 1$ assignment is performed, then the left branch cannot be taken the next time, because the test $x \test 0$ will fail. Hence, the program produces traces with sequences of $x \test 0$ packets, with singular $x \test 1$ packets between them.

As these simple examples show, despite the fact that assignments and tests in NetKAT programs are always against constant values, the behavior of NetKAT programs can be complicated, particularly when conditionals, iteration, and multiple packet fields are involved.

% \paragraph{Negation and the test fragment}
Traditional presentations of NetKAT distinguish a syntactic test fragment that includes negation:
\begin{align*}
    t,r &\Coloneqq \bot \mid \top \mid f \test v \mid t \lor r \mid t \land r \mid \neg t \tag{Test fragment}\\
    p,q &\Coloneqq t \mid f \mut\, v \mid p + q \mid p \cdot q \mid \text{dup} \mid p^\star \tag{General expressions}
\end{align*}
This distinction is necessary because negation $\neg p$ only makes sense on the test fragment.
We replace general negation $\neg t$ in our presentation with only negated field tests $f \testNE v$. There is no loss of generality---we translate a test fragment expression $t$ to $[t]_0$ using the DeMorgan rules:

\noindent\begin{minipage}{0.25\textwidth}
    \begin{align*}
        [\bot]_0 &= \bot\\
        [\top]_0 &= \top\\
        [f \test v]_0 &= (f \test v)
    \end{align*}
\end{minipage}\begin{minipage}{0.25\textwidth}
    \begin{align*}
        [t \lor r]_0 &= [t]_0 + [r]_0\\
        [t \land r]_0 &= [t]_0 \cdot [r]_0\\
        [\neg t]_0 &= [t]_1
    \end{align*}
\end{minipage}\begin{minipage}{0.25\textwidth}
    \begin{align*}
        [\bot]_1 &= \top\\
        [\top]_1 &= \bot\\
        [f \test v]_1 &= (f \testNE v)
    \end{align*}
\end{minipage}\begin{minipage}{0.25\textwidth}
    \begin{align*}
        [t \lor r]_1 &= [t]_1 \cdot [r]_1\\
        [t \land r]_1 &= [t]_1 + [r]_1\\
        [\neg t]_1 &= [t]_0
    \end{align*}
\end{minipage}

\paragraph{Conditional dup, equivalence, and subsumption}
An interesting feature of NetKAT is that dup need not appear the same number of times on all paths through the expression.
Consider the program $(f \mut\, v + f \mut\, w + \mathsf{dup})^\star$, which either sets field $f$ to $v$ or to $w$, or logs the packet to the trace, and then repeats.
Any number of writes can happen between two dups, and only the effect of the last write before the dup gets logged to the trace.
This is allowed and fully supported by \KATch{}.

The use of dup controls the type of equivalence that is checked.
If we insert a dup after every packet field write, then we obtain trace equivalence.
If we insert no dups whatsoever, we obtain input-output equivalence.
By inserting dups in some places but not others, we can control the equivalence that is checked in a fine-grained manner.

In addition to $p \equiv q$, we may want to check inclusion $p \sqsubseteq q$.
This can be expressed as $p + q \equiv q$ or equivalently, in \KATch{}, as $p \setminus q \equiv \bot$.
As before, the number of dups controls the type of inclusion that is checked: from trace inclusion to input-output inclusion, or something in between.

\subsection{\NetKAT Automata}\label{sec:automata}

\NetKAT's semantics induce an equivalence $(p \equiv q) \triangleq (\Sem{p} = \Sem{q})$ on syntactic programs.  This equivalence is decidable, and can be computed by converting programs to automata and checking for automata equivalence. To convert \NetKAT programs to automata, the standard approach is to use Antimirov derivatives.  We do not give the details here, but the interested reader can find them in \cite{Foster2015}, and for our symbolic automata, in \Cref{sec:symaut}.  We continue with a brief overview of \NetKAT automata and how to check their equivalence.  A \NetKAT automaton $(S, s_0, \left<\epsilon, \delta\right>)$ consists of a set of states $S$, a start state $s_0$, and a pair of functions that depend on an input packet:
\[
\epsilon : S \times \Pk \to 2^\Pk \qquad\qquad  \delta : S\times\Pk \to S^\Pk
\]
Intuitively, the observation function $\epsilon$ is analogous to the notion of a final state, and models the output packets produced from an input packet at a given state. The transition function $\delta$
models the state transitions that can occur when processing an input packet at a given state. Because transitions can modify the input packet, the transition structure $\delta(s,\pk)$ is itself a function of the modified packet, and tells us to transition to state $\delta(s,\pk)(\pk')$ when input packet $\pk$ is modified to $\pk'$.

In terms of the semantics, a transition in the automaton corresponds to executing to the next dup in a NetKAT program, and appends the current packet to the trace.  It should also be noted that states in the automaton do not necessarily correspond to network devices;
because the location of the packet is modeled as just a field in the packet, there can be states in the automaton that handle packets from multiple devices, and the packets of a device can be handled by multiple states.

\paragraph{Carry-on packets}
Unlike traditional regular expressions and automata, NetKAT is stateful, and the output packet of a transition is carried on to the next state.  This makes NetKAT fundamentally different from traditional regular expressions, resulting in challenges in the semantics (which is not a straightforward trace or language semantics) and in designing equivalence procedures, as they have to account for the carry-on packet. Technically it would be possible to fold the packet into the state of the automaton and approach the semantics and the problem of checking equivalence using classical automata. However, because the number of possible packets is exponential in the number of fields and values, such conversion results in enormous state blowup and is therefore impractical.

\begin{figure}
    \begin{algorithmic}
    \State \textbf{Input:} A pair of NetKAT automata $(S, s_0, \delta, \epsilon)$, $(S', s'_0, \delta', \epsilon').$
    \State \textbf{Returns:} A Boolean indicating whether the automata are equivalent.
    \smallskip
    \State $W \gets \{ (s_0, s'_0, \pk) \mid \pk \in \Pk \}$;
    \While{$W$ changes}
        \For {$(s,s',\pk) \in W$}
            \If {$\epsilon(s, \pk) \neq \epsilon'(s', \pk)$}
                \Return \textbf{false};
            \EndIf
            \State $W \gets W \cup \{ (\delta(s, \pk)(\pk'), \delta'(s', \pk)(\pk'), \pk') \mid \pk' \in \Pk \}$
        \EndFor
    \EndWhile
    \State
    \Return \textbf{true};
    \end{algorithmic}
    \caption{Bisimulation algorithm à la Hopcroft \& Karp \cite{Hopcroft1971}}
    \label{fig:bisim}
\end{figure}

\subsection{Bisimulation of NetKAT Automata}\label{sec:bisimulation}

\NetKAT automata can be used to decide $p \equiv q$.  Intuitively, while traces may be infinite, transitions only depend on the current packet, which has a finite number of distinct possible values.

Given two automata $(S, s_0, \left<\epsilon, \delta\right>)$ and $(S', s'_0, \left<\epsilon', \delta'\right>)$, we check their equivalence by considering all possible input packets separately. We run the two automata in parallel, starting from an input packet $\pk$ at their respective start states. We first check that the immediate outputs $\epsilon(s_0, \pk)$ and $\epsilon'(s'_0, \pk)$ are the same. If so, we check that the transitions $\delta(s_0, \pk)$ and $\delta'(s'_0, \pk)$ are equivalent, by recursively checking the equivalence of the states $\delta(s_0, \pk)(\pk')$ and $\delta'(s'_0, \pk)(\pk')$ for all $\pk'$.

We can implement this strategy using a work list algorithm. The work list contains triples $(s, s', \pk)$, where $s$ and $s'$ are states in the two automata, and $\pk$ is an input packet. Initially, the work list contains $(s_0, s'_0, \pk)$ for all packets $\pk$. We then repeatedly take triples $(s, s', \pk)$ from the work list, and check that $\epsilon(s, \pk) = \epsilon'(s', \pk)$ and add $(\delta(s, \pk)(\pk'), \delta'(s', \pk)(\pk'), \pk')$ to the work list for all $\pk'$. If at any point we find that $\epsilon(s, \pk) \neq \epsilon'(t, \pk)$, then the automata are not equivalent. If the work list stabilizes, the automata are equivalent. This algorithm is shown in \Cref{fig:bisim}.

Of course, the issue with this algorithm is that the space of packets is huge. The rest of this paper develops techniques for working with symbolic representations of packets and automata, allowing us to represent and manipulate large sets of packets and to efficiently check equivalence of automata without explicitly enumerating the space of packets.

\section{Symbolic \NetKAT Representations}\label{sec:symbolic}

In this section, we develop symbolic techniques to trace an entire set of packets at once, vastly reducing the number of iterations required to compute a bisimulation in practice. These techniques will allow us to revise the approach of \Cref{fig:bisim} with a much more efficient version in \Cref{sec:fwdbwd}.

First, we introduce a representation for symbolic packets, which represent sets of packets, or equivalently, the fragment of \NetKAT where all atoms are tests. This representation is essentially a natural n-ary variant of binary decision diagrams (BDDs) \cite{Bryant1992}. Second, we introduce Symbolic Packet Programs (\SPPn{}s), a new representation for symbolic transitions in NetKAT automata, representing the $\text{dup}$-free fragment of \NetKAT---i.e., the fragment in which all atoms are tests or assignments. The primary challenge is to design this representation so that it is efficiently closed under the NetKAT operations. Furthermore, we need to be able to efficiently compute the transition of a symbolic packet over a \SPPn{}, both forward and backward.

\begin{figure}
  \vspace{-0.5cm}

  \centering
  % \hspace{-1cm}
  \begin{subfigure}[b]{0.29\textwidth}
      \centering
      % \vspace*{-1cm}
      \begin{tikzpicture}[>=latex',line join=bevel, scale=0.55]
\tikzstyle{every node}=[inner sep=0pt, minimum size=12pt, font=\scriptsize]
\node (n0) at (61.234bp,213.1bp) [draw,fill=skyblue,circle] {a};
  \definecolor{fillcolor}{rgb}{0.76,0.88,0.76};
  \node (n1) at (41.234bp,141.25bp) [draw,fill=fillcolor,circle] {b};
  \node (n2) at (95.234bp,10.8bp) [draw,rectangle] {$\bot$};
  \node (n3) at (22.234bp,10.8bp) [draw,rectangle] {$\top$};
  \node (n4) at (76.234bp,82.65bp) [draw,fill=peachpuff,circle] {c};
  \definecolor{fillcolor}{rgb}{0.76,0.88,0.76};
  \node (n5) at (86.234bp,141.25bp) [draw,fill=fillcolor,circle] {b};
  \definecolor{strokecolor}{rgb}{0.02,0.02,0.02};
  \draw [strokecolor,->] (n1) ..controls (35.902bp,111.11bp) and (31.316bp,67.892bp)  .. (50.234bp,39.6bp) .. controls (56.913bp,29.612bp) and (68.52bp,22.571bp)  .. (n2);
  \definecolor{strokecol}{rgb}{0.0,0.0,0.0};
  \pgfsetstrokecolor{strokecol}
  \draw (43.859bp,82.65bp) node {5};
  \definecolor{strokecolor}{rgb}{0.02,0.02,0.02};
  \draw [strokecolor,->,dashed] (n5) ..controls (82.967bp,121.76bp) and (80.838bp,109.71bp)  .. (n4);
  \definecolor{strokecolor}{rgb}{0.02,0.02,0.02};
  \draw [strokecolor,->,dashed] (n0) ..controls (68.852bp,190.82bp) and (75.824bp,171.33bp)  .. (n5);
  \definecolor{strokecolor}{rgb}{0.02,0.02,0.02};
  \draw [strokecolor,->] (n0) ..controls (55.988bp,197.06bp) and (53.762bp,190.35bp)  .. (51.984bp,184.3bp) .. controls (49.462bp,175.71bp) and (46.966bp,166.05bp)  .. (n1);
  \draw (57.859bp,177.18bp) node {3};
  \definecolor{strokecolor}{rgb}{0.02,0.02,0.02};
  \draw [strokecolor,->] (n1) ..controls (23.878bp,125.38bp) and (8.9767bp,110.25bp)  .. (2.9839bp,93.45bp) .. controls (-4.9467bp,71.215bp) and (4.8833bp,44.402bp)  .. (n3);
  \draw (8.8589bp,82.65bp) node {4};
  \definecolor{strokecolor}{rgb}{0.02,0.02,0.02};
  \draw [strokecolor,->] (n4) ..controls (60.969bp,61.905bp) and (44.831bp,41.03bp)  .. (n3);
  \draw (59.859bp,46.725bp) node {5};
  \definecolor{strokecolor}{rgb}{0.02,0.02,0.02};
  \draw [strokecolor,->,dashed] (n4) ..controls (81.96bp,60.598bp) and (87.11bp,41.666bp)  .. (n2);
  \definecolor{strokecolor}{rgb}{0.02,0.02,0.02};
  \draw [strokecolor,->] (n5) ..controls (91.839bp,125.33bp) and (93.99bp,118.62bp)  .. (95.234bp,112.45bp) .. controls (101.65bp,80.637bp) and (103.12bp,71.926bp)  .. (100.23bp,39.6bp) .. controls (99.907bp,35.937bp) and (99.343bp,32.034bp)  .. (n2);
  \draw (106.86bp,82.65bp) node {5};
  \definecolor{strokecolor}{rgb}{0.02,0.02,0.02};
  \draw [strokecolor,->,dashed] (n1) ..controls (52.283bp,122.38bp) and (61.027bp,108.24bp)  .. (n4);
\end{tikzpicture}
      \vspace*{-0.1cm}
      \caption*{$p$}
      \label{fig:sp1}
  \end{subfigure}
  \hfill % add some horizontal spacing
  \begin{subfigure}[b]{0.29\textwidth}
      \centering
      % \vspace*{-0.6cm}
      \begin{tikzpicture}[>=latex',line join=bevel, scale=0.55]
\tikzstyle{every node}=[inner sep=0pt, minimum size=12pt, font=\scriptsize]
\definecolor{fillcolor}{rgb}{0.76,0.88,0.76};
  \node (n10) at (27.8bp,154.5bp) [draw,fill=fillcolor,circle] {b};
  \node (n11) at (27.8bp,82.65bp) [draw,fill=peachpuff,circle] {c};
  \node (n8) at (10.8bp,10.8bp) [draw,rectangle] {$\bot$};
  \node (n9) at (76.8bp,10.8bp) [draw,rectangle] {$\top$};
  \node (n6) at (51.8bp,213.1bp) [draw,fill=skyblue,circle] {a};
  \node (n7) at (77.8bp,82.65bp) [draw,fill=peachpuff,circle] {c};
  \definecolor{strokecolor}{rgb}{0.02,0.02,0.02};
  \draw [strokecolor,->,dashed] (n10) ..controls (18.814bp,132.8bp) and (10.825bp,112.11bp)  .. (7.8bp,93.45bp) .. controls (4.1521bp,70.953bp) and (6.2649bp,44.532bp)  .. (n8);
  \definecolor{strokecolor}{rgb}{0.02,0.02,0.02};
  \draw [strokecolor,->,dashed] (n6) ..controls (44.176bp,194.12bp) and (38.669bp,181.13bp)  .. (n10);
  \definecolor{strokecolor}{rgb}{0.02,0.02,0.02};
  \draw [strokecolor,->] (n7) ..controls (70.972bp,62.912bp) and (64.671bp,49.128bp)  .. (55.8bp,39.6bp) .. controls (47.895bp,31.11bp) and (36.834bp,24.229bp)  .. (n8);
  \definecolor{strokecol}{rgb}{0.0,0.0,0.0};
  \pgfsetstrokecolor{strokecol}
  \draw (71.425bp,46.725bp) node {5};
  \definecolor{strokecolor}{rgb}{0.02,0.02,0.02};
  \draw [strokecolor,->,dashed] (n7) ..controls (85.439bp,67.39bp) and (88.439bp,60.44bp)  .. (89.8bp,53.85bp) .. controls (91.081bp,47.648bp) and (91.174bp,45.782bp)  .. (89.8bp,39.6bp) .. controls (88.905bp,35.572bp) and (87.362bp,31.438bp)  .. (n9);
  \definecolor{strokecolor}{rgb}{0.02,0.02,0.02};
  \draw [strokecolor,->] (n10) ..controls (27.8bp,132.18bp) and (27.8bp,113.61bp)  .. (n11);
  \draw (33.425bp,118.57bp) node {3};
  \definecolor{strokecolor}{rgb}{0.02,0.02,0.02};
  \draw [strokecolor,->,dashed] (n11) ..controls (22.587bp,60.231bp) and (18.008bp,41.415bp)  .. (n8);
  \definecolor{strokecolor}{rgb}{0.02,0.02,0.02};
  \draw [strokecolor,->] (n6) ..controls (58.321bp,179.88bp) and (69.054bp,126.86bp)  .. (n7);
  \draw (71.425bp,154.5bp) node {5};
  \definecolor{strokecolor}{rgb}{0.02,0.02,0.02};
  \draw [strokecolor,->] (n11) ..controls (32.891bp,62.67bp) and (37.526bp,49.455bp)  .. (44.55bp,39.6bp) .. controls (48.952bp,33.424bp) and (55.005bp,27.734bp)  .. (n9);
  \draw (50.425bp,46.725bp) node {4};
\end{tikzpicture}
      \vspace*{-0.1cm}
      \caption*{$q$}
      \label{fig:sp2}
  \end{subfigure}
  \hfill % add some horizontal spacing
  \begin{subfigure}[b]{0.40\textwidth}
      \centering
      \begin{tikzpicture}[>=latex',line join=bevel, scale=0.55]
\tikzstyle{every node}=[inner sep=0pt, minimum size=12pt, font=\scriptsize]
\node (n12) at (172.74bp,226.35bp) [draw,fill=skyblue,circle] {a};
  \definecolor{fillcolor}{rgb}{0.76,0.88,0.76};
  \node (n13) at (22.74bp,154.5bp) [draw,fill=fillcolor,circle] {b};
  \node (n16) at (81.74bp,10.8bp) [draw,rectangle] {$\top$};
  \definecolor{fillcolor}{rgb}{0.76,0.88,0.76};
  \node (n17) at (172.74bp,154.5bp) [draw,fill=fillcolor,circle] {b};
  \node (n14) at (22.74bp,82.65bp) [draw,fill=peachpuff,circle] {c};
  \node (n15) at (233.74bp,10.8bp) [draw,rectangle] {$\bot$};
  \node (n18) at (159.74bp,82.65bp) [draw,fill=peachpuff,circle] {c};
  \node (n19) at (247.74bp,82.65bp) [draw,fill=peachpuff,circle] {c};
  \definecolor{fillcolor}{rgb}{0.76,0.88,0.76};
  \node (n20) at (256.74bp,154.5bp) [draw,fill=fillcolor,circle] {b};
  \definecolor{strokecolor}{rgb}{0.02,0.02,0.02};
  \draw [strokecolor,->] (n12) ..controls (138.22bp,209.28bp) and (69.188bp,177.13bp)  .. (n13);
  \definecolor{strokecol}{rgb}{0.0,0.0,0.0};
  \pgfsetstrokecolor{strokecol}
  \draw (117.36bp,190.43bp) node {5};
  \definecolor{strokecolor}{rgb}{0.02,0.02,0.02};
  \draw [strokecolor,->,dashed] (n14) ..controls (25.722bp,62.569bp) and (29.292bp,49.013bp)  .. (36.74bp,39.6bp) .. controls (44.032bp,30.385bp) and (55.278bp,23.436bp)  .. (n16);
  \definecolor{strokecolor}{rgb}{0.02,0.02,0.02};
  \draw [strokecolor,->,dashed] (n12) ..controls (195.38bp,206.52bp) and (226.29bp,180.82bp)  .. (n20);
  \definecolor{strokecolor}{rgb}{0.02,0.02,0.02};
  \draw [strokecolor,->] (n14) ..controls (42.185bp,65.728bp) and (65.101bp,48.592bp)  .. (87.49bp,39.6bp) .. controls (131.9bp,21.765bp) and (188.57bp,15.186bp)  .. (n15);
  \draw (93.365bp,46.725bp) node {5};
  \definecolor{strokecolor}{rgb}{0.02,0.02,0.02};
  \draw [strokecolor,->] (n18) ..controls (140.9bp,71.951bp) and (126.33bp,63.505bp)  .. (115.49bp,53.85bp) .. controls (106.58bp,45.909bp) and (98.167bp,35.445bp)  .. (n16);
  \draw (121.36bp,46.725bp) node {5};
  \definecolor{strokecolor}{rgb}{0.02,0.02,0.02};
  \draw [strokecolor,->] (n18) ..controls (148.19bp,63.796bp) and (137.95bp,49.619bp)  .. (126.74bp,39.6bp) .. controls (118.2bp,31.969bp) and (107.31bp,25.177bp)  .. (n16);
  \draw (146.36bp,46.725bp) node {4};
  \definecolor{strokecolor}{rgb}{0.02,0.02,0.02};
  \draw [strokecolor,->,dashed] (n17) ..controls (193.64bp,134.03bp) and (220.45bp,109.07bp)  .. (n19);
  \definecolor{strokecolor}{rgb}{0.02,0.02,0.02};
  \draw [strokecolor,->,dashed] (n18) ..controls (175.37bp,70.163bp) and (186.09bp,61.865bp)  .. (194.74bp,53.85bp) .. controls (203.93bp,45.333bp) and (213.52bp,35.025bp)  .. (n15);
  \definecolor{strokecolor}{rgb}{0.02,0.02,0.02};
  \draw [strokecolor,->] (n17) ..controls (186.94bp,120.52bp) and (213.98bp,57.695bp)  .. (n15);
  \draw (212.36bp,82.65bp) node {5};
  \definecolor{strokecolor}{rgb}{0.02,0.02,0.02};
  \draw [strokecolor,->] (n19) ..controls (221.55bp,69.012bp) and (186.23bp,52.275bp)  .. (155.74bp,39.6bp) .. controls (136.45bp,31.578bp) and (113.93bp,23.265bp)  .. (n16);
  \draw (192.36bp,46.725bp) node {5};
  \definecolor{strokecolor}{rgb}{0.02,0.02,0.02};
  \draw [strokecolor,->] (n17) ..controls (140.41bp,136.39bp) and (83.158bp,101.86bp)  .. (61.74bp,53.85bp) .. controls (57.649bp,44.679bp) and (62.237bp,34.508bp)  .. (n16);
  \draw (94.365bp,82.65bp) node {4};
  \definecolor{strokecolor}{rgb}{0.02,0.02,0.02};
  \draw [strokecolor,->,dashed] (n13) ..controls (9.0829bp,129.77bp) and (-5.6238bp,97.982bp)  .. (2.74bp,71.85bp) .. controls (8.1065bp,55.083bp) and (12.344bp,51.023bp)  .. (25.74bp,39.6bp) .. controls (37.275bp,29.765bp) and (52.89bp,22.387bp)  .. (n16);
  \definecolor{strokecolor}{rgb}{0.02,0.02,0.02};
  \draw [strokecolor,->] (n20) ..controls (265.01bp,139.31bp) and (268.26bp,132.36bp)  .. (269.74bp,125.7bp) .. controls (271.75bp,116.66bp) and (270.35bp,89.469bp)  .. (267.74bp,71.85bp) .. controls (265.58bp,57.241bp) and (266.79bp,52.577bp)  .. (259.74bp,39.6bp) .. controls (257.06bp,34.675bp) and (253.31bp,29.962bp)  .. (n15);
  \draw (275.36bp,82.65bp) node {5};
  \definecolor{strokecolor}{rgb}{0.02,0.02,0.02};
  \draw [strokecolor,->] (n20) ..controls (231.26bp,135.15bp) and (193.74bp,108.13bp)  .. (n18);
  \draw (222.36bp,118.57bp) node {3};
  \definecolor{strokecolor}{rgb}{0.02,0.02,0.02};
  \draw [strokecolor,->] (n13) ..controls (22.74bp,132.18bp) and (22.74bp,113.61bp)  .. (n14);
  \draw (28.365bp,118.57bp) node {5};
  \definecolor{strokecolor}{rgb}{0.02,0.02,0.02};
  \draw [strokecolor,->] (n17) ..controls (165.0bp,139.24bp) and (161.94bp,132.29bp)  .. (160.49bp,125.7bp) .. controls (158.65bp,117.34bp) and (158.34bp,107.83bp)  .. (n18);
  \draw (166.36bp,118.57bp) node {3};
  \definecolor{strokecolor}{rgb}{0.02,0.02,0.02};
  \draw [strokecolor,->,dashed] (n19) ..controls (247.79bp,62.901bp) and (247.18bp,50.304bp)  .. (244.74bp,39.6bp) .. controls (243.85bp,35.71bp) and (242.51bp,31.649bp)  .. (n15);
  \definecolor{strokecolor}{rgb}{0.02,0.02,0.02};
  \draw [strokecolor,->] (n12) ..controls (172.74bp,204.03bp) and (172.74bp,185.46bp)  .. (n17);
  \draw (178.36bp,190.43bp) node {3};
  \definecolor{strokecolor}{rgb}{0.02,0.02,0.02};
  \draw [strokecolor,->,dashed] (n20) ..controls (257.46bp,134.85bp) and (257.4bp,122.28bp)  .. (255.74bp,111.45bp) .. controls (255.14bp,107.51bp) and (254.14bp,103.33bp)  .. (n19);
\end{tikzpicture}
      \vspace*{-0.5cm}
      \caption*{\hspace{0.5cm} $p + q$}
      \label{fig:sp1cup2}
  \end{subfigure}
  \caption{Examples of $\SP$s, where $p \triangleq (a\test3 \cdot b\test4 \mathrel{+} b \testNE 5 \cdot c\test5)$ and $q \triangleq (b=3 \cdot c\test4 \mathrel{+} a\test5 \cdot c \testNE 5)$.}
  \label{fig:sympk}
\end{figure}

\subsection{Symbolic Packets}\label{sec:sympk}

We begin by choosing a representation for symbolic packets. A symbolic packet $p \subseteq \Pk$ is a set of concrete packets, represented compactly as a decision diagram. Syntactically, symbolic packets are NetKAT expressions with atoms restricted to tests, represented in the following canonical form:
%
% \begin{align*}
%     p \in \SP \Coloneqq \bot \mid \top \mid
%     \underbrace{(f = v_0) \cdot p_0 + \ldots + (f = v_n) \cdot p_n + (f \neq v_0 \cdots f \neq v_n) \cdot p_{n+1}}_{\SP(f, \ \{ v_0 \,\mapsto\, p_0,\ \ldots,\ v_n\, \mapsto p_n \},\ p_{n+1})}
% \end{align*}
\begin{align*}
  p \in \SP \Coloneqq \bot \mid \top \mid
  \underbrace{\SP(f, \ \{ \ldots, v_i \mapsto q_i, \ldots \},\ q)}_{\equiv \ \ \sum_i f \test v_i \cdot q_i\ +\ (\prod_i f \testNE v_i)\cdot q}
\end{align*}
That is, symbolic packets form an n-ary tree, where each child $q_i$ is labeled with a test of the current field $f$, and the default case $q$ is labeled with the negation of the other tests. Hence, a given concrete packet has a unique path through the tree.

The following conditions need to be satisfied for a symbolic packet to be in canonical form, inspired by the analogous properties of BDDs:

\begin{description}
    \item[Reduced] If a child $q_i$ is equal to the default case $q$, it is removed. If only the default case remains, the symbolic packet is reduced to the default case itself.
    \item[Ordered] A path down the tree always follows the same order of fields,
    and the children $q_i$ are ordered by the value $v_i$.
\end{description}

These conditions ensure the representation of a symbolic packet is unique, in the sense that two symbolic packets are semantically equal if and only if they are syntactically equal.

\paragraph{Representation}
As with BDDs, we share nodes in the tree, making the representation of a symbolic packet a directed acyclic graph, as shown in \Cref{fig:sympk}. Vertices are labeled with the packet field they test. Solid arrows labeled with a number encode the test value, and dashed arrows represent a default case. The sinks of the graph are labeled with $\top$ or $\bot$, indicating membership in the set.  On the left, we have a symbolic packet representing all packets where $a=3$ and $b=4$, or $b \neq 5$ and $c = 5$.  In the middle, we have a symbolic packet representing all packets where $b=3$ and $c=4$, or $a=5$ and $c \neq 5$.  On the right, we have the union of these symbolic packets.

\begin{figure}[t]
\[
  \begin{array}{ll}
  \textbf{Primitive tests:}\\
  \quad (f \test v) \triangleq \SP(f, \{v \mapsto \top\}, \bot) \qquad (f \testNE v) \triangleq \SP(f, \{v \mapsto \bot\}, \top)\\
  \textbf{Operations, base cases:}\\
  \quad \top \plusop \top \triangleq \top \quad \top \cdotop \top \triangleq \top \quad \top \capop \top \triangleq \top \quad \top \oplusop \top \triangleq \bot \quad \top \minusop \top \triangleq \bot\\
  \quad \top \plusop \bot \triangleq \top \quad \top \cdotop \bot \triangleq \bot \quad \top \capop \bot \triangleq \bot \quad \top \oplusop \bot \triangleq \top \quad \top \minusop \bot \triangleq \top\\
  \quad \bot \plusop \top \triangleq \top \quad \bot \cdotop \top \triangleq \bot \quad \bot \capop \top \triangleq \bot \quad \bot \oplusop \top \triangleq \top \quad \bot \minusop \top \triangleq \bot\\
  \quad \bot \plusop \bot \triangleq \bot \quad \bot \cdotop \bot \triangleq \bot \quad \bot \capop \bot \triangleq \bot \quad \bot \oplusop \bot \triangleq \bot \quad \bot \minusop \bot \triangleq \bot\\
  \textbf{Operations, inductive case:}\\
  \quad \SP(f, b_p, d_p)\pmop \SP(f, b_q, d_q) \triangleq  \spsc(f,b'\!,d_p \pmop  d_q)\\
  \quad\quad \text{ where }\  b' = \{v \mapsto b_p(v;d_p) \pmop b_q(v;d_q) \mid v\in \keys(b_p \cup b_q)\}\\
  \textbf{Expansion:}\\
  \quad p \equiv \SP(f, \emptyset, p) \text {\qquad if \qquad} p \in \{\top, \bot, \SP(f', b, d)\} \text{ where } f \sqsubset f'\\
  \textbf{Smart constructor:}\\
  \quad \spsc(f,b,d) \triangleq
    \text{ if } b' = \emptyset \text{ then } d \text{ else } \SP(f, b', d)\\
  \qquad \text{where } b' \triangleq \{v \mapsto p \mid v \mapsto p \in b, p \neq d\}
  % \quad \spsc(f,b,d) \triangleq \begin{cases}
  %   \SP(f, b', d) &\text{ if } b' \triangleq \{v \mapsto p \mid v \mapsto p \in b, p \neq d\} \neq \emptyset\\
  %   d &\text{ if } b' = \emptyset
  % \end{cases}
  \end{array}
\]
    \caption{Definition of the $\SP$ operations. The inductive case is identical for all operations (indicated by $\pmop$), and applies when both $\SP$s test the same field. Expansion inserts a trivial $\SP$ node to reduce the remaining cases to the inductive case. The notation $b(v;d)$ means the child $b(v)$ if $v \in \keys(b)$, or the default case $d$ otherwise.
    % The formal definitions are in \Cref{app:sp_ops}.
    }
    \label{fig:spops}
\end{figure}

\paragraph{Operations}
Symbolic packets are closed under the NetKAT operations:
\begin{align*}
    \nf{+}, \nf{\cdot}, \nf{\cap}, \nf{\oplus}, \nf{-}\ :\ \SP \times \SP \to \SP \qquad
    \nf{\star}\ :\ \SP \to \SP \qquad
    f \test v, f \testNE v\ :\ \SP
\end{align*}
As defined in \Cref{fig:spops}, operations on symbolic packets are computed by traversing the two trees in parallel, recursively taking the operation of the children, making sure to maintain canonical form using the smart constructor. Repetition $p^\star \triangleq \top$ is trivial for symbolic packets.
All other operations are defined in \appendixref{app:sp_ops}.
\paragraph{Specifications}
Each of these operations is uniquely specified by two correctness conditions:
\begin{enumerate}
    \item They semantically match their counterpart: $\Sem{p \mathop{\nf{+}} q} = \Sem{p + q}$ for all $p,q \in \SP$.
    \item They maintain canonical form: $p \mathop{\nf{+}} q$ is reduced and ordered if $p,q$ are.
\end{enumerate}
For further formal treatment of SPs, see \appendixref{app:sp_ops}.

\begin{figure}[!t]
    \vspace*{0.7cm}
    \begin{subfigure}[b]{0.2\textwidth}
        \centering
        % \vspace*{-.6cm}
        % \hspace{-1cm}
        \begin{tikzpicture}[>=latex',line join=bevel, scale=0.41]
\tikzstyle{every node}=[inner sep=0pt, minimum size=10pt, font=\tiny]
\tikzstyle{diam}=[diamond,minimum size=5pt]
\begin{scope}
  \pgfsetstrokecolor{black}
  \definecolor{strokecol}{rgb}{1.0,1.0,1.0};
  \pgfsetstrokecolor{strokecol}
  \definecolor{fillcol}{rgb}{1.0,1.0,1.0};
  \pgfsetfillcolor{fillcol}
  \filldraw (0.0bp,0.0bp) -- (0.0bp,420.3bp) -- (128.65bp,420.3bp) -- (128.65bp,0.0bp) -- cycle;
\end{scope}
\begin{scope}
  \pgfsetstrokecolor{black}
  \definecolor{strokecol}{rgb}{1.0,1.0,1.0};
  \pgfsetstrokecolor{strokecol}
  \definecolor{fillcol}{rgb}{1.0,1.0,1.0};
  \pgfsetfillcolor{fillcol}
  \filldraw (0.0bp,0.0bp) -- (0.0bp,420.3bp) -- (128.65bp,420.3bp) -- (128.65bp,0.0bp) -- cycle;
\end{scope}
\begin{scope}
  \pgfsetstrokecolor{black}
  \definecolor{strokecol}{rgb}{1.0,1.0,1.0};
  \pgfsetstrokecolor{strokecol}
  \definecolor{fillcol}{rgb}{1.0,1.0,1.0};
  \pgfsetfillcolor{fillcol}
  \filldraw (0.0bp,0.0bp) -- (0.0bp,420.3bp) -- (128.65bp,420.3bp) -- (128.65bp,0.0bp) -- cycle;
\end{scope}
\begin{scope}
  \pgfsetstrokecolor{black}
  \definecolor{strokecol}{rgb}{1.0,1.0,1.0};
  \pgfsetstrokecolor{strokecol}
  \definecolor{fillcol}{rgb}{1.0,1.0,1.0};
  \pgfsetfillcolor{fillcol}
  \filldraw (0.0bp,0.0bp) -- (0.0bp,420.3bp) -- (128.65bp,420.3bp) -- (128.65bp,0.0bp) -- cycle;
\end{scope}
\begin{scope}
  \pgfsetstrokecolor{black}
  \definecolor{strokecol}{rgb}{1.0,1.0,1.0};
  \pgfsetstrokecolor{strokecol}
  \definecolor{fillcol}{rgb}{1.0,1.0,1.0};
  \pgfsetfillcolor{fillcol}
  \filldraw (0.0bp,0.0bp) -- (0.0bp,420.3bp) -- (128.65bp,420.3bp) -- (128.65bp,0.0bp) -- cycle;
\end{scope}
\begin{scope}
  \pgfsetstrokecolor{black}
  \definecolor{strokecol}{rgb}{1.0,1.0,1.0};
  \pgfsetstrokecolor{strokecol}
  \definecolor{fillcol}{rgb}{1.0,1.0,1.0};
  \pgfsetfillcolor{fillcol}
  \filldraw (0.0bp,0.0bp) -- (0.0bp,420.3bp) -- (128.65bp,420.3bp) -- (128.65bp,0.0bp) -- cycle;
\end{scope}
  \node (n26) at (89.4bp,143.7bp) [draw,fill=peachpuff,circle] {c};
  \definecolor{fillcolor}{rgb}{0.76,0.88,0.76};
  \node (n23) at (95.4bp,276.6bp) [draw,fill=fillcolor,circle] {b};
  \node (n29) at (33.4bp,10.8bp) [draw,rectangle] {$\bot$};
  \node (n28) at (74.4bp,77.25bp) [draw,fill=peachpuff,diam] {};
  \definecolor{fillcolor}{rgb}{0.76,0.88,0.76};
  \node (n31) at (50.4bp,276.6bp) [draw,fill=fillcolor,circle] {b};
  \node (n30) at (53.4bp,343.05bp) [draw,fill=skyblue,diam] {};
  \node (n22) at (88.4bp,343.05bp) [draw,fill=skyblue,diam] {};
  \definecolor{fillcolor}{rgb}{0.76,0.88,0.76};
  \node (n32) at (50.4bp,210.15bp) [draw,fill=fillcolor,diam] {};
  \definecolor{fillcolor}{rgb}{0.76,0.88,0.76};
  \node (n33) at (5.4bp,210.15bp) [draw,fill=fillcolor,diam] {};
  \node (n27) at (103.4bp,77.25bp) [draw,fill=peachpuff,diam] {};
  \node (n21) at (69.4bp,409.5bp) [draw,fill=skyblue,circle] {a};
  \node (n25) at (103.4bp,10.8bp) [draw,rectangle] {$\top$};
  \definecolor{fillcolor}{rgb}{0.76,0.88,0.76};
  \node (n24) at (100.4bp,210.15bp) [draw,fill=fillcolor,diam] {};
  \definecolor{strokecolor}{rgb}{0.02,0.02,0.02};
  \draw [strokecolor,->,dashed] (n21) ..controls (63.988bp,386.7bp) and (58.958bp,366.44bp)  .. (n30);
  \definecolor{strokecolor}{rgb}{0.02,0.02,0.02};
  \draw [strokecolor,->] (n26) ..controls (94.21bp,120.55bp) and (98.581bp,100.44bp)  .. (n27);
  \definecolor{strokecol}{rgb}{0.0,0.0,0.0};
  \pgfsetstrokecolor{strokecol}
  \draw (104.03bp,107.78bp) node {5};
  \definecolor{strokecolor}{rgb}{0.02,0.02,0.02};
  \draw [strokecolor,->,dashed] (n24) ..controls (98.186bp,196.18bp) and (94.692bp,175.71bp)  .. (n26);
  \definecolor{strokecolor}{rgb}{0.02,0.02,0.02};
  \draw [strokecolor,->] (n32) ..controls (47.985bp,187.16bp) and (42.515bp,121.29bp)  .. (60.4bp,71.85bp) .. controls (66.659bp,54.547bp) and (79.412bp,37.774bp)  .. (n25);
  \draw (58.025bp,107.78bp) node {1};
  \definecolor{strokecolor}{rgb}{0.02,0.02,0.02};
  \draw [strokecolor,->] (n27) ..controls (103.4bp,62.114bp) and (103.4bp,42.679bp)  .. (n25);
  \draw (109.03bp,46.725bp) node {5};
  \definecolor{strokecolor}{rgb}{0.02,0.02,0.02};
  \draw [strokecolor,->,dashed] (n26) ..controls (84.327bp,120.9bp) and (79.61bp,100.64bp)  .. (n28);
  \definecolor{strokecolor}{rgb}{0.02,0.02,0.02};
  \draw [strokecolor,->] (n31) ..controls (50.4bp,253.81bp) and (50.4bp,234.51bp)  .. (n32);
  \draw (56.025bp,240.68bp) node {2};
  \definecolor{strokecolor}{rgb}{0.02,0.02,0.02};
  \draw [strokecolor,->] (n22) ..controls (89.809bp,329.08bp) and (92.032bp,308.61bp)  .. (n23);
  \draw (97.025bp,312.53bp) node {5};
  \definecolor{strokecolor}{rgb}{0.02,0.02,0.02};
  \draw [strokecolor,->,dashed] (n28) ..controls (67.237bp,64.99bp) and (52.973bp,42.568bp)  .. (n29);
  \definecolor{strokecolor}{rgb}{0.02,0.02,0.02};
  \draw [strokecolor,->,dashed] (n30) ..controls (52.765bp,328.4bp) and (51.822bp,308.14bp)  .. (n31);
  \definecolor{strokecolor}{rgb}{0.02,0.02,0.02};
  \draw [strokecolor,->,dashed] (n31) ..controls (35.834bp,254.74bp) and (19.691bp,231.62bp)  .. (n33);
  \definecolor{strokecolor}{rgb}{0.02,0.02,0.02};
  \draw [strokecolor,->] (n21) ..controls (75.861bp,386.58bp) and (81.918bp,366.04bp)  .. (n22);
  \draw (87.025bp,373.57bp) node {5};
  \definecolor{strokecolor}{rgb}{0.02,0.02,0.02};
  \draw [strokecolor,->] (n24) ..controls (103.8bp,201.42bp) and (106.89bp,193.69bp)  .. (108.4bp,186.75bp) .. controls (121.02bp,128.88bp) and (117.37bp,113.07bp)  .. (118.4bp,53.85bp) .. controls (118.51bp,47.518bp) and (119.95bp,45.74bp)  .. (118.4bp,39.6bp) .. controls (117.36bp,35.483bp) and (115.58bp,31.312bp)  .. (n25);
  \draw (123.03bp,107.78bp) node {1};
  \definecolor{strokecolor}{rgb}{0.02,0.02,0.02};
  \draw [strokecolor,->,dashed] (n33) ..controls (7.7951bp,185.07bp) and (16.108bp,104.82bp)  .. (27.4bp,39.6bp) .. controls (28.032bp,35.951bp) and (28.808bp,32.054bp)  .. (n29);
  \definecolor{strokecolor}{rgb}{0.02,0.02,0.02};
  \draw [strokecolor,->] (n32) ..controls (57.376bp,197.62bp) and (71.685bp,173.97bp)  .. (n26);
  \draw (77.025bp,179.62bp) node {2};
  \definecolor{strokecolor}{rgb}{0.02,0.02,0.02};
  \draw [strokecolor,->,dashed] (n23) ..controls (97.1bp,253.69bp) and (98.618bp,234.12bp)  .. (n24);
\end{tikzpicture}
        \vspace*{-0.1cm}
        \caption*{$p$}
        \label{fig:spp1}
    \end{subfigure}
    \hspace{-0.8cm}
    \begin{subfigure}[b]{0.25\textwidth}
        \centering
        % \vspace*{-.6cm}
        \begin{tikzpicture}[>=latex',line join=bevel,scale=0.41]
\tikzstyle{every node}=[inner sep=0pt, minimum size=10pt, font=\tiny]
\tikzstyle{diam}=[diamond,minimum size=5pt]
\begin{scope}
  \pgfsetstrokecolor{black}
  \definecolor{strokecol}{rgb}{1.0,1.0,1.0};
  \pgfsetstrokecolor{strokecol}
  \definecolor{fillcol}{rgb}{1.0,1.0,1.0};
  \pgfsetfillcolor{fillcol}
  \filldraw (0.0bp,0.0bp) -- (0.0bp,420.3bp) -- (190.2bp,420.3bp) -- (190.2bp,0.0bp) -- cycle;
\end{scope}
\begin{scope}
  \pgfsetstrokecolor{black}
  \definecolor{strokecol}{rgb}{1.0,1.0,1.0};
  \pgfsetstrokecolor{strokecol}
  \definecolor{fillcol}{rgb}{1.0,1.0,1.0};
  \pgfsetfillcolor{fillcol}
  \filldraw (0.0bp,0.0bp) -- (0.0bp,420.3bp) -- (190.2bp,420.3bp) -- (190.2bp,0.0bp) -- cycle;
\end{scope}
\begin{scope}
  \pgfsetstrokecolor{black}
  \definecolor{strokecol}{rgb}{1.0,1.0,1.0};
  \pgfsetstrokecolor{strokecol}
  \definecolor{fillcol}{rgb}{1.0,1.0,1.0};
  \pgfsetfillcolor{fillcol}
  \filldraw (0.0bp,0.0bp) -- (0.0bp,420.3bp) -- (190.2bp,420.3bp) -- (190.2bp,0.0bp) -- cycle;
\end{scope}
\begin{scope}
  \pgfsetstrokecolor{black}
  \definecolor{strokecol}{rgb}{1.0,1.0,1.0};
  \pgfsetstrokecolor{strokecol}
  \definecolor{fillcol}{rgb}{1.0,1.0,1.0};
  \pgfsetfillcolor{fillcol}
  \filldraw (0.0bp,0.0bp) -- (0.0bp,420.3bp) -- (190.2bp,420.3bp) -- (190.2bp,0.0bp) -- cycle;
\end{scope}
\begin{scope}
  \pgfsetstrokecolor{black}
  \definecolor{strokecol}{rgb}{1.0,1.0,1.0};
  \pgfsetstrokecolor{strokecol}
  \definecolor{fillcol}{rgb}{1.0,1.0,1.0};
  \pgfsetfillcolor{fillcol}
  \filldraw (0.0bp,0.0bp) -- (0.0bp,420.3bp) -- (190.2bp,420.3bp) -- (190.2bp,0.0bp) -- cycle;
\end{scope}
\begin{scope}
  \pgfsetstrokecolor{black}
  \definecolor{strokecol}{rgb}{1.0,1.0,1.0};
  \pgfsetstrokecolor{strokecol}
  \definecolor{fillcol}{rgb}{1.0,1.0,1.0};
  \pgfsetfillcolor{fillcol}
  \filldraw (0.0bp,0.0bp) -- (0.0bp,420.3bp) -- (190.2bp,420.3bp) -- (190.2bp,0.0bp) -- cycle;
\end{scope}
  \definecolor{fillcolor}{rgb}{0.76,0.88,0.76};
  \node (n49) at (38.4bp,210.15bp) [draw,fill=fillcolor,diam] {};
  \definecolor{fillcolor}{rgb}{0.76,0.88,0.76};
  \node (n48) at (104.4bp,276.6bp) [draw,fill=fillcolor,circle] {b};
  \definecolor{fillcolor}{rgb}{0.76,0.88,0.76};
  \node (n41) at (67.4bp,210.15bp) [draw,fill=fillcolor,diam] {};
  \node (n40) at (64.4bp,10.8bp) [draw,rectangle] {$\top$};
  \node (n43) at (103.4bp,77.25bp) [draw,fill=peachpuff,diam] {};
  \node (n42) at (103.4bp,143.7bp) [draw,fill=peachpuff,circle] {c};
  \node (n45) at (104.4bp,343.05bp) [draw,fill=skyblue,diam] {};
  \node (n44) at (179.4bp,10.8bp) [draw,rectangle] {$\bot$};
  \definecolor{fillcolor}{rgb}{0.76,0.88,0.76};
  \node (n47) at (144.4bp,210.15bp) [draw,fill=fillcolor,diam] {};
  \definecolor{fillcolor}{rgb}{0.76,0.88,0.76};
  \node (n46) at (144.4bp,276.6bp) [draw,fill=fillcolor,circle] {b};
  \node (n38) at (19.4bp,143.7bp) [draw,fill=peachpuff,circle] {c};
  \node (n39) at (19.4bp,77.25bp) [draw,fill=peachpuff,diam] {};
  \definecolor{fillcolor}{rgb}{0.76,0.88,0.76};
  \node (n50) at (112.4bp,210.15bp) [draw,fill=fillcolor,diam] {};
  \node (n34) at (86.4bp,409.5bp) [draw,fill=skyblue,circle] {a};
  \node (n35) at (72.4bp,343.05bp) [draw,fill=skyblue,diam] {};
  \definecolor{fillcolor}{rgb}{0.76,0.88,0.76};
  \node (n36) at (64.4bp,276.6bp) [draw,fill=fillcolor,circle] {b};
  \definecolor{fillcolor}{rgb}{0.76,0.88,0.76};
  \node (n37) at (5.4bp,210.15bp) [draw,fill=fillcolor,diam] {};
  \definecolor{strokecolor}{rgb}{0.02,0.02,0.02};
  \draw [strokecolor,->,dashed] (n45) ..controls (104.4bp,327.91bp) and (104.4bp,308.48bp)  .. (n48);
  \definecolor{strokecolor}{rgb}{0.02,0.02,0.02};
  \draw [strokecolor,->] (n34) ..controls (79.392bp,393.92bp) and (76.644bp,387.09bp)  .. (75.15bp,380.7bp) .. controls (73.189bp,372.31bp) and (72.531bp,362.58bp)  .. (n35);
  \definecolor{strokecol}{rgb}{0.0,0.0,0.0};
  \pgfsetstrokecolor{strokecol}
  \draw (80.025bp,373.57bp) node {1};
  \definecolor{strokecolor}{rgb}{0.02,0.02,0.02};
  \draw [strokecolor,->] (n37) ..controls (3.734bp,197.84bp) and (2.403bp,183.88bp)  .. (5.15bp,172.5bp) .. controls (6.3048bp,167.72bp) and (8.4316bp,162.83bp)  .. (n38);
  \draw (11.025bp,179.62bp) node {1};
  \definecolor{strokecolor}{rgb}{0.02,0.02,0.02};
  \draw [strokecolor,->,dashed] (n47) ..controls (149.02bp,202.31bp) and (153.97bp,194.31bp)  .. (156.4bp,186.75bp) .. controls (174.46bp,130.63bp) and (178.36bp,60.018bp)  .. (n44);
  \definecolor{strokecolor}{rgb}{0.02,0.02,0.02};
  \draw [strokecolor,->] (n41) ..controls (66.955bp,179.87bp) and (65.293bp,70.561bp)  .. (n40);
  \draw (71.025bp,107.78bp) node {1};
  \definecolor{strokecolor}{rgb}{0.02,0.02,0.02};
  \draw [strokecolor,->] (n36) ..controls (40.417bp,270.54bp) and (18.626bp,263.33bp)  .. (8.15bp,247.8bp) .. controls (2.9055bp,240.02bp) and (2.7773bp,229.09bp)  .. (n37);
  \draw (14.025bp,240.68bp) node {1};
  \definecolor{strokecolor}{rgb}{0.02,0.02,0.02};
  \draw [strokecolor,->] (n47) ..controls (145.59bp,182.65bp) and (146.64bp,92.562bp)  .. (104.4bp,39.6bp) .. controls (98.059bp,31.65bp) and (88.849bp,25.035bp)  .. (n40);
  \draw (142.03bp,107.78bp) node {1};
  \definecolor{strokecolor}{rgb}{0.02,0.02,0.02};
  \draw [strokecolor,->,dashed] (n46) ..controls (144.4bp,253.81bp) and (144.4bp,234.51bp)  .. (n47);
  \definecolor{strokecolor}{rgb}{0.02,0.02,0.02};
  \draw [strokecolor,->,dashed] (n38) ..controls (19.4bp,120.91bp) and (19.4bp,101.61bp)  .. (n39);
  \definecolor{strokecolor}{rgb}{0.02,0.02,0.02};
  \draw [strokecolor,->,dashed] (n42) ..controls (103.4bp,120.91bp) and (103.4bp,101.61bp)  .. (n43);
  \definecolor{strokecolor}{rgb}{0.02,0.02,0.02};
  \draw [strokecolor,->] (n39) ..controls (14.228bp,66.171bp) and (7.7513bp,50.453bp)  .. (14.15bp,39.6bp) .. controls (21.242bp,27.571bp) and (35.616bp,20.383bp)  .. (n40);
  \draw (20.025bp,46.725bp) node {4};
  \definecolor{strokecolor}{rgb}{0.02,0.02,0.02};
  \draw [strokecolor,->,dashed] (n39) ..controls (26.835bp,65.601bp) and (42.879bp,42.623bp)  .. (n40);
  \definecolor{strokecolor}{rgb}{0.02,0.02,0.02};
  \draw [strokecolor,->,dashed] (n43) ..controls (114.81bp,66.575bp) and (145.05bp,40.933bp)  .. (n44);
  \definecolor{strokecolor}{rgb}{0.02,0.02,0.02};
  \draw [strokecolor,->] (n43) ..controls (96.586bp,64.99bp) and (83.019bp,42.568bp)  .. (n40);
  \draw (95.025bp,46.725bp) node {4};
  \definecolor{strokecolor}{rgb}{0.02,0.02,0.02};
  \draw [strokecolor,->,dashed] (n50) ..controls (110.59bp,196.18bp) and (107.73bp,175.71bp)  .. (n42);
  \definecolor{strokecolor}{rgb}{0.02,0.02,0.02};
  \draw [strokecolor,->] (n35) ..controls (70.79bp,329.08bp) and (68.249bp,308.61bp)  .. (n36);
  \draw (75.025bp,312.53bp) node {1};
  \definecolor{strokecolor}{rgb}{0.02,0.02,0.02};
  \draw [strokecolor,->,dashed] (n41) ..controls (73.798bp,197.69bp) and (86.817bp,174.39bp)  .. (n42);
  \definecolor{strokecolor}{rgb}{0.02,0.02,0.02};
  \draw [strokecolor,->,dashed] (n36) ..controls (65.42bp,253.69bp) and (66.331bp,234.12bp)  .. (n41);
  \definecolor{strokecolor}{rgb}{0.02,0.02,0.02};
  \draw [strokecolor,->] (n49) ..controls (34.724bp,196.68bp) and (28.431bp,175.33bp)  .. (n38);
  \draw (38.025bp,179.62bp) node {1};
  \definecolor{strokecolor}{rgb}{0.02,0.02,0.02};
  \draw [strokecolor,->] (n48) ..controls (84.281bp,255.95bp) and (57.976bp,230.27bp)  .. (n49);
  \draw (80.025bp,240.68bp) node {1};
  \definecolor{strokecolor}{rgb}{0.02,0.02,0.02};
  \draw [strokecolor,->,dashed] (n34) ..controls (92.521bp,386.58bp) and (98.259bp,366.04bp)  .. (n45);
  \definecolor{strokecolor}{rgb}{0.02,0.02,0.02};
  \draw [strokecolor,->] (n45) ..controls (108.89bp,334.89bp) and (114.04bp,326.66bp)  .. (118.4bp,319.65bp) .. controls (124.22bp,310.29bp) and (130.74bp,299.75bp)  .. (n46);
  \draw (132.03bp,312.53bp) node {1};
  \definecolor{strokecolor}{rgb}{0.02,0.02,0.02};
  \draw [strokecolor,->,dashed] (n48) ..controls (107.13bp,253.57bp) and (109.6bp,233.73bp)  .. (n50);
\end{tikzpicture}
        \vspace*{-0.1cm}
        \caption*{$q$}
        \label{fig:spp2}
    \end{subfigure}
    \hspace{-0.6cm}
    \begin{subfigure}[b]{0.24\textwidth}
        \centering
        % \vspace*{-.6cm}
        % \hspace{-2.1cm}
        \input{viz/spp1seq2.tex}
        \vspace*{-0.55cm}
        \caption*{\hspace{2cm} $p \cdot q$}
        \label{fig:spp1seq2}
    \end{subfigure}
    \hspace{0.5cm}
    \begin{subfigure}[b]{0.3\textwidth}
      \centering
    %  \vspace*{-.6cm}
      \begin{tikzpicture}[>=latex',line join=bevel,scale=0.41]
\tikzstyle{every node}=[inner sep=0pt, minimum size=10pt, font=\tiny]
\tikzstyle{diam}=[diamond,minimum size=5pt]
\begin{scope}
  \pgfsetstrokecolor{black}
  \definecolor{strokecol}{rgb}{1.0,1.0,1.0};
  \pgfsetstrokecolor{strokecol}
  \definecolor{fillcol}{rgb}{1.0,1.0,1.0};
  \pgfsetfillcolor{fillcol}
  \filldraw (0.0bp,0.0bp) -- (0.0bp,407.05bp) -- (269.2bp,407.05bp) -- (269.2bp,0.0bp) -- cycle;
\end{scope}
\begin{scope}
  \pgfsetstrokecolor{black}
  \definecolor{strokecol}{rgb}{1.0,1.0,1.0};
  \pgfsetstrokecolor{strokecol}
  \definecolor{fillcol}{rgb}{1.0,1.0,1.0};
  \pgfsetfillcolor{fillcol}
  \filldraw (0.0bp,0.0bp) -- (0.0bp,407.05bp) -- (269.2bp,407.05bp) -- (269.2bp,0.0bp) -- cycle;
\end{scope}
\begin{scope}
  \pgfsetstrokecolor{black}
  \definecolor{strokecol}{rgb}{1.0,1.0,1.0};
  \pgfsetstrokecolor{strokecol}
  \definecolor{fillcol}{rgb}{1.0,1.0,1.0};
  \pgfsetfillcolor{fillcol}
  \filldraw (0.0bp,0.0bp) -- (0.0bp,407.05bp) -- (269.2bp,407.05bp) -- (269.2bp,0.0bp) -- cycle;
\end{scope}
\begin{scope}
  \pgfsetstrokecolor{black}
  \definecolor{strokecol}{rgb}{1.0,1.0,1.0};
  \pgfsetstrokecolor{strokecol}
  \definecolor{fillcol}{rgb}{1.0,1.0,1.0};
  \pgfsetfillcolor{fillcol}
  \filldraw (0.0bp,0.0bp) -- (0.0bp,407.05bp) -- (269.2bp,407.05bp) -- (269.2bp,0.0bp) -- cycle;
\end{scope}
\begin{scope}
  \pgfsetstrokecolor{black}
  \definecolor{strokecol}{rgb}{1.0,1.0,1.0};
  \pgfsetstrokecolor{strokecol}
  \definecolor{fillcol}{rgb}{1.0,1.0,1.0};
  \pgfsetfillcolor{fillcol}
  \filldraw (0.0bp,0.0bp) -- (0.0bp,407.05bp) -- (269.2bp,407.05bp) -- (269.2bp,0.0bp) -- cycle;
\end{scope}
\begin{scope}
  \pgfsetstrokecolor{black}
  \definecolor{strokecol}{rgb}{1.0,1.0,1.0};
  \pgfsetstrokecolor{strokecol}
  \definecolor{fillcol}{rgb}{1.0,1.0,1.0};
  \pgfsetfillcolor{fillcol}
  \filldraw (0.0bp,0.0bp) -- (0.0bp,407.05bp) -- (269.2bp,407.05bp) -- (269.2bp,0.0bp) -- cycle;
\end{scope}
  \node (n269) at (10.8bp,130.45bp) [draw,rectangle] {$\bot$};
  \node (n265) at (170.8bp,77.25bp) [draw,fill=peachpuff,diam] {};
  \definecolor{fillcolor}{rgb}{0.76,0.88,0.76};
  \node (n268) at (43.8bp,196.9bp) [draw,fill=fillcolor,diam] {};
  \node (n274) at (145.8bp,329.8bp) [draw,fill=skyblue,diam] {};
  \definecolor{fillcolor}{rgb}{0.76,0.88,0.76};
  \node (n271) at (229.8bp,263.35bp) [draw,fill=fillcolor,circle] {b};
  \node (n270) at (221.8bp,329.8bp) [draw,fill=skyblue,diam] {};
  \definecolor{fillcolor}{rgb}{0.76,0.88,0.76};
  \node (n273) at (263.8bp,196.9bp) [draw,fill=fillcolor,diam] {};
  \definecolor{fillcolor}{rgb}{0.76,0.88,0.76};
  \node (n272) at (229.8bp,196.9bp) [draw,fill=fillcolor,diam] {};
  \definecolor{fillcolor}{rgb}{0.76,0.88,0.76};
  \node (n262) at (103.8bp,263.35bp) [draw,fill=fillcolor,circle] {b};
  \definecolor{fillcolor}{rgb}{0.76,0.88,0.76};
  \node (n263) at (102.8bp,196.9bp) [draw,fill=fillcolor,diam] {};
  \definecolor{fillcolor}{rgb}{0.76,0.88,0.76};
  \node (n277) at (196.8bp,196.9bp) [draw,fill=fillcolor,diam] {};
  \definecolor{fillcolor}{rgb}{0.76,0.88,0.76};
  \node (n276) at (165.8bp,196.9bp) [draw,fill=fillcolor,diam] {};
  \node (n266) at (170.8bp,10.8bp) [draw,rectangle] {$\top$};
  \definecolor{fillcolor}{rgb}{0.76,0.88,0.76};
  \node (n275) at (165.8bp,263.35bp) [draw,fill=fillcolor,circle] {b};
  \node (n260) at (145.8bp,396.25bp) [draw,fill=skyblue,circle] {a};
  \definecolor{fillcolor}{rgb}{0.76,0.88,0.76};
  \node (n267) at (63.8bp,263.35bp) [draw,fill=fillcolor,circle] {b};
  \node (n261) at (103.8bp,329.8bp) [draw,fill=skyblue,diam] {};
  \node (n264) at (170.8bp,130.45bp) [draw,fill=peachpuff,circle] {c};
  \definecolor{strokecolor}{rgb}{0.02,0.02,0.02};
  \draw [strokecolor,->,dashed] (n264) ..controls (170.8bp,110.96bp) and (170.8bp,98.45bp)  .. (n265);
  \definecolor{strokecolor}{rgb}{0.02,0.02,0.02};
  \draw [strokecolor,->] (n270) ..controls (223.41bp,315.83bp) and (225.95bp,295.36bp)  .. (n271);
  \definecolor{strokecol}{rgb}{0.0,0.0,0.0};
  \pgfsetstrokecolor{strokecol}
  \draw (231.43bp,299.28bp) node {1};
  \definecolor{strokecolor}{rgb}{0.02,0.02,0.02};
  \draw [strokecolor,->,dashed] (n260) ..controls (145.8bp,373.46bp) and (145.8bp,354.16bp)  .. (n274);
  \definecolor{strokecolor}{rgb}{0.02,0.02,0.02};
  \draw [strokecolor,->] (n260) ..controls (132.05bp,374.15bp) and (117.3bp,351.51bp)  .. (n261);
  \draw (133.43bp,360.32bp) node {5};
  \definecolor{strokecolor}{rgb}{0.02,0.02,0.02};
  \draw [strokecolor,->] (n260) ..controls (168.38bp,376.1bp) and (199.98bp,349.3bp)  .. (n270);
  \draw (201.43bp,360.32bp) node {1};
  \definecolor{strokecolor}{rgb}{0.02,0.02,0.02};
  \draw [strokecolor,->,dashed] (n274) ..controls (149.67bp,316.33bp) and (156.29bp,294.98bp)  .. (n275);
  \definecolor{strokecolor}{rgb}{0.02,0.02,0.02};
  \draw [strokecolor,->] (n273) ..controls (256.12bp,186.92bp) and (241.04bp,170.08bp)  .. (225.8bp,159.25bp) .. controls (213.44bp,150.47bp) and (197.9bp,142.89bp)  .. (n264);
  \draw (247.43bp,166.38bp) node {1};
  \definecolor{strokecolor}{rgb}{0.02,0.02,0.02};
  \draw [strokecolor,->,dashed] (n273) ..controls (264.86bp,183.86bp) and (264.86bp,168.38bp)  .. (256.8bp,159.25bp) .. controls (239.42bp,139.58bp) and (208.04bp,133.74bp)  .. (n264);
  \definecolor{strokecolor}{rgb}{0.02,0.02,0.02};
  \draw [strokecolor,->,dashed] (n262) ..controls (103.46bp,240.56bp) and (103.16bp,221.26bp)  .. (n263);
  \definecolor{strokecolor}{rgb}{0.02,0.02,0.02};
  \draw [strokecolor,->] (n263) ..controls (109.89bp,187.43bp) and (123.69bp,171.33bp)  .. (136.55bp,159.25bp) .. controls (143.09bp,153.11bp) and (150.82bp,146.76bp)  .. (n264);
  \draw (142.43bp,166.38bp) node {1};
  \definecolor{strokecolor}{rgb}{0.02,0.02,0.02};
  \draw [strokecolor,->,dashed] (n263) ..controls (105.63bp,184.78bp) and (110.59bp,169.08bp)  .. (119.8bp,159.25bp) .. controls (129.37bp,149.04bp) and (143.55bp,141.73bp)  .. (n264);
  \definecolor{strokecolor}{rgb}{0.02,0.02,0.02};
  \draw [strokecolor,->,dashed] (n271) ..controls (237.53bp,248.14bp) and (241.42bp,240.95bp)  .. (244.8bp,234.55bp) .. controls (249.83bp,225.04bp) and (255.49bp,214.09bp)  .. (n273);
  \definecolor{strokecolor}{rgb}{0.02,0.02,0.02};
  \draw [strokecolor,->] (n261) ..controls (98.781bp,321.72bp) and (93.083bp,313.55bp)  .. (88.55bp,306.4bp) .. controls (82.729bp,297.21bp) and (76.522bp,286.67bp)  .. (n267);
  \draw (94.425bp,299.28bp) node {1};
  \definecolor{strokecolor}{rgb}{0.02,0.02,0.02};
  \draw [strokecolor,->] (n274) ..controls (139.97bp,319.05bp) and (129.29bp,302.51bp)  .. (116.8bp,292.15bp) .. controls (113.58bp,289.48bp) and (93.448bp,279.17bp)  .. (n267);
  \draw (136.43bp,299.28bp) node {1};
  \definecolor{strokecolor}{rgb}{0.02,0.02,0.02};
  \draw [strokecolor,->,dashed] (n277) ..controls (193.14bp,185.92bp) and (187.73bp,171.36bp)  .. (182.8bp,159.25bp) .. controls (181.16bp,155.22bp) and (179.33bp,150.9bp)  .. (n264);
  \definecolor{strokecolor}{rgb}{0.02,0.02,0.02};
  \draw [strokecolor,->,dashed] (n275) ..controls (173.74bp,248.25bp) and (177.63bp,241.06bp)  .. (180.8bp,234.55bp) .. controls (185.27bp,225.37bp) and (189.88bp,214.65bp)  .. (n277);
  \definecolor{strokecolor}{rgb}{0.02,0.02,0.02};
  \draw [strokecolor,->] (n271) ..controls (229.8bp,240.56bp) and (229.8bp,221.26bp)  .. (n272);
  \draw (235.43bp,227.43bp) node {2};
  \definecolor{strokecolor}{rgb}{0.02,0.02,0.02};
  \draw [strokecolor,->] (n275) ..controls (165.8bp,240.56bp) and (165.8bp,221.26bp)  .. (n276);
  \draw (171.43bp,227.43bp) node {2};
  \definecolor{strokecolor}{rgb}{0.02,0.02,0.02};
  \draw [strokecolor,->] (n276) ..controls (166.86bp,182.25bp) and (168.43bp,161.99bp)  .. (n264);
  \draw (174.43bp,166.38bp) node {1};
  \definecolor{strokecolor}{rgb}{0.02,0.02,0.02};
  \draw [strokecolor,->] (n276) ..controls (160.38bp,189.23bp) and (154.7bp,181.34bp)  .. (152.55bp,173.5bp) .. controls (150.88bp,167.39bp) and (150.74bp,165.32bp)  .. (152.55bp,159.25bp) .. controls (154.13bp,153.96bp) and (157.08bp,148.74bp)  .. (n264);
  \draw (158.43bp,166.38bp) node {2};
  \definecolor{strokecolor}{rgb}{0.02,0.02,0.02};
  \draw [strokecolor,->,dashed] (n267) ..controls (56.999bp,240.43bp) and (50.623bp,219.89bp)  .. (n268);
  \definecolor{strokecolor}{rgb}{0.02,0.02,0.02};
  \draw [strokecolor,->] (n261) ..controls (103.8bp,314.66bp) and (103.8bp,295.23bp)  .. (n262);
  \draw (109.43bp,299.28bp) node {5};
  \definecolor{strokecolor}{rgb}{0.02,0.02,0.02};
  \draw [strokecolor,->] (n272) ..controls (222.99bp,186.9bp) and (210.96bp,171.32bp)  .. (199.8bp,159.25bp) .. controls (194.58bp,153.6bp) and (188.48bp,147.67bp)  .. (n264);
  \draw (216.43bp,166.38bp) node {2};
  \definecolor{strokecolor}{rgb}{0.02,0.02,0.02};
  \draw [strokecolor,->,dashed] (n268) ..controls (37.203bp,189.67bp) and (29.4bp,181.68bp)  .. (24.8bp,173.5bp) .. controls (20.231bp,165.38bp) and (16.851bp,155.53bp)  .. (n269);
  \definecolor{strokecolor}{rgb}{0.02,0.02,0.02};
  \draw [strokecolor,->] (n265) ..controls (165.78bp,69.514bp) and (160.53bp,61.58bp)  .. (158.55bp,53.85bp) .. controls (156.3bp,45.044bp) and (158.82bp,35.294bp)  .. (n266);
  \draw (164.43bp,46.725bp) node {4};
  \definecolor{strokecolor}{rgb}{0.02,0.02,0.02};
  \draw [strokecolor,->,dashed] (n265) ..controls (170.8bp,62.114bp) and (170.8bp,42.679bp)  .. (n266);
  \definecolor{strokecolor}{rgb}{0.02,0.02,0.02};
  \draw [strokecolor,->] (n268) ..controls (53.79bp,187.45bp) and (75.292bp,169.93bp)  .. (95.55bp,159.25bp) .. controls (114.96bp,149.02bp) and (138.97bp,140.86bp)  .. (n264);
  \draw (101.43bp,166.38bp) node {1};
\end{tikzpicture}
      \vspace*{-0.05cm}
      \caption*{\hspace{1cm} $(p \oplus q)^\star$}
      \label{fig:spp1xor2star}
    \end{subfigure}
    \caption{Examples of \SPPn{}s, where $p \triangleq (a \test 5 + b \test 2) \cdot (b \mut 1 + c \test 5)$, and $q \triangleq (b \test 1 \mathrel{+} c \mut 4 \mathrel{+} a \mut 1 \cdot b \mut 1)$.}
    \label{fig:symt}
\end{figure}

\subsection{Symbolic Transitions}\label{sec:symtrans}

We now introduce a representation for symbolic transitions in NetKAT automata.  Whereas symbolic packets represent the fragment of NetKAT where all atoms are tests, symbolic transitions correspond to the larger dup-free fragment, where all atoms are tests or assignments. This introduces additional challenges for a canonical representation, as well as for the operations. For instance, sequential composition is no longer equivalent to intersection, and the star operator is no longer trivial.
\begin{align*}
    p \in \SPP \Coloneqq\ &\bot \mid \top \mid
      \underbrace{
        \SPP(f, \{ \ldots, v_i \mapsto \{ \ldots, w_{ij} \mapsto q_{ij}, \ldots \}, \ldots \}, \{ \ldots, w_i \mapsto q_i, \ldots \}, q)
      }_{
        \equiv\ \ \sum_i f \test v_i \cdot \sum_j f \mut w_{ij} \cdot q_{ij}\ +\ \
        (\prod_i f \testNE v_i) \cdot (\sum_i f \mut w_i\ +\ (\prod_i f \testNE w_i) \cdot q)
      }
    % \\
    % & \Sigma_m (f = v_m) \cdot \Sigma_\ldots + f \mut w_{m}^{k_m} \cdot p_{m}^{k_m})\ + \ldots ) + \\
    % & (f \neq v_0 \cdots f \neq v_n) \cdot (\ldots + f \mut w_{n+1}^{k_{m}} \cdot p_{n+1}^{k_{m}}\ + \ldots + f \neq w_{n+1}^{0} \cdots f \neq w_{n+1}^{k_{n+1}}  \cdot p_{n+1}^{k_{n+1}+1})
    % & (f = v_0) \cdot (f \mut w_{0}^{0} \cdot p_{0}^{0} + \ldots + f \mut w_{0}^{k_0} \cdot p_{0}^{k_0})\ + \\
    % & \cdots \\
    % & (f = v_n) \cdot (f \mut w_{n}^{0} \cdot p_{n}^{0} + \ldots + f \mut w_{n}^{k_n} \cdot p_{n}^{k_n})\ + \\
    % & (f \neq v_0 \cdots f \neq v_n) \cdot (f \mut w_{n+1}^{0} \cdot p_{n+1}^{0} + \ldots + f \mut w_{n+1}^{k_{n+1}} \cdot p_{n+1}^{k_{n+1}}\ + \\
    % & \phantom{(f \neq v_0 \cdots f \neq v_n) \cdot\ \ \ } f \neq w_{n+1}^{0} \cdots f \neq w_{n+1}^{k_{n+1}}  \cdot p_{n+1}^{k_{n+1}+1})
\end{align*}

Like SPs, \SPPn{}s have two base cases, $\top$ and $\bot$. Also like SPs, \SPPn{}s test a field $f$ of the input packet against a series of values $v_0,\dots,v_n$, with a default case $f \testNE v_0 \cdots f \testNE v_n$. However, instead of continuing recursively after the test, \SPPn{}s non-deterministically assign a value $w_{ij}$ to the field $f$ of the output packet, and continue recursively with the corresponding child $q_{ij}$. This way, \SPPn{}s can output more than one packet for a given input packet, and can also output packets with different values for the same field. The default case is further split into two cases, one where the field $f$ is non-deterministically assigned a new value (like in the other cases), and an identity case where the field $f$ keeps the same value as the input packet. However, the packet for the latter case is not produced when the input packet's field $f$ had a value that could also be produced by the explicit assignments in the default case. This ensures that for a given input packet, a given output packet is produced by a unique path through the \SPPn{}.

The following conditions need to be satisfied for a \SPPn{} to be in canonical form:

\begin{description}
    \item[Reduced] If a child $q_{ij}$ is equal to $\bot$, it is removed (with the exception of the default-identity case, which is always kept).  If one of the default-assignment cases for a value $w$ is $\bot$, it is removed, but to keep the behavior equivalent, an additional test for the same value $w$ is added with a copy of the default-assignments (if a test for the value $w$ was already present, nothing is added)\footnote{see the formal definition in \Cref{app:sppsc}}. Further, each of the non-default cases is analyzed, and if we determine that it behaves semantically like the default case for that input value, then it is removed. If only the default case remains, the \SPPn{} is reduced to the default case itself.
    \item[Ordered] A path down the tree always follows the same order of fields, and both the tests and the assignments are ordered by the value $v_i$ or $w_{ij}$.
\end{description}

These conditions ensure the representation of a \SPPn{} is unique, in the sense that two \SPPn{}s are semantically equal if and only if they are syntactically equal.

\paragraph{Representation}
Like symbolic packets, we represent \SPPn{}s as a directed acyclic graph, with duplicate nodes shared. Examples of \SPPn{}s are shown in \Cref{fig:symt}. Vertices are labeled with the packet field that they test. Solid arrows labeled with a number represent the tests, and dashed arrows represent the default case. Each test is followed by a non-deterministic assignment of a new value to the field, indicated by the small diamonds. In the default case, the non-deterministic assignment also has an identity case (keeping the value of the field unchanged), indicated by a dashed arrow emanating from the diamond.

Consider the action of the first \SPPn{} in \Cref{fig:symt} on the concrete packet $a \test 5 \cdot b \test 3 \cdot c \test 5$. The first test $a \test 5$ succeeds, and the field $a$ is assigned the value $5$. The $b$ field is then tested, but the $b$ node only has a default case. The default case does have a non-deterministic assignment, which can set the $b$ field to the new value $1$, or it can keep the old value $3$. In case the new value of $b$ is $1$, the packet is immediately accepted by the $\top$ node, so the packet $a \test 5 \cdot b \test 1 \cdot c \test 5$ is produced. In case the old value of $b$ is kept, the $c$ field is tested, and in our case the value of the $c$ field is $5$. In this case, the value of the field is unchanged, and the packet $a \test 5 \cdot b \test 3 \cdot c \test 5$ is produced. In summary, for input packet $a \test 5 \cdot b \test 3 \cdot c \test 5$, the \SPPn{} produces packets $a \test 5 \cdot b \test 1 \cdot c \test 5$ and $a \test 5 \cdot b \test 3 \cdot c \test 5$.

When we sequentially compose the two \SPPn{}s on the left, we get the third \SPPn{} shown.  In other words, if we take a concrete packet, and first apply the first \SPPn{}, and then apply the second \SPPn{} to all of the resulting packets, then we get the same result as if we apply the \SPPn{} on the right directly to the concrete packet. Sequential composition is a relatively complex operation, but algorithmically it is a key strength of our representation and reduced/ordered invariant. To understand why, consider taking a concrete packet, and applying the first \SPPn{} to it.  Once we have applied the $a$-layer of the \SPPn{} to the packet, we already know what the value of the $a$ field in the output packet will be.  We can therefore immediately continue with the $a$ layer of the second \SPPn{}, without having to consider the other layers of the first \SPPn{}.  This is in contrast to a naive algorithm, which would have to consider the entire first \SPPn{} before being able to apply the second \SPPn{}.  This property makes it possible to compute the sequential composition of two \SPPn{}s efficiently in practice, and is a key contributor to the performance and scalability of our system.

\begin{figure}[t]
  \[
    \begin{array}{ll}
    \textbf{Primitive tests and mutation:}\\
    \quad (f \test v) \triangleq \SPP(f, \{v \sto \{ v \sto \top \} \}, \emptyset, \bot) \qquad \text{\smash{\raisebox{-0.4\normalbaselineskip}{$(f \mut\, v) \triangleq \SPP(f, \emptyset, \{v \sto \top \}, \bot)$}}}\\
    \quad (f \testNE v) \triangleq \SPP(f, \{v \sto \{ v \sto \bot \} \}, \emptyset, \top)\\
    \textbf{Operations, base cases:}\\
    \quad \textit{ Identical to those in \Cref{fig:spops}.}\\
    \textbf{Operations, inductive cases:}\\
    \quad \SPP(f, b_p, m_p, d_p)\pmop \SPP(f, b_q, m_q, d_q) \triangleq  \sppsc(f,b'\!,m_p \mathop{\vec{\pm}} m_q,d_p \pmop  d_q)\\
    \whereindent \text{ where }\  b' = \{v \sto \get{p}{v} \mathop{\vec{\pm}} \get{q}{v} \mid v\in \keys(b_p \cup b_q \cup m_p \cup m_q) \}\\
    \quad \SPP(f, b_p, m_p, d_p)\cdotop \SPP(f, b_q, m_q, d_q) \triangleq  \sppsc(f,b'\!,m_A \mathop{\vec{+}} m_B, d_p \cdotop  d_q)\\
    \whereindent \text{ where }\  b' = \{v \sto \unionmaps{v'\! \sto p' \in \get{p}{v}}{ \{ w'\! \sto p' \nf{\cdot} q' \mid w'\! \sto q' \in \get{q}{v'}\} } \mid v\in \keys(b_p \cup b_q \cup m_p \cup m_q \cup m_A)\}\\
    \whereindent \phantom{\text{ where }}\  m_A = \unionmaps{v'\! \sto p' \in m_p}{ \{ w'\! \sto p' \nf{\cdot} q' \mid w'\! \sto q' \in \get{q}{v'}\} },
      \quad m_B = \{ w'\! \sto d_p \nf{\cdot} q' \mid w'\! \sto q' \in m_q \} \\
    \quad m_1 \mathop{\vec{\pm}} m_2 \triangleq \{v \,\sto\, m_1(v;\bot) \pmop m_2(v;\bot) \mid v\in \keys(m_1 \cup m_2)\},
      \quad  \vec{\Sigma} \triangleq \text{$n$-ary sum w.r.t. $\vec{+}$}\\
    % if branches.contains(v) then branches(v)
    % else if other.contains(v) || (id eq False) then other
    % else other + (v -> id)
    \quad \get{p}{v} \triangleq b_p(v; \text{ if } v \in m_p \lor d_p = \bot \text{ then } m_p \text{ else } m_p \cup \{ v \sto d_p \}) \text{\quad (similarly for $q(v)$)}\\
    \textbf{Expansion:}\\
    \quad p \equiv \SPP(f, \emptyset, \emptyset, p) \text {\qquad if \qquad} p \in \{\top, \bot, \SPP(f'\!, b, m, d)\} \text{ where } f \sqsubset f'\\
    \textbf{Repetition:}\\
    \quad p^\star \triangleq \nf{\Sigma}_i p^i = \top \plusop p \plusop p \!\nf{\cdot}\! p \plusop p \!\nf{\cdot}\! p \!\nf{\cdot} \!p \plusop \ldots \text{\quad\ \ (sums in $\SPP$ converge in a finite number of steps)}
    \end{array}
  \]
      \caption{Definition of the $\SPP$ operations. The inductive case is identical for all operations (indicated by $\pmop$), except $\cdotop$, which is given separately. Expansion inserts a trivial $\SPP$ node to reduce the remaining cases to the inductive case. The notation $b(v;d)$ means the child $b(v)$ if $v \in \keys(b)$, or by default $d$ otherwise.
      % The formal definitions are in \Cref{app:sp_ops}.
      }
      \label{fig:sppops}
  \end{figure}

\paragraph{Operations}
\SPPn{}s are closed not just under sequential composition, but under all of the NetKAT operations.  These operations are largely mechanical, but more complex than for SPs, as we need to take assignments into account while respecting the conditions that ensure uniqueness.  In particular, sequential composition is no longer equivalent to intersection, as the assignments in the first \SPPn{} clearly affect the tests in the second \SPPn{}.  Furthermore, the star operator is no longer trivial, as the assignments in the \SPPn{} can affect the tests in the \SPPn{} itself, and a fixed point needs to be taken. We have the following operations on \SPPn{}s:
\begin{align*}
    \nf{+}, \nf{\cdot}, \nf{\cap}, \nf{\oplus}, \nf{-}\ :\ \SPP \times \SPP \to \SPP \qquad
    \nf{\star}\ :\ \SPP \to \SPP \qquad
    f \test v, f \testNE v, f \mut v \ :\ \SPP
\end{align*}

% \renewcommand\thesubfigure{\arabic{subfigure}}
% \begin{figure}
%     \centering
%     \begin{subfigure}[b]{0.6\textwidth}
%         \centering
%          \vspace*{-.6cm}
%         \includegraphics[scale=0.45]{viz/spp1xor2cp.pdf}
%         \caption{Symmetric difference $(i) \oplus (ii)$}
%         \label{fig:spp1xor2}
%     \end{subfigure}
%     \hfill % add some horizontal spacing
%     \begin{subfigure}[b]{0.39\textwidth}
%         \centering
%        \vspace*{-.6cm}
%         \includegraphics[scale=0.45]{viz/spp1xor2star.pdf}
%         \caption{Star $((\mathsf{i}) \oplus (\mathsf{ii}))^\star$}
%         \label{fig:spp1xor2star}
%     \end{subfigure}
%     \caption{Symbolic operations (where (i) and (ii) are from \Cref{fig:symt})}
%     \label{fig:symt2}
% \end{figure}

\paragraph{Push and pull}
In addition to these operations for combining \SPPn{}s, we also have the following operations, which ``push'' and ``pull'' a symbolic packet through a \SPPn:
\begin{align*}
    \mathsf{push} :\ \SP \times \SPP \to \SP \qquad\qquad
    \mathsf{pull} :\ \SPP \times \SP \to \SP
\end{align*}

The push operation computes the effect of a \SPPn{} on a symbolic packet, giving a symbolic packet as a result. The new symbolic packet contains all of the packets that are produced by the \SPPn{} when applied to the packets in the input symbolic packet.

The pull operation simulates the effect of a \SPPn{} in reverse. This operation is used when computing the backward transition of a symbolic packet over a symbolic transition, for counter-example generation. The pull operation answers this question: given a set of output packets (represented symbolically), what are the possible input packets (also represented symbolically) that could have produced them?
In other words, a concrete packet $\pk$ is an element of $\mathsf{pull}(p,q)$ if and only if running the \SPPn{} $p$ on $\pk$ produces an output packet in $q$. In particular, it is okay if running the \SPPn{} on $\pk$ produces multiple output packets, as long as at least one of them is in $q$.

\paragraph{Specifications}
Each of these operations is uniquely specified by two correctness conditions:
\begin{enumerate}
    \item They semantically match their counterpart (e.g. $\Sem{p \mathop{\nf{+}} q} \equiv \Sem{p + q}$ for all $p,q \in \SPP$)
    \item They maintain the reduced \& ordered invariant (e.g., $p \mathop{\nf{+}} q$ is reduced \& ordered if $p,q$ are)
\end{enumerate}

For push and pull, the correctness condition is as follows:
\[
    % \Sem{\mathsf{push}\ p\ s} &= \{\pkp\in\Pk\mid \pk\in\pSem{p} \wedge \pkp\in\Sem{s}(\pk)\}\\
    % \Sem{\mathsf{pull}\ s\ p} &= \{\pkp\in\Pk\mid \pk\in\pSem{p} \cap \Sem{s}(\pkp)\}\\
\beta \in \Sem{\mathsf{push}\ p\ s} \iff \exists \pk\in\pSem{p}.\ \pkp\in\Sem{s}(\pk)
\quad\text{and}\quad
\beta \in \Sem{\mathsf{pull}\ s\ p} \iff \exists \alpha \in \pSem{p}.\ \alpha \in \Sem{s}(\pkp)
\]

For further details, see \appendixref{app:spp}.

\section{Symbolic NetKAT Automata via Brzozowski Derivatives}\label{sec:symaut}

With the symbolic packet and symbolic transition representations in place, we can now define symbolic NetKAT automata.  The construction of NetKAT automata shares some similarities with the construction of automata for regular expressions.  In prior work~\cite{Foster2015,Smolka2015}, NetKAT automata were constructed via Antimirov derivatives \cite{Antimirov1996}, which can be extended from regular expressions to NetKAT.  The Antimirov derivative constructs a non-deterministic automaton, and as such, is well suited for handling \NetKAT's union operator by inserting transitions for the two sub-terms.  However, the Antimirov derivative is not well suited for our extended set of logical operators, as we cannot simply insert (non-deterministic) transitions for the operands of intersection, difference, and symmetric difference.  Instead, we extend the Brzozowski derivative \cite{Brzozowski1962} to NetKAT, which constructs a deterministic automaton directly, and is better suited for the logical operators.

In the rest of this section, we will first describe what symbolic NetKAT automata are, and then describe how to construct them via the Brzozowski derivative.

\subsection{Symbolic NetKAT Automata}

An automaton for a regular expression consists of a set of states, and transitions labeled with symbols from the alphabet.  Additionally, one of the states is designated as the initial state, and a subset of the states are designated as accepting states.  This way, an automaton for a regular expression models the set of strings that are accepted by the regular expression.

An automaton for a NetKAT program is similar, but instead of modeling a set of strings, it models the traces that are produced for a given input packet.  When a packet travels through the automaton, every state that it traverses acts as a $\text{dup}$ operation, appending a copy of the packet to the packet's trace.  Therefore, the transition between two states is labeled with a dup-free NetKAT program, represented symbolically as a \SPPn{}. That is, each edge in the automaton is not labeled with a single letter, as it would be in a standard DFA or NFA, but with a potentially large \SPPn{}. This is necessary because an edge in the automaton intuitively represents ``what happens between two dups''. This is precisely a dup-free NetKAT program, which can itself have rich behavior, represented canonically as an \SPPn{}. This is needed to support fine-grained control over equivalence and inclusion as discussed in \Cref{sec:netkat}---e.g., via conditionally executed dups or assignments between dups.

Secondly, instead of having a boolean at each state to determine acceptance, NetKAT automata have an additional \SPPn{} at each state that determines the set of packets that are produced as output of the automaton when a packet reaches that state.  Therefore, a symbolic NetKAT automaton is a tuple $\mathcal{A} = \langle Q, q_o, \delta, \epsilon \rangle$ where $Q$ is a finite set of states, 
 $q_0 \in Q$ is the initial state,  
\begin{itemize}
 \item $\delta : Q \times Q \to \SPP$ is the transition function, and
 \item $\epsilon : Q \to \SPP$ is the output function.
\end{itemize}

\paragraph{Deterministic NetKAT automata}
The notion of a deterministic NetKAT automaton is more subtle than for classic finite automata.  For classic automata, a deterministic automaton is one where for every state $q$ and symbol $a$, there is at most one transition from $q$ labeled with $a$.  For a NetKAT automaton, we need to take into account the fact that the transitions are labeled with \SPPn{}s, which may produce multiple packets for a given input packet.  Therefore, we define a deterministic NetKAT automaton as one where for every state $q$ and packet $\pk$, the set of packets produced by $\mathsf{push}(\pk,\delta(q,q'))$ are disjoint for all $q' \in Q$.  This is equivalent to saying that $\delta(q,q_1')
\mathop{\nf{\cap}} \delta(q,q_2') = \bot$ for all $q_1',q_2' \in Q$
($q_1' \neq q_2'$).  Note that an input packet $\pk$ may produce multiple different packets at each successor state, but these packet sets at different successor states are disjoint.

\subsection{Constructing Automata via Brzozowski Derivatives}

We now describe how to construct a symbolic NetKAT automaton for a NetKAT program via the Brzozowski derivative.  We take the set of states to be the set of NetKAT expressions, and the initial state to be the NetKAT program itself.  Because we want to construct a deterministic automaton, we need a way to represent a non-intersecting outgoing symbolic transition structure (STS). We represent such a transition structure as a NetKAT expression in the following form:
\begin{align*}
    r \in \mathsf{STS} \Coloneqq p_1 \cdot \mathrel{\text{dup}} \cdot\ q_1 + \ldots + p_n \cdot \mathrel{\text{dup}} \cdot\ q_n
\end{align*}
where $p_i \in \SPP$, and $q_i$ are NetKAT expressions, and the $p_i$ are disjoint ($p_i \mathop{\nf{\cap}} p_j = \bot$ for $i \neq j$).

\noindent
\paragraph{Operations}
Dup is an STS, by taking $r = \top \cdot\text{dup} \cdot \top$.
We extend the logical operators to STSs:
\begin{align*}
    \tilde{+}, \tilde{\cap}, \tilde{\oplus}, \tilde{-}\ :\ \mathsf{STS} \times \mathsf{STS} \to \mathsf{STS}
\end{align*}
These operations need to be defined such that the resulting STS is deterministic. For $r_1 \cap r_2$:
\begin{align*}
    (p_1 \cdot \text{dup} \cdot q_1 + \ldots + p_n \cdot \text{dup}\ \cdot q_n) \mathrel{\tilde{\cap}} (p'_1 \cdot \text{dup} \cdot q'_1 + \ldots + p'_n \cdot \text{dup} \cdot q'_n)
\end{align*}
To bring this in STS form, we distribute the intersection over the union, and combine the terms:
\begin{align*}
    (p_1 \mathop{\nf{\cap}} p'_1) \cdot \text{dup} \cdot (q_1 \cap q'_1) +
    (p_1 \mathop{\nf{\cap}} p'_2) \cdot \text{dup} \cdot (q_1 \cap q'_2) +
    \ldots +
    (p_n \mathop{\nf{\cap}} p'_n) \cdot \text{dup} \cdot (q_n \cap q'_n)
\end{align*}
This is not yet in STS form, as the $q_i \cap q'_i$ terms may not all be different, so we need to collect the terms with the same $q_i \cap q'_i$, and union their \SPPn{}s. The other operators need a similar (albeit slightly more complicated) treatment in order to maintain determinism.

Second, we extend sequential composition to STSs, in two forms (denoted with the same symbol). We can compose an STS with a \SPPn{} on the left, or with a NetKAT expression on the right:
\begin{align*}
    \tilde{\cdot} : \SPP \times \mathsf{STS} \to \mathsf{STS} \qquad
    \tilde{\cdot} : \mathsf{STS} \times \mathsf{Exp} \to \mathsf{STS}
\end{align*}
Like the other operators, these need to be defined such that the resulting STS is deterministic.

All of the STS operations are defined in terms of \SPPn{} operations. Therefore, an efficient \SPPn{} implementation is crucial both for operations on symbolic packets and for operations on STSs.

\paragraph{Specifications}
Each of these operations is uniquely specified by two correctness conditions:
\begin{enumerate}
    \item They semantically match their counterpart (e.g. $\Sem{p \mathop{\nf{+}} q} \equiv \Sem{p + q}$ for all $p,q \in \mathsf{STS}$)
    \item They maintain the invariant that the transitions are pairwise disjoint.
\end{enumerate}

\paragraph{Brzozowski derivative}
With these operations in place, it is straightforward to define the Brzozowski derivative for NetKAT expressions:

\noindent
\begin{minipage}{0.26\textwidth}
    \begin{align*}
        \epsilon(p + q) &\triangleq \epsilon(p) \mathop{\nf{+}} \epsilon(q) \\
        \epsilon(p \cap q) &\triangleq \epsilon(p) \mathop{\nf{\cap}} \epsilon(q) \\
        \epsilon(p \oplus q) &\triangleq \epsilon(p) \mathop{\nf{\oplus}} \epsilon(q) \\
        \epsilon(p - q) &\triangleq \epsilon(p) \mathop{\nf{-}} \epsilon(q) \\
        \epsilon(p \cdot q) &\triangleq \epsilon(p) \mathop{\nf{\cdot}} \epsilon(q) \\
        \epsilon(p^\star) &\triangleq \epsilon(p)^\star
    \end{align*}
\end{minipage}\begin{minipage}{0.23\textwidth}
  \begin{align*}
        \epsilon(\text{dup}) &\triangleq \bot\\
        \epsilon(f \test v) &\triangleq f \test v\\
        \epsilon(f \testNE v) &\triangleq f \testNE v\\
        \epsilon(f \mut v) &\triangleq f \mut v\\
        \epsilon(\top) &\triangleq \top\\
        \epsilon(\bot) &\triangleq \bot
    \end{align*}
\end{minipage}\begin{minipage}{0.26\textwidth}
\begin{align*}
    \delta(p + q) &\triangleq \delta(p) \mathop{\tilde{+}} \delta(q) \\
    \delta(p \cap q) &\triangleq \delta(p) \mathop{\tilde{\cap}} \delta(q) \\
    \delta(p \oplus q) &\triangleq \delta(p) \mathop{\tilde{\oplus}} \delta(q) \\
    \delta(p - q) &\triangleq \delta(p) \mathop{\tilde{-}} \delta(q) \\
    \delta(p \cdot q) &\triangleq
        \delta(p) \mathop{\tilde{\cdot}} q \mathrel{\tilde{+}}
        \epsilon(p) \mathop{\tilde{\cdot}} \delta(q) \\
    \delta(p^\star) &\triangleq \epsilon(p)^\star \mathop{\tilde{\cdot}} \delta(p) \mathrel{\tilde{\cdot}} p^\star
    \end{align*}
\end{minipage}\begin{minipage}{0.23\textwidth}
  \begin{align*}
    \delta(\text{dup}) &\triangleq \dup\\
    \delta(f \test v) &\triangleq \bot\\
    \delta(f \testNE v) &\triangleq \bot\\
    \delta(f \mut v) &\triangleq \bot\\
    \delta(\top) &\triangleq \bot\\
    \delta(\bot) &\triangleq \bot
\end{align*}
\end{minipage}
\medskip

We construct a deterministic symbolic NetKAT automaton for a NetKAT program $p$ as follows:
\begin{description}
    \item[States] $Q \triangleq \mathsf{Exp}$.
    \item[Initial state] $q_0 \triangleq p$.
    \item[Transitions] take $\delta(q,q')$ to be the \SPPn{} of $q'$ in the Brzozowski derivative of $\delta(q)$.
    \item[Output] $\epsilon : Q \to \SPP$, defined above.
\end{description}
The Brzozowski derivative is guaranteed to transitively reach only finitely many essentially different NetKAT expressions from a given start state. Therefore, in the actual implementation, we do not use the infinite set $\mathsf{Exp}$ for the states, but instead use only the finitely many essentially different NetKAT terms reached from the start state.

The NetKAT automaton for an example program is shown in \Cref{fig:exaut}. The edges are labeled with the \SPPn{}s that represent the transitions. The output \SPPn{} of every state is $\top$, i.e., the automaton accepts all packets that reach a state.
Note that the \SPPn{}s in the figure feature several diamonds without outgoing edges. These diamonds produce no output packet, i.e., the packet is dropped.

\begin{figure}
        %%\vspace*{-.6cm}
    \centering
    \begin{subfigure}[b]{0.4\textwidth}
        \centering
        \begin{tikzpicture}[>=latex',line join=bevel,scale=0.37]
\tikzstyle{every node}=[inner sep=0pt, minimum size=10pt, font=\tiny]
\definecolor{fillcolor}{rgb}{1.0,0.6,0.6};
  \node (n0) at (42.0bp,370.05bp) [draw,fill=fillcolor,circle] {$\delta$};
  \definecolor{fillcolor}{rgb}{0.83,0.69,0.9};
  \node (n1) at (42.0bp,298.2bp) [draw,fill=fillcolor,circle] {$\delta$};
  \definecolor{fillcolor}{rgb}{0.83,0.69,0.9};
  \node (n2) at (121.0bp,226.35bp) [draw,fill=fillcolor,circle] {$\delta$};
  \definecolor{fillcolor}{rgb}{0.83,0.69,0.9};
  \node (n4) at (107.0bp,82.65bp) [draw,fill=fillcolor,circle] {$\delta$};
  \definecolor{fillcolor}{rgb}{0.83,0.69,0.9};
  \node (n3) at (130.0bp,154.5bp) [draw,fill=fillcolor,circle] {$\delta$};
  \definecolor{fillcolor}{rgb}{0.83,0.69,0.9};
  \node (n5) at (144.0bp,10.8bp) [draw,fill=fillcolor,circle] {$\delta$};
  \draw [->] (n0) ..controls (42.0bp,347.73bp) and (42.0bp,329.16bp)  .. (n1);
  \definecolor{strokecol}{rgb}{0.0,0.0,0.0};
  \pgfsetstrokecolor{strokecol}
  \draw (48.75bp,334.13bp) node {A};
  \draw [->] (n0) ..controls (60.558bp,351.39bp) and (80.826bp,330.4bp)  .. (93.0bp,309.0bp) .. controls (105.04bp,287.84bp) and (113.07bp,260.66bp)  .. (n2);
  \draw (110.75bp,298.2bp) node {B};
  \draw [->] (n0) ..controls (22.172bp,350.47bp) and (0.0bp,325.52bp)  .. (0.0bp,299.2bp) .. controls (0.0bp,299.2bp) and (0.0bp,299.2bp)  .. (0.0bp,153.5bp) .. controls (0.0bp,111.01bp) and (59.884bp,93.003bp)  .. (n4);
  \draw (6.75bp,226.35bp) node {C};
  \draw [->] (n1) ..controls (60.881bp,306.47bp) and (70.8bp,304.66bp)  .. (70.8bp,298.2bp) .. controls (70.8bp,293.16bp) and (64.746bp,290.95bp)  .. (n1);
  \draw (77.55bp,298.2bp) node {A};
  \draw [->] (n1) ..controls (63.587bp,278.11bp) and (92.012bp,252.98bp)  .. (n2);
  \draw (95.75bp,262.28bp) node {B};
  \draw [->] (n1) ..controls (49.793bp,259.21bp) and (67.905bp,176.68bp)  .. (93.0bp,111.45bp) .. controls (94.669bp,107.11bp) and (96.796bp,102.55bp)  .. (n4);
  \draw (75.75bp,190.43bp) node {C};
  \draw [->] (n2) ..controls (139.88bp,234.62bp) and (149.8bp,232.81bp)  .. (149.8bp,226.35bp) .. controls (149.8bp,221.31bp) and (143.75bp,219.1bp)  .. (n2);
  \draw (156.55bp,226.35bp) node {B};
  \draw [->] (n2) ..controls (117.62bp,206.98bp) and (116.33bp,194.2bp)  .. (118.5bp,183.3bp) .. controls (119.38bp,178.89bp) and (120.96bp,174.31bp)  .. (n3);
  \draw (125.75bp,190.43bp) node {A};
  \draw [->] (n2) ..controls (110.01bp,204.88bp) and (100.23bp,184.3bp)  .. (96.5bp,165.3bp) .. controls (92.031bp,142.52bp) and (97.461bp,115.77bp)  .. (n4);
  \draw (103.75bp,154.5bp) node {C};
  \draw [->] (n3) ..controls (136.04bp,173.56bp) and (138.78bp,186.54bp)  .. (136.0bp,197.55bp) .. controls (134.78bp,202.38bp) and (132.54bp,207.27bp)  .. (n2);
  \draw (144.75bp,190.43bp) node {B};
  \draw [->] (n3) ..controls (148.88bp,162.77bp) and (158.8bp,160.96bp)  .. (158.8bp,154.5bp) .. controls (158.8bp,149.46bp) and (152.75bp,147.25bp)  .. (n3);
  \draw (165.55bp,154.5bp) node {A};
  \draw [->] (n3) ..controls (123.03bp,132.32bp) and (116.69bp,113.08bp)  .. (n4);
  \draw (127.75bp,118.57bp) node {C};
  \draw [->] (n4) ..controls (133.63bp,96.192bp) and (167.55bp,115.47bp)  .. (181.0bp,143.7bp) .. controls (194.96bp,173.0bp) and (157.63bp,202.48bp)  .. (n2);
  \draw (190.75bp,154.5bp) node {B};
  \draw [->] (n4) ..controls (125.88bp,90.921bp) and (135.8bp,89.105bp)  .. (135.8bp,82.65bp) .. controls (135.8bp,77.607bp) and (129.75bp,75.396bp)  .. (n4);
  \draw (142.55bp,82.65bp) node {C};
  \draw [->] (n4) ..controls (101.32bp,63.126bp) and (99.244bp,49.726bp)  .. (104.5bp,39.6bp) .. controls (109.52bp,29.924bp) and (119.63bp,22.906bp)  .. (n5);
  \draw (111.75bp,46.725bp) node {A};
  \draw [->] (n5) ..controls (169.56bp,42.031bp) and (218.95bp,108.22bp)  .. (202.0bp,165.3bp) .. controls (197.13bp,181.7bp) and (193.93bp,186.35bp)  .. (181.0bp,197.55bp) .. controls (168.17bp,208.66bp) and (150.16bp,216.21bp)  .. (n2);
  \draw (210.75bp,118.57bp) node {B};
  \draw [->] (n5) ..controls (134.88bp,29.368bp) and (128.03bp,42.5bp)  .. (122.0bp,53.85bp) .. controls (119.66bp,58.245bp) and (117.1bp,63.017bp)  .. (n4);
  \draw (136.75bp,46.725bp) node {C};
  \draw [->] (n5) ..controls (162.88bp,19.071bp) and (172.8bp,17.255bp)  .. (172.8bp,10.8bp) .. controls (172.8bp,5.7568bp) and (166.75bp,3.5455bp)  .. (n5);
  \draw (179.55bp,10.8bp) node {A};
\end{tikzpicture}
        \caption*{NetKAT automaton}
    \end{subfigure}
    \begin{subfigure}[b]{0.16\textwidth}
        \centering
        \begin{tikzpicture}[>=latex',line join=bevel, scale=0.37]
\tikzstyle{every node}=[inner sep=0pt, minimum size=10pt, font=\tiny]
\tikzstyle{diam}=[diamond,minimum size=5pt]
\begin{scope}
  \pgfsetstrokecolor{black}
  \definecolor{strokecol}{rgb}{1.0,1.0,1.0};
  \pgfsetstrokecolor{strokecol}
  \definecolor{fillcol}{rgb}{1.0,1.0,1.0};
  \pgfsetfillcolor{fillcol}
  \filldraw (0.0bp,0.0bp) -- (0.0bp,420.3bp) -- (115.2bp,420.3bp) -- (115.2bp,0.0bp) -- cycle;
\end{scope}
\begin{scope}
  \pgfsetstrokecolor{black}
  \definecolor{strokecol}{rgb}{1.0,1.0,1.0};
  \pgfsetstrokecolor{strokecol}
  \definecolor{fillcol}{rgb}{1.0,1.0,1.0};
  \pgfsetfillcolor{fillcol}
  \filldraw (0.0bp,0.0bp) -- (0.0bp,420.3bp) -- (115.2bp,420.3bp) -- (115.2bp,0.0bp) -- cycle;
\end{scope}
\begin{scope}
  \pgfsetstrokecolor{black}
  \definecolor{strokecol}{rgb}{1.0,1.0,1.0};
  \pgfsetstrokecolor{strokecol}
  \definecolor{fillcol}{rgb}{1.0,1.0,1.0};
  \pgfsetfillcolor{fillcol}
  \filldraw (0.0bp,0.0bp) -- (0.0bp,420.3bp) -- (115.2bp,420.3bp) -- (115.2bp,0.0bp) -- cycle;
\end{scope}
\begin{scope}
  \pgfsetstrokecolor{black}
  \definecolor{strokecol}{rgb}{1.0,1.0,1.0};
  \pgfsetstrokecolor{strokecol}
  \definecolor{fillcol}{rgb}{1.0,1.0,1.0};
  \pgfsetfillcolor{fillcol}
  \filldraw (0.0bp,0.0bp) -- (0.0bp,420.3bp) -- (115.2bp,420.3bp) -- (115.2bp,0.0bp) -- cycle;
\end{scope}
\begin{scope}
  \pgfsetstrokecolor{black}
  \definecolor{strokecol}{rgb}{1.0,1.0,1.0};
  \pgfsetstrokecolor{strokecol}
  \definecolor{fillcol}{rgb}{1.0,1.0,1.0};
  \pgfsetfillcolor{fillcol}
  \filldraw (0.0bp,0.0bp) -- (0.0bp,420.3bp) -- (115.2bp,420.3bp) -- (115.2bp,0.0bp) -- cycle;
\end{scope}
\begin{scope}
  \pgfsetstrokecolor{black}
  \definecolor{strokecol}{rgb}{1.0,1.0,1.0};
  \pgfsetstrokecolor{strokecol}
  \definecolor{fillcol}{rgb}{1.0,1.0,1.0};
  \pgfsetfillcolor{fillcol}
  \filldraw (0.0bp,0.0bp) -- (0.0bp,420.3bp) -- (115.2bp,420.3bp) -- (115.2bp,0.0bp) -- cycle;
\end{scope}
  \node (n37) at (10.8bp,77.25bp) [draw,fill=peachpuff,diam] {};
  \node (n39) at (39.8bp,77.25bp) [draw,fill=peachpuff,diam] {};
  \definecolor{fillcolor}{rgb}{0.76,0.88,0.76};
  \node (n35) at (41.8bp,210.15bp) [draw,fill=fillcolor,diam] {};
  \definecolor{fillcolor}{rgb}{0.76,0.88,0.76};
  \node (n41) at (72.8bp,210.15bp) [draw,fill=fillcolor,diam] {};
  \node (n32) at (92.8bp,409.5bp) [draw,fill=skyblue,circle] {a};
  \node (n33) at (78.8bp,343.05bp) [draw,fill=skyblue,diam] {};
  \node (n42) at (109.8bp,343.05bp) [draw,fill=skyblue,diam] {};
  \definecolor{fillcolor}{rgb}{0.76,0.88,0.76};
  \node (n34) at (71.8bp,276.6bp) [draw,fill=fillcolor,circle] {b};
  \node (n36) at (33.8bp,143.7bp) [draw,fill=peachpuff,circle] {c};
  \node (n38) at (10.8bp,10.8bp) [draw,rectangle] {$\top$};
  \node (n40) at (72.8bp,10.8bp) [draw,rectangle] {$\bot$};
  \definecolor{strokecolor}{rgb}{0.02,0.02,0.02};
  \draw [strokecolor,->] (n32) ..controls (85.792bp,393.92bp) and (83.044bp,387.09bp)  .. (81.55bp,380.7bp) .. controls (79.589bp,372.31bp) and (78.931bp,362.58bp)  .. (n33);
  \definecolor{strokecol}{rgb}{0.0,0.0,0.0};
  \pgfsetstrokecolor{strokecol}
  \draw (87.425bp,373.57bp) node {1};
  \definecolor{strokecolor}{rgb}{0.02,0.02,0.02};
  \draw [strokecolor,->,dashed] (n32) ..controls (98.581bp,386.58bp) and (104.0bp,366.04bp)  .. (n42);
  \definecolor{strokecolor}{rgb}{0.02,0.02,0.02};
  \draw [strokecolor,->] (n33) ..controls (77.391bp,329.08bp) and (75.168bp,308.61bp)  .. (n34);
  \draw (82.425bp,312.53bp) node {1};
  \definecolor{strokecolor}{rgb}{0.02,0.02,0.02};
  \draw [strokecolor,->] (n34) ..controls (64.291bp,261.47bp) and (60.598bp,254.28bp)  .. (57.55bp,247.8bp) .. controls (53.216bp,238.59bp) and (48.666bp,227.87bp)  .. (n35);
  \draw (63.425bp,240.68bp) node {2};
  \definecolor{strokecolor}{rgb}{0.02,0.02,0.02};
  \draw [strokecolor,->,dashed] (n34) ..controls (72.138bp,253.81bp) and (72.438bp,234.51bp)  .. (n41);
  \definecolor{strokecolor}{rgb}{0.02,0.02,0.02};
  \draw [strokecolor,->] (n35) ..controls (40.19bp,196.18bp) and (37.649bp,175.71bp)  .. (n36);
  \draw (45.425bp,179.62bp) node {2};
  \definecolor{strokecolor}{rgb}{0.02,0.02,0.02};
  \draw [strokecolor,->] (n36) ..controls (26.929bp,128.19bp) and (23.866bp,121.24bp)  .. (21.55bp,114.9bp) .. controls (18.276bp,105.93bp) and (15.285bp,95.465bp)  .. (n37);
  \draw (27.425bp,107.78bp) node {3};
  \definecolor{strokecolor}{rgb}{0.02,0.02,0.02};
  \draw [strokecolor,->,dashed] (n36) ..controls (35.84bp,120.79bp) and (37.661bp,101.22bp)  .. (n39);
  \definecolor{strokecolor}{rgb}{0.02,0.02,0.02};
  \draw [strokecolor,->] (n37) ..controls (10.8bp,62.114bp) and (10.8bp,42.679bp)  .. (n38);
  \draw (16.425bp,46.725bp) node {3};
  \definecolor{strokecolor}{rgb}{0.02,0.02,0.02};
  \draw [strokecolor,->,dashed] (n39) ..controls (45.529bp,65.061bp) and (56.845bp,42.961bp)  .. (n40);
  \definecolor{strokecolor}{rgb}{0.02,0.02,0.02};
  \draw [strokecolor,->,dashed] (n41) ..controls (72.8bp,179.87bp) and (72.8bp,70.561bp)  .. (n40);
  \definecolor{strokecolor}{rgb}{0.02,0.02,0.02};
  \draw [strokecolor,->,dashed] (n42) ..controls (110.19bp,326.73bp) and (110.8bp,299.98bp)  .. (110.8bp,277.6bp) .. controls (110.8bp,277.6bp) and (110.8bp,277.6bp)  .. (110.8bp,76.25bp) .. controls (110.8bp,56.985bp) and (98.046bp,38.41bp)  .. (n40);
\end{tikzpicture}
        \caption*{A}
    \end{subfigure}
    \begin{subfigure}[b]{0.25\textwidth}
        \centering
        \begin{tikzpicture}[>=latex',line join=bevel, scale=0.37]
\tikzstyle{every node}=[inner sep=0pt, minimum size=10pt, font=\tiny]
\tikzstyle{diam}=[diamond,minimum size=5pt]
\begin{scope}
  \pgfsetstrokecolor{black}
  \definecolor{strokecol}{rgb}{1.0,1.0,1.0};
  \pgfsetstrokecolor{strokecol}
  \definecolor{fillcol}{rgb}{1.0,1.0,1.0};
  \pgfsetfillcolor{fillcol}
  \filldraw (0.0bp,0.0bp) -- (0.0bp,420.3bp) -- (202.4bp,420.3bp) -- (202.4bp,0.0bp) -- cycle;
\end{scope}
\begin{scope}
  \pgfsetstrokecolor{black}
  \definecolor{strokecol}{rgb}{1.0,1.0,1.0};
  \pgfsetstrokecolor{strokecol}
  \definecolor{fillcol}{rgb}{1.0,1.0,1.0};
  \pgfsetfillcolor{fillcol}
  \filldraw (0.0bp,0.0bp) -- (0.0bp,420.3bp) -- (202.4bp,420.3bp) -- (202.4bp,0.0bp) -- cycle;
\end{scope}
\begin{scope}
  \pgfsetstrokecolor{black}
  \definecolor{strokecol}{rgb}{1.0,1.0,1.0};
  \pgfsetstrokecolor{strokecol}
  \definecolor{fillcol}{rgb}{1.0,1.0,1.0};
  \pgfsetfillcolor{fillcol}
  \filldraw (0.0bp,0.0bp) -- (0.0bp,420.3bp) -- (202.4bp,420.3bp) -- (202.4bp,0.0bp) -- cycle;
\end{scope}
\begin{scope}
  \pgfsetstrokecolor{black}
  \definecolor{strokecol}{rgb}{1.0,1.0,1.0};
  \pgfsetstrokecolor{strokecol}
  \definecolor{fillcol}{rgb}{1.0,1.0,1.0};
  \pgfsetfillcolor{fillcol}
  \filldraw (0.0bp,0.0bp) -- (0.0bp,420.3bp) -- (202.4bp,420.3bp) -- (202.4bp,0.0bp) -- cycle;
\end{scope}
\begin{scope}
  \pgfsetstrokecolor{black}
  \definecolor{strokecol}{rgb}{1.0,1.0,1.0};
  \pgfsetstrokecolor{strokecol}
  \definecolor{fillcol}{rgb}{1.0,1.0,1.0};
  \pgfsetfillcolor{fillcol}
  \filldraw (0.0bp,0.0bp) -- (0.0bp,420.3bp) -- (202.4bp,420.3bp) -- (202.4bp,0.0bp) -- cycle;
\end{scope}
\begin{scope}
  \pgfsetstrokecolor{black}
  \definecolor{strokecol}{rgb}{1.0,1.0,1.0};
  \pgfsetstrokecolor{strokecol}
  \definecolor{fillcol}{rgb}{1.0,1.0,1.0};
  \pgfsetfillcolor{fillcol}
  \filldraw (0.0bp,0.0bp) -- (0.0bp,420.3bp) -- (202.4bp,420.3bp) -- (202.4bp,0.0bp) -- cycle;
\end{scope}
  \node (n23) at (34.4bp,77.25bp) [draw,fill=peachpuff,diam] {};
  \node (n28) at (79.4bp,77.25bp) [draw,fill=peachpuff,diam] {};
  \definecolor{fillcolor}{rgb}{0.76,0.88,0.76};
  \node (n20) at (61.4bp,210.15bp) [draw,fill=fillcolor,diam] {};
  \definecolor{fillcolor}{rgb}{0.76,0.88,0.76};
  \node (n26) at (97.4bp,210.15bp) [draw,fill=fillcolor,diam] {};
  \definecolor{fillcolor}{rgb}{0.76,0.88,0.76};
  \node (n31) at (146.4bp,210.15bp) [draw,fill=fillcolor,diam] {};
  \node (n17) at (129.4bp,409.5bp) [draw,fill=skyblue,circle] {a};
  \node (n18) at (115.4bp,343.05bp) [draw,fill=skyblue,diam] {};
  \node (n29) at (146.4bp,343.05bp) [draw,fill=skyblue,diam] {};
  \definecolor{fillcolor}{rgb}{0.76,0.88,0.76};
  \node (n19) at (97.4bp,276.6bp) [draw,fill=fillcolor,circle] {b};
  \definecolor{fillcolor}{rgb}{0.76,0.88,0.76};
  \node (n30) at (146.4bp,276.6bp) [draw,fill=fillcolor,circle] {b};
  \node (n21) at (34.4bp,143.7bp) [draw,fill=peachpuff,circle] {c};
  \node (n27) at (79.4bp,143.7bp) [draw,fill=peachpuff,circle] {c};
  \node (n22) at (5.4bp,77.25bp) [draw,fill=peachpuff,diam] {};
  \node (n24) at (38.4bp,10.8bp) [draw,rectangle] {$\top$};
  \node (n25) at (118.4bp,10.8bp) [draw,rectangle] {$\bot$};
  \definecolor{strokecolor}{rgb}{0.02,0.02,0.02};
  \draw [strokecolor,->] (n17) ..controls (122.39bp,393.92bp) and (119.64bp,387.09bp)  .. (118.15bp,380.7bp) .. controls (116.19bp,372.31bp) and (115.53bp,362.58bp)  .. (n18);
  \definecolor{strokecol}{rgb}{0.0,0.0,0.0};
  \pgfsetstrokecolor{strokecol}
  \draw (123.03bp,373.57bp) node {1};
  \definecolor{strokecolor}{rgb}{0.02,0.02,0.02};
  \draw [strokecolor,->,dashed] (n17) ..controls (135.18bp,386.58bp) and (140.6bp,366.04bp)  .. (n29);
  \definecolor{strokecolor}{rgb}{0.02,0.02,0.02};
  \draw [strokecolor,->] (n18) ..controls (111.92bp,329.58bp) and (105.96bp,308.23bp)  .. (n19);
  \draw (114.03bp,312.53bp) node {1};
  \definecolor{strokecolor}{rgb}{0.02,0.02,0.02};
  \draw [strokecolor,->] (n19) ..controls (85.483bp,254.26bp) and (73.098bp,232.09bp)  .. (n20);
  \draw (88.025bp,240.68bp) node {2};
  \definecolor{strokecolor}{rgb}{0.02,0.02,0.02};
  \draw [strokecolor,->,dashed] (n19) ..controls (97.4bp,253.81bp) and (97.4bp,234.51bp)  .. (n26);
  \definecolor{strokecolor}{rgb}{0.02,0.02,0.02};
  \draw [strokecolor,->] (n20) ..controls (56.39bp,197.19bp) and (47.061bp,174.92bp)  .. (n21);
  \draw (58.025bp,179.62bp) node {2};
  \definecolor{strokecolor}{rgb}{0.02,0.02,0.02};
  \draw [strokecolor,->] (n21) ..controls (26.141bp,128.69bp) and (22.181bp,121.5bp)  .. (19.15bp,114.9bp) .. controls (14.982bp,105.82bp) and (11.067bp,95.084bp)  .. (n22);
  \draw (25.025bp,107.78bp) node {3};
  \definecolor{strokecolor}{rgb}{0.02,0.02,0.02};
  \draw [strokecolor,->,dashed] (n21) ..controls (34.4bp,120.91bp) and (34.4bp,101.61bp)  .. (n23);
  \definecolor{strokecolor}{rgb}{0.02,0.02,0.02};
  \draw [strokecolor,->] (n23) ..controls (25.66bp,71.581bp) and (13.764bp,64.112bp)  .. (9.15bp,53.85bp) .. controls (6.5527bp,48.074bp) and (6.6522bp,45.42bp)  .. (9.15bp,39.6bp) .. controls (11.985bp,32.995bp) and (17.245bp,27.219bp)  .. (n24);
  \draw (15.025bp,46.725bp) node {3};
  \definecolor{strokecolor}{rgb}{0.02,0.02,0.02};
  \draw [strokecolor,->,dashed] (n23) ..controls (37.67bp,64.897bp) and (43.316bp,48.963bp)  .. (53.4bp,39.6bp) .. controls (66.972bp,26.999bp) and (87.126bp,19.594bp)  .. (n25);
  \definecolor{strokecolor}{rgb}{0.02,0.02,0.02};
  \draw [strokecolor,->,dashed] (n26) ..controls (100.4bp,180.96bp) and (112.21bp,70.003bp)  .. (n25);
  \definecolor{strokecolor}{rgb}{0.02,0.02,0.02};
  \draw [strokecolor,->] (n26) ..controls (90.741bp,203.12bp) and (83.137bp,195.27bp)  .. (80.15bp,186.75bp) .. controls (77.283bp,178.57bp) and (76.914bp,168.93bp)  .. (n27);
  \draw (86.025bp,179.62bp) node {2};
  \definecolor{strokecolor}{rgb}{0.02,0.02,0.02};
  \draw [strokecolor,->,dashed] (n27) ..controls (79.4bp,120.91bp) and (79.4bp,101.61bp)  .. (n28);
  \definecolor{strokecolor}{rgb}{0.02,0.02,0.02};
  \draw [strokecolor,->] (n28) ..controls (72.237bp,64.99bp) and (57.973bp,42.568bp)  .. (n24);
  \draw (71.025bp,46.725bp) node {3};
  \definecolor{strokecolor}{rgb}{0.02,0.02,0.02};
  \draw [strokecolor,->,dashed] (n28) ..controls (86.214bp,64.99bp) and (99.781bp,42.568bp)  .. (n25);
  \definecolor{strokecolor}{rgb}{0.02,0.02,0.02};
  \draw [strokecolor,->,dashed] (n29) ..controls (161.79bp,333.93bp) and (202.4bp,309.91bp)  .. (202.4bp,277.6bp) .. controls (202.4bp,277.6bp) and (202.4bp,277.6bp)  .. (202.4bp,76.25bp) .. controls (202.4bp,42.841bp) and (160.06bp,24.383bp)  .. (n25);
  \definecolor{strokecolor}{rgb}{0.02,0.02,0.02};
  \draw [strokecolor,->] (n29) ..controls (146.4bp,327.91bp) and (146.4bp,308.48bp)  .. (n30);
  \draw (152.03bp,312.53bp) node {1};
  \definecolor{strokecolor}{rgb}{0.02,0.02,0.02};
  \draw [strokecolor,->,dashed] (n30) ..controls (146.4bp,253.81bp) and (146.4bp,234.51bp)  .. (n31);
  \definecolor{strokecolor}{rgb}{0.02,0.02,0.02};
  \draw [strokecolor,->,dashed] (n31) ..controls (145.84bp,184.12bp) and (143.1bp,103.65bp)  .. (128.4bp,39.6bp) .. controls (127.52bp,35.757bp) and (126.26bp,31.715bp)  .. (n25);
  \definecolor{strokecolor}{rgb}{0.02,0.02,0.02};
  \draw [strokecolor,->] (n31) ..controls (135.36bp,198.53bp) and (107.63bp,171.85bp)  .. (n27);
  \draw (128.03bp,179.62bp) node {2};
\end{tikzpicture}
        \caption*{B}
    \end{subfigure}
    \begin{subfigure}[b]{0.16\textwidth}
        \centering
        \begin{tikzpicture}[>=latex',line join=bevel, scale=0.37]
\tikzstyle{every node}=[inner sep=0pt, minimum size=10pt, font=\tiny]
\tikzstyle{diam}=[diamond,minimum size=5pt]
\begin{scope}
  \pgfsetstrokecolor{black}
  \definecolor{strokecol}{rgb}{1.0,1.0,1.0};
  \pgfsetstrokecolor{strokecol}
  \definecolor{fillcol}{rgb}{1.0,1.0,1.0};
  \pgfsetfillcolor{fillcol}
  \filldraw (0.0bp,0.0bp) -- (0.0bp,407.05bp) -- (77.2bp,407.05bp) -- (77.2bp,0.0bp) -- cycle;
\end{scope}
\begin{scope}
  \pgfsetstrokecolor{black}
  \definecolor{strokecol}{rgb}{1.0,1.0,1.0};
  \pgfsetstrokecolor{strokecol}
  \definecolor{fillcol}{rgb}{1.0,1.0,1.0};
  \pgfsetfillcolor{fillcol}
  \filldraw (0.0bp,0.0bp) -- (0.0bp,407.05bp) -- (77.2bp,407.05bp) -- (77.2bp,0.0bp) -- cycle;
\end{scope}
\begin{scope}
  \pgfsetstrokecolor{black}
  \definecolor{strokecol}{rgb}{1.0,1.0,1.0};
  \pgfsetstrokecolor{strokecol}
  \definecolor{fillcol}{rgb}{1.0,1.0,1.0};
  \pgfsetfillcolor{fillcol}
  \filldraw (0.0bp,0.0bp) -- (0.0bp,407.05bp) -- (77.2bp,407.05bp) -- (77.2bp,0.0bp) -- cycle;
\end{scope}
\begin{scope}
  \pgfsetstrokecolor{black}
  \definecolor{strokecol}{rgb}{1.0,1.0,1.0};
  \pgfsetstrokecolor{strokecol}
  \definecolor{fillcol}{rgb}{1.0,1.0,1.0};
  \pgfsetfillcolor{fillcol}
  \filldraw (0.0bp,0.0bp) -- (0.0bp,407.05bp) -- (77.2bp,407.05bp) -- (77.2bp,0.0bp) -- cycle;
\end{scope}
\begin{scope}
  \pgfsetstrokecolor{black}
  \definecolor{strokecol}{rgb}{1.0,1.0,1.0};
  \pgfsetstrokecolor{strokecol}
  \definecolor{fillcol}{rgb}{1.0,1.0,1.0};
  \pgfsetfillcolor{fillcol}
  \filldraw (0.0bp,0.0bp) -- (0.0bp,407.05bp) -- (77.2bp,407.05bp) -- (77.2bp,0.0bp) -- cycle;
\end{scope}
\begin{scope}
  \pgfsetstrokecolor{black}
  \definecolor{strokecol}{rgb}{1.0,1.0,1.0};
  \pgfsetstrokecolor{strokecol}
  \definecolor{fillcol}{rgb}{1.0,1.0,1.0};
  \pgfsetfillcolor{fillcol}
  \filldraw (0.0bp,0.0bp) -- (0.0bp,407.05bp) -- (77.2bp,407.05bp) -- (77.2bp,0.0bp) -- cycle;
\end{scope}
  \node (n12) at (10.8bp,77.25bp) [draw,fill=peachpuff,diam] {};
  \node (n14) at (39.8bp,77.25bp) [draw,fill=peachpuff,diam] {};
  \definecolor{fillcolor}{rgb}{0.76,0.88,0.76};
  \node (n10) at (40.8bp,210.15bp) [draw,fill=fillcolor,diam] {};
  \definecolor{fillcolor}{rgb}{0.76,0.88,0.76};
  \node (n16) at (71.8bp,210.15bp) [draw,fill=fillcolor,diam] {};
  \node (n6) at (38.8bp,396.25bp) [draw,fill=skyblue,circle] {a};
  \node (n8) at (54.8bp,329.8bp) [draw,fill=skyblue,diam] {};
  \definecolor{fillcolor}{rgb}{0.76,0.88,0.76};
  \node (n9) at (54.8bp,276.6bp) [draw,fill=fillcolor,circle] {b};
  \node (n11) at (33.8bp,143.7bp) [draw,fill=peachpuff,circle] {c};
  \node (n7) at (25.8bp,329.8bp) [draw,fill=skyblue,diam] {};
  \node (n13) at (10.8bp,10.8bp) [draw,rectangle] {$\top$};
  \node (n15) at (61.8bp,10.8bp) [draw,rectangle] {$\bot$};
  \definecolor{strokecolor}{rgb}{0.02,0.02,0.02};
  \draw [strokecolor,->] (n6) ..controls (31.704bp,380.69bp) and (28.957bp,373.86bp)  .. (27.55bp,367.45bp) .. controls (25.706bp,359.05bp) and (25.32bp,349.32bp)  .. (n7);
  \definecolor{strokecol}{rgb}{0.0,0.0,0.0};
  \pgfsetstrokecolor{strokecol}
  \draw (33.425bp,360.32bp) node {1};
  \definecolor{strokecolor}{rgb}{0.02,0.02,0.02};
  \draw [strokecolor,->,dashed] (n6) ..controls (44.212bp,373.45bp) and (49.242bp,353.19bp)  .. (n8);
  \definecolor{strokecolor}{rgb}{0.02,0.02,0.02};
  \draw [strokecolor,->,dashed] (n8) ..controls (54.8bp,316.81bp) and (54.8bp,304.21bp)  .. (n9);
  \definecolor{strokecolor}{rgb}{0.02,0.02,0.02};
  \draw [strokecolor,->] (n9) ..controls (47.792bp,261.02bp) and (45.044bp,254.19bp)  .. (43.55bp,247.8bp) .. controls (41.589bp,239.41bp) and (40.931bp,229.68bp)  .. (n10);
  \draw (49.425bp,240.68bp) node {2};
  \definecolor{strokecolor}{rgb}{0.02,0.02,0.02};
  \draw [strokecolor,->,dashed] (n9) ..controls (60.581bp,253.68bp) and (66.0bp,233.14bp)  .. (n16);
  \definecolor{strokecolor}{rgb}{0.02,0.02,0.02};
  \draw [strokecolor,->] (n10) ..controls (39.391bp,196.18bp) and (37.168bp,175.71bp)  .. (n11);
  \draw (44.425bp,179.62bp) node {2};
  \definecolor{strokecolor}{rgb}{0.02,0.02,0.02};
  \draw [strokecolor,->] (n11) ..controls (26.929bp,128.19bp) and (23.866bp,121.24bp)  .. (21.55bp,114.9bp) .. controls (18.276bp,105.93bp) and (15.285bp,95.465bp)  .. (n12);
  \draw (27.425bp,107.78bp) node {3};
  \definecolor{strokecolor}{rgb}{0.02,0.02,0.02};
  \draw [strokecolor,->,dashed] (n11) ..controls (35.84bp,120.79bp) and (37.661bp,101.22bp)  .. (n14);
  \definecolor{strokecolor}{rgb}{0.02,0.02,0.02};
  \draw [strokecolor,->] (n12) ..controls (10.8bp,62.114bp) and (10.8bp,42.679bp)  .. (n13);
  \draw (16.425bp,46.725bp) node {3};
  \definecolor{strokecolor}{rgb}{0.02,0.02,0.02};
  \draw [strokecolor,->,dashed] (n14) ..controls (43.825bp,64.457bp) and (51.187bp,42.891bp)  .. (n15);
  \definecolor{strokecolor}{rgb}{0.02,0.02,0.02};
  \draw [strokecolor,->,dashed] (n16) ..controls (70.372bp,180.96bp) and (64.749bp,70.003bp)  .. (n15);
\end{tikzpicture}
        \caption*{C}
    \end{subfigure}
    \caption{Symbolic NetKAT automaton for $p \triangleq \left((a \mut 1 \cdot b \mut 2 \cdot c \mut 3 \cdot \text{dup})^\star + (b \test 2 \cdot c \test 3 \cdot \text{dup})^\star\right)^\star$}
    \label{fig:exaut}
\end{figure}

For more details, see \appendixref{app:spp}.

\section{Bisimilarity and Counter-Example Generation}\label{sec:fwdbwd}

After converting the NetKAT programs to automata, we can check equivalence of two NetKAT programs by checking whether their automata are bisimilar. Bisimilarity is a well-known notion of equivalence for automata. For NetKAT automata, two states are bisimilar for a given packet $\pk$ if (1) the states produce the same output packets, and (2) for each possible modified packet $\pk'$, the two states transition to states that are bisimilar for $\pk'$.

Historically, methods for computing bisimulations of automata \cite{Bonchi2013,Doenges2022} have been based on either on Moore's algorithm for minimization \cite{Moore1956} or Hopcroft and Karp's algorithm \cite{Hopcroft1971}. Indeed, the naive approach shown in \Cref{fig:bisim} follows the basic structure of Hopcroft-Karp, in the sense that it relates the start states and proceeds to follow transitions forward in the two automata.

Our situation is different, because we have already added the symmetric difference operator to NetKAT.  Therefore, we can reduce equivalence checks $A \equiv B$ to emptiness checks $A \oplus B \equiv
\bot$.

\paragraph{Subtlety of symmetric difference}
The reader should note that the situation is a bit more subtle than it may seem at first sight.  In particular, consider the query $A \equiv B$, which asks whether two NetKAT programs $A$ and $B$ are equivalent. If they are inequivalent, then there must be some input packet that causes $A$ to produce a different output than $B$.  Output in this sense does not just mean the final output packet, but also the trace of the packet.  It can be the case that $A$ and $B$ produce the same final output packets, but differ in the trace of the packets. Therefore, the symmetric difference $A \oplus B$ takes the symmetric difference of the traces, not just of the output packets. In other words, the set $\Sem{A \oplus B}$ is the set of counter-example traces to the equivalence of $A$ and $B$.

\paragraph{Subtlety of the emptiness of an automaton}
A second subtlety is the emptiness of an automaton.  Whereas for regular expressions, it is easy to check whether a DFA is empty (just check whether there are any reachable accepting states), this is not the case for NetKAT automata.  Because NetKAT automata manipulate and test the fields of packets, it is possible that a packet travels through the automaton, and even causes multiple packets to be produced (e.g., due to the presence of union in the original NetKAT expression), but nevertheless, it is possible that all of these packets are eventually dropped by the automaton, before producing any output packets.

The goal of this section is to develop a symbolic algorithm for this check, which is more efficient than the naive algorithm that checks all concrete input packets separately.

\subsection{Forward Algorithm}

To check whether a NetKAT automaton drops all input packets, we use an algorithm that works in the forward direction. The forward algorithm computes \emph{all} output packets that the automaton can produce, across all possible input packets. More formally: given a NetKAT automaton for a NetKAT program $p$, the forward algorithm computes the set of final packets occurring in the traces $\Sem{p}(\pk)$ across all input packets $\pk$. The forward algorithm therefore starts with the complete symbolic packet $\top$ at the input state, and repeatedly applies all outgoing transitions $\delta$ to it.  In this way, the algorithm iteratively accumulates a symbolic packet at every state, which represents the set of packets that can reach that state from the start state.  Once we know the set of packets that can reach a state, we can determine the set of output packets by applying the output function $\epsilon$ to the symbolic packet of each state, and taking the union of the results.

\newcommand{\pluseq}{\mathrel{\nf{+}}=}
\begin{figure}
    \begin{subfigure}[t]{0.49\textwidth}
    \begin{algorithmic}
    % \State \textbf{Input:} A NetKAT automaton $(Q, q_0, \delta, \epsilon).$
    % \State \textbf{Returns:} The symbolic set of output packets that the automaton can produce.\\
    \State $\doneR(q) \gets \bot \text{ for } q \in Q$;
    \State $\todoR(q) \gets \bot \text{ for } q \in Q \setminus \{q_0\}$;
    \State $\todoR(q_0) \gets \top$;
    \While{$\exists q.\ \todoR(q) \neq \bot$}
        \State set $p := \todoR(q) \mathrel{\nf{-}} \doneR(q)$;
        \State $\todoR(q) \gets \bot$;
        \State $\doneR(q) \pluseq p$;
        \For {$q' \in Q$}
            \State $\todoR(q') \pluseq \mathsf{push}(p,\delta(q,q'))$;
          \EndFor
    \EndWhile
    \State
    \Return $\sum_{q \in Q} \mathsf{push}(\doneR(q), \epsilon(q))$;
    \end{algorithmic}
    \caption{Forward algorithm}
    \end{subfigure}
    \begin{subfigure}[t]{0.49\textwidth}
    \begin{algorithmic}
    % \State \textbf{Input:} a NetKAT automaton $(Q, q_0, \delta, \epsilon).$
    % \State \textbf{Returns:} The symbolic set of input packets that the automaton does not drop.\\
    \State $\doneR(q) \gets \bot \text{ for } q \in Q$;
    \State $\todoR(q) \gets \mathsf{pull}(\epsilon(q), \top) \text{ for } q \in Q$;
    \State
    \While{$\exists q.\ \todoR(q) \neq \bot$}
        \State set $p := \todoR(q) \mathrel{\nf{-}} \doneR(q)$;
        \State $\todoR(q) \gets \bot$;
        \State $\doneR(q) \pluseq p$;
        \For {$q' \in Q$}
            \State $\todoR(q') \pluseq \mathsf{pull}(\delta(q',q), p)$;
          \EndFor
    \EndWhile
    \State
    \Return $\doneR(q_0)$;
    \end{algorithmic}
    \caption{Backward algorithm}
    \end{subfigure}
    \caption{Forward and backward algorithms for NetKAT automata}
    \label{fig:forward}
    \label{fig:backward}
\end{figure}

The forward algorithm is shown in \Cref{fig:forward}.  To determine the set of packets that a given NetKAT expression $p$ can produce, we convert it to a NetKAT automaton, and then use the forward algorithm to determine the symbolic output packet.  In our implementation, we also have a way to stop the algorithm early, if the user is only interested in a ``yes'' or ``no'' answer for the query $p \equiv\bot$. In this case, we can stop the algorithm as soon as any output packet is produced.

\subsection{Backward Algorithm}

The forward algorithm gives us the set of output packets that a NetKAT automaton can produce, but we are also interested in which input packets cause this output to be produced. For instance, for a given check $A \equiv B$, we want to know \emph{which} input packets cause $A$ and $B$ to produce different sets of output traces.  More generally,  we want to know which input packets cause a NetKAT automaton to produce a nonempty set of output packets or, formally, given an automaton for a NetKAT program $p$, what is the set of initial packets $\pk$ for which $\Sem{p}(\pk)$ is nonempty.

To answer this question, we developed a \emph{backward algorithm}, shown in \Cref{fig:backward}.  The backward algorithm computes \emph{all} input packets that can cause the automaton to produce a nonempty set of output packets.  The backward algorithm therefore starts with the complete symbolic packet $\top$ at every state, and pulls it backwards through all observation functions $\epsilon$.  The algorithm then iteratively accumulates a symbolic packet at every state, which represents the set of packets that will cause the automaton to produce a nonempty set of output packets, when starting from that state.  The accumulation is done by pulling the symbolic packet backwards through all transitions until a fixpoint is reached.  Once the fixpoint is reached, we return the symbolic packet at the start state, which represents the set of input packets that will cause the automaton to produce a nonempty set of output packets when starting from the start state.

\section{Implementation}\label{sec:impl}

\begin{wrapfigure}{r}{0.25\textwidth}
  \centering \vspace{-.5cm}
  \begin{tabular}{l}
      % \toprule
      \textbf{Statements}\\[2mm]
      % \midrule
        $\NKcheck\ e_1\ \equiv\ e_2$ \\
        $\NKcheck\ e_1\ \nequiv\ e_2$ \\
        $\NKprint\ e$\\
        $x = e$\\
        $\NKfor\ i\ \in n_1..n_2\ \NKdo\ c$\\[4mm]
      % \toprule
      \textbf{Expressions}\\[2mm]
      % \midrule
        $\NKforward\ e$,
        $\NKbackward\ e$\\[1mm]
        $e_1 \!\cap \!e_2$,\ \
        $e_1 \!\oplus \!e_2$,\ \
        $e_1 \!-\! e_2$\\[1mm]
        $\NKexists{} f\ e$,\ \
        $\NKforall{} f\ e$\\[1mm]
        (+ \NetKAT, see \cref{fig:synsem})\\
      % \bottomrule
  \end{tabular}
  \caption{NKPL syntax.}
  \label{fig:nknf}
  \vspace{-1cm}
\end{wrapfigure}

We have implemented the algorithms in a new system, \KATch{}, comprising 2500 lines of Scala. In this section we explain the system's interface. The implementation provides a surface syntax for expressing queries, which extends the core \NetKAT syntax from \Cref{fig:synsem}. The extended syntax is shown in \Cref{fig:nknf}. The language has statements and expressions, which we describe below.

\paragraph*{Statements}
The $\NKcheck\ e_1 \equiv e_2$ command runs the bisimulation algorithm. The system reports success if the expressions are equivalent (when $\equiv$ is used) or inequivalent (when $\nequiv$) is used, or reports failure otherwise. The other statements behave as expected:
The $\NKprint$ statement invokes the pretty printer, $x = e$ let-binds names to expressions, and $\NKfor$ runs a statement in a loop.

\paragraph*{Expressions}
The $\NKforward\ e$ expression computes, in forward-flowing mode, the set of output packets resulting from running the symbolic packet $\top$ through the given expression $e$ (see \Cref{fig:forward}). Conversely $\NKbackward\ e$ computes the set of input packets which generate some output packet when run on the given expression (\Cref{fig:backward}). These are generally used in conjunction with the $\oplus, -, \cap$ operators to express the desired query. The user may combine these expressions with the $\NKexists$ and $\NKforall$ operators to reason about symbolic packets, and the $\NKprint$ operator to pretty print the symbolic packets, or the $\NKcheck$ operator to assert (in)equivalence of two expressions.

We note that the generality of the language allows us to express some queries in different, equivalent ways. For example, the two checks:
$e_1 \equiv e_2$ and $e_1 \oplus e_2 \equiv \zero$ are equivalent. However, the expression on the right lends itself to inspection of counterexample input packets: $\NKprint\ (\NKbackward\ e_1 \oplus e_2)$.
This statement pretty prints a symbolic packet having different behavior on $e_1$ and $e_2$ (i.e., packets that result in some valid history in one expression and not the other). In other words, it prints the set of all counter-example input packets for $e_1 \equiv e_2$.

\paragraph{Topologies and routing tables}
While NetKAT can be used to specify routing policies declaratively, it is also possible to import topologies and routing tables into NetKAT. For example, a simple way to define routing tables and topologies in NetKAT is as follows:

\medskip
\noindent
\begin{tabular}{lc@{\ \ \ \ }|@{\ \ \ \ }lc}

    $\mathsf{R} \triangleq$\!\!\!\!&$\sw \test 5 \cdot (\dst \test 6 \cdot \pt \mut 3 + \dst \test 8 \cdot \pt \mut 4)$&
    $\mathsf{T} \triangleq$\!\!\!\!&$\sw \test 5 \cdot (\pt=3 \cdot \sw \mut 6 + \pt=4 \cdot \sw \mut 8)$ \\
        &$+\ \cdots\ +$&
        &$+\ \cdots\ +$\\
        &$ \sw \test 7 \cdot (\dst \test 8 \cdot \pt \mut 2 + \dst \test 9 \cdot \pt \mut 1) $&
        &$ \sw \test 7 \cdot (\pt=2 \cdot \sw \mut 8 + \pt=1 \cdot \sw \mut 9)$
\end{tabular}
\medskip

The routing $\mathsf{R}$ tests the switch field (where the packet currently is), and the destination field (where the packet is supposed to end up), and then sets the port over which the packet should be sent out. The topology $\mathsf{T}$ tests the switch and port fields, and then transports the packet to the next switch. We define the action of the network by composing the route and topology:
\begin{align*}
    \mathsf{net} \triangleq \mathsf{R} \cdot \mathsf{T} \cdot \text{dup}
\end{align*}
We include a dup to extend the trace of the packet at every hop.

\begin{example}[All-Pairs Reachability Queries]\label{ex:linear-reachability}
A naive way to check reachability of all pairs of hosts in a network is to run the following command for each pair of end hosts $n_i$, $n_j$:
\[ \NKcheck\ (\sw \test n_i) \cdot \textsf{net}^\star \cdot (\sw \test n_j) \nequiv \bot \]

Of course, this requires a number of queries which is quadratic in the number of hosts---quickly becoming prohibitive. One might think that we could reduce the number of queries to $n$ by running:
\begin{align*}
    &\NKfor\ i \in 1..n\ \NKdo\ \NKcheck\ (\NKforward\ (\sw \test i \cdot \textsf{net}^\star)) \equiv (\sw \in 1..n)
\end{align*}
This query does work if $\sw$ is the only field of our packets, because the left hand side contains the packets that can be reached from $i$ via the network, and the right hand side contains packets where the $\sw$ field is any value in $1..n$. Unfortunately, this does not quite work if the network also operates on other packet fields (say, fields $f_1$ and $f_2$), as the packets on the left hand side will have those fields, whereas the packets on the right hand side will only have a $\sw$ field. The $\NKexists$ and $\NKforall$ operators allow us to manipulate symbolic packets and express all pairs reachability in $n$ queries:
\begin{align*}
    &\NKfor\ i \in 1..n\ \NKdo\ \NKcheck\ (\NKexists f_1\ (\NKexists f_2\
    (\NKforward\ (\sw=i \cdot \textsf{net}^\star)))) \equiv (\sw \in 1..n)
\end{align*}
The operators $\NKexists$ and $\NKforall$ give the programmer the ability to reason about symbolic packets that may be computed, for instance, by $\NKforward$ or $\NKbackward$. The specifications are:
\[
  p \in \NKexists f\ e \overset{\triangle}{\iff} \exists v\in V\colon p[f\mut v] \in \pSem{e} \quad\text{and}\quad
  p \in \NKforall f\ e \overset{\triangle}{\iff} \forall v\in V\colon p[f\mut v] \in \pSem{e}
\]
The implementation does not iterate over $V$ (indeed, $V$ can be unknown). Rather, $\NKexists$ and $\NKforall$ are implemented directly as operations on symbolic packets. In fact, the implementation does not need to fix the sets of fields or 
  values up-front at all. Instead, it operates on a conceptually infinite set of fields and values. This works because \SPPn{}s' default cases handle all remaining fields and all values.
\end{example}

\paragraph{Correctness and testing methodology}

We validated our implementation with the following property-based fuzz testing methodology:
\begin{enumerate}
    \item We implemented the semantics $\Sem{p}(\pk)$ in Scala, as described in \Cref{fig:synsem}.
    \item We repeatedly pick two SPs/SPPs $p$ and $q$ and a packet $\pk$. We enumerated small SPPs up to a bound and generated larger SPPs randomly. We generated packets exhaustively.
    \item We check the soundness (e.g., $\Sem{p \mathop{\nf{+}} q}(\pk) = \Sem{p}(\pk) \cup \Sem{q}(\pk)$) and canonicity (e.g., $\Sem{p}(\pk) = \Sem{q}(\pk)$ for all $\pk$ if and only if $p = q$) of the operations, as well as additional algebraic laws derived from the NetKAT axioms (e.g., $p \mathop{\nf{+}} q = q \mathop{\nf{+}} p$).
\end{enumerate}

In addition to property-based testing for SPs and SPPs, we also generated hundreds of thousands of pairs of random NetKAT expressions and checked that \KATch's bisimilarity check matches the output of \Frenetic. This revealed a subtle bug in Frenetic, which we reported and is now fixed.

\begin{figure}
       \includegraphics[scale=0.6]{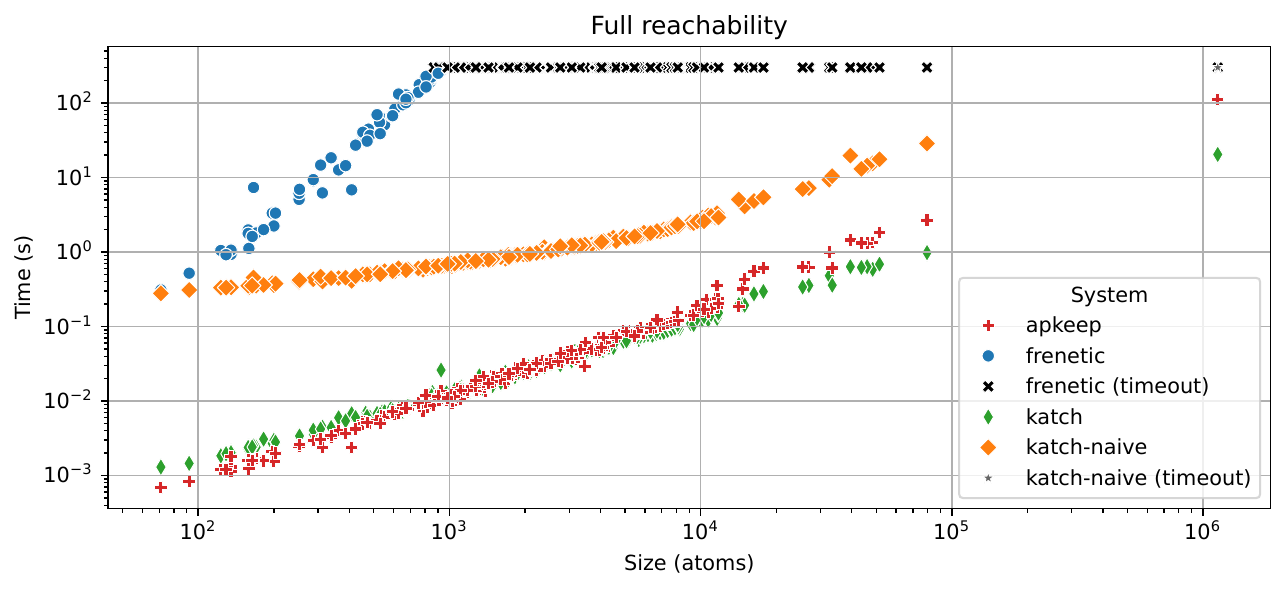}
    \caption{Full reachability queries on Topology Zoo. \KATch{} and \APK results are averages of 100 runs, preceded by 10 runs of JIT warmup. \KATch{}-naive is without JIT warmup and uses $O(n^2)$ 1-to-1 queries.}\label{fig:topology-zoo-full}
\end{figure}

\section{Evaluation}\label{sec:eval}

To evaluate \KATch{}, we conducted experiments in which we used it to solve a variety of verification tasks for a range of topologies and routes, as well as challenging combinatorial \NetKAT terms. The goal of our evaluation is to answer the following three questions:

\begin{enumerate}
  \item How does \KATch{} perform compared to the state of the art
    \NetKAT verifier, \Frenetic?
  \item How does \KATch{} perform compared to the state of the art
    specialized network verification tool, \APK?
  \item How well does \KATch{} scale with the size of topology?
  \item When does \KATch{} perform asymptotically better than prior work?
\end{enumerate}

\subsection{Topology Zoo}

\begin{figure}
  \scalebox{0.85}{\begin{tabular}{|c|r|rr|rr|rr|}
\toprule
\multicolumn{1}{|c|}{\textbf{Name}} & \multicolumn{1}{c|}{\textbf{Size}} & \multicolumn{2}{c|}{\textbf{1-to-1 reachability}} & \multicolumn{2}{c|}{\textbf{Slicing}} & \multicolumn{2}{c|}{\textbf{Full reachability}} \\
(Topology Zoo) & \multicolumn{1}{c|}{(atoms)} & \KATch & Frenetic & \KATch & Frenetic & \KATch & \APK\\
\midrule
Layer42   & 135   & 0.00   & 0.04 & 0.00 & 0.07 & 0.00 & 0.00\\
Compuserv & 539   & 0.00   & 0.35 & 0.01 & 0.81 & 0.01 & 0.01\\
Airtel    & 785   & 0.01   & 0.82 & 0.02 & 1.98 & 0.01 & 0.01\\
Belnet    & 1388   & 0.01  & 3.19 & 0.03 & 7.99 & 0.02 & 0.01\\
Shentel   & 1865   & 0.01  & 3.93 & 0.03 & 9.76 & 0.02 & 0.03\\
Arpa      & 1964   & 0.01  & 4.23 & 0.04 & 10.42 & 0.03 & 0.03\\
Sanet     & 4100   & 0.03  & 23.42 & 0.07 & 62.21 & 0.05 & 0.07\\
Uunet     & 5456   & 0.04  & 80.85 & 0.11 & 203.80 & 0.07 & 0.07\\
Missouri  & 9680   & 0.06  & 166.87 & 0.21 & 441.72 & 0.13 & 0.19\\
Telcove   & 10720  & 0.07  & 441.47 & 0.21 & 1121.30 & 0.12 & 0.16\\
Deltacom  & 27092  & 0.18  & 2087.34 & 0.48 & 5098.57 & 0.36 & 0.63\\
Cogentco  & 79682  & 0.53  & 18910.82 & 1.38 & 54247.73 & 0.98 & 2.67\\
Kdl       & 1144691 & 9.87 & \raisebox{.4mm}{\tiny{out of memory}} & 23.67 & \raisebox{.4mm}{\tiny{out of memory}} & 19.44 & 110.92\\
% Kdl       & 1144691 & 9.87 & \smaller{$\substack{\text{out of memory} \\ (>200\text{GB})}$} & 23.67 & \smaller{$\substack{\text{out of 200GB} \\ \text{of memory}}$} & 19.44 & 110.92\\
\bottomrule
\end{tabular}

}\\
  \caption{Running time (in seconds) of \KATch{}, \Frenetic, and \APK on Topology Zoo queries. \KATch{} and \APK results are averages of 100 runs, preceded by 10 runs of JIT warmup.}\label{fig:topology-zoo}
\end{figure}

To begin to answer the first three questions, we conducted our experiments using \emph{The Internet Topology Zoo} \citep{Knight2011} dataset, a publicly available set of 261 network topologies, ranging in size from just 4 nodes (the original ARPANet) to 754 nodes (KDL, the Kentucky Data Link ISP topology).  For each topology, we generated a destination-based routing policy using an all-pairs shortest path scheme that connects every pair of routers to each other.

To demonstrate \KATch{}'s scalability, we first ran full (i.e., $O(n^2)$) reachability queries for every topology in the zoo using \KATch{}, \Frenetic\footnote{\url{https://github.com/frenetic-lang/frenetic}}, and \APK{}~\cite{Zhang2020}.
The results are shown in \Cref{fig:topology-zoo-full} in a log-log plot, so straight lines correspond to  polynomials with exponents related to their slope.

The figure shows two different configurations of \KATch{}: one that verifies full reachability using a quadratic number of point-to-point queries and does not use JIT warmup (``katch-naive''), and one that verifies full reachability using a linear number of queries and does use JIT warmup (``katch'').

Because of the size of the dataset, we set a timeout of 5 minutes per topology. Under these conditions, \Frenetic{} was unable to complete for all but the smallest topologies. \KATch{}-naive handles most of the topologies in well under a second, and all but the largest in under 2 minutes. \KATch{}-naive exceeds the timeout on Kentucky Data Link---using a quadratic number of queries to check full reachability produces over 500k individual point-to-point queries in a network with 754 nodes! (However, \KATch{}-naive is able to finish it in just under 20 minutes.)

To avoid combinatorial blowup in the verification query itself, we also used \KATch{}'s high-level verification interface, to check full reachability using a \emph{linear} number of queries, as discussed in \Cref{ex:linear-reachability}. \Frenetic{} does not have an analogous linear mode. \APK{} does have an analogous mode, as it has specialized support for full reachability queries. \APK{} is faster than \KATch{} for the smaller topologies, but slower for the larger ones. Indeed, the slope of the \APK{} line is steeper than the \KATch{} line, indicating that \KATch{} scales slightly better with the size of the network on these queries.
It is important to note that \APK{} has support for prefix-matching, ACLs, NAT, incrementality, and other features for which \KATch{} does not have specialized support and which are not tested in this comparison.
One should therefore not draw strong conclusions from this comparison; we include it as it is encouraging that reachability via \NetKAT{} equivalence can be competitive with a state-of-the-art specialized tool.

We also randomly sampled a subset of toplogies from the Topology Zoo of varying size and generated point-to-point (1-to-1) reachability and slicing queries to be checked against the routing configurations with no timeout. For each query, we ran \KATch{} and also generated an equivalent query in the syntax of \Frenetic, and ran \Frenetic's bisimulation verifier on those queries. We present a full table of results of these experiments in \Cref{fig:topology-zoo}. The table shows that \KATch{}'s relative speedup over \Frenetic{} is considerable, and increases problem size. This is encouraging, because it shows that \KATch{} is more scalable. \Frenetic{} was unable to complete KDL within a 200GB memory limit (for comparison, \KATch{} is able to complete the same query with well under 1GB).
We also include a selection of full-reachability results from \Cref{fig:topology-zoo-full} in the table.
The reader may wonder why full reachability takes only twice as long as 1-to-1 reachability. This is because a significant fraction of the time is spent constructing automata (which are similar for both queries), and much of the work in the full reachability queries is shared due to memoization of $\SPP$ operations.

\subsection{Combinatorial Examples}

Finally, we ran experiments to test the hypothesis that \SPPs{} have an asymptotic advantage for certain types of queries. \Frenetic{} uses forwarding decision diagrams (\FDDs{}), which is a different (and non-canonical) representation for automaton transitions. A key advantage of \SPPs{} over \FDDs{} is that \SPPs{} keep the updates to each field next to the tests of the same field. On the other hand, \FDDs{} keep all updates at the leaves, which can result in combinatorial explosion during sequencing. To test this, we generated the following NetKAT programs:
\begin{description}
    \item[Inc:] Treating the input packet's $n$ boolean fields as a binary number, increment it by one.
    \item[Flip:] Sequentially flip the value of each of the $n$ boolean fields.
    \item[Nondet:] Set each field of the packet to a range of values from $0$ to $n$.
\end{description}
For Inc, we tested that repeatedly incrementing (using the $\star$
operator) can turn packet $00\cdots0$ into $11\cdots1$.  For Flip, we tested that flipping all bits twice returns the original packet.  For Nondet, we tested that setting the fields non-deterministically twice is the same as doing it once.  The results of this experiment are shown in \Cref{fig:combinatorial}.  Because the available fields are hardcoded in \Frenetic, we only ran the Inc and Flip experiments up to $n=10$. We ran the non-determinism test up to $n=15$.  \KATch{} finishes all three queries up to $n=100$ in under a one minute, demonstrating its asymptotic advantage for these queries, while Frenetic shows combinatorial blowup on larger packets.

\begin{figure}
     \includegraphics[scale=0.53]{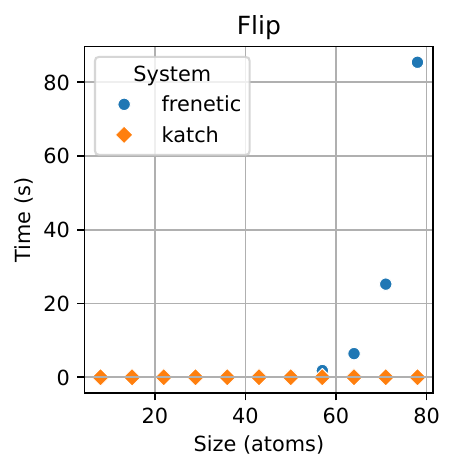}
     \includegraphics[scale=0.53]{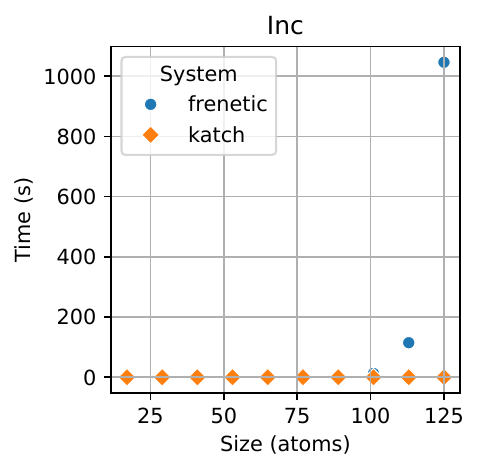}
     \includegraphics[scale=0.53]{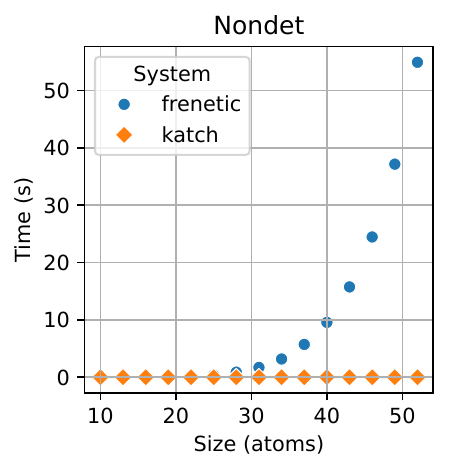}
  \caption{Results of running \KATch{} and \Frenetic on combinatorial benchmarks}\label{fig:combinatorial}
\end{figure}

\paragraph{Source of speedup}
The speedup of \KATch{} over \Frenetic{} comes from several sources: (1) the use of \SPPs{} instead of \FDDs{}, which can be exponentially more compact and support more efficient sequential composition, (2) the symbolic bisimulation algorithm operating on SPs, which can handle exponentially large sets of packets at once, and (3) a quadratic-doubling implementation of star, which can handle exponentially long traces quickly.
These differences show up most strongly in the combinatorially adversarial examples above, but the speedup is also considerable for the Topology Zoo benchmarks, which are based on real-world topologies and not designed to be adversarial.

\paragraph{Field order}
The field order of SPs and SPPs can affect performance, just as for BDDs.
In practice, field orders that keep related fields close together are beneficial.
Our implementation allows the user to control the field order, but we did not use this for our benchmarks.
By default, the fields are in the order in which they first occur in the input file, which turned out to work well enough.

\section{Related Work}

This section discusses the most closely related prior work to this paper, focusing on three areas: \NetKAT, network verification, and automata-theoretic approaches to verification.

\paragraph*{\NetKAT}

\NetKAT was originally proposed as a semantic foundation for SDN data planes~\cite{Anderson2014}. Indeed, being based on KAT~\cite{Kozen1996}, the language provides a sound and complete algebraic reasoning system. Later work on \NetKAT developed an automata-theoretic (or coalgebraic) account of the language~\cite{Foster2015}, including a decision procedure based on Brzozowski derivatives and bisimulation, implemented in \Frenetic{}. However, the performance of this approach turns out to be poor, as shown in our experiments, due to the use of ad hoc data structures (``bases'') to encode packets and automata. \NetKAT's compiler uses a variant of BDDs, called Forwarding Decision Diagrams (FDDs), as well as an algorithm for converting programs to automata using Antimirov derivatives~\cite{Smolka2015}, improving on the earlier representation of transitions as sets of bases~\cite{Foster2015}. The \SPPn{}s proposed in this paper improve on FDDs by ensuring uniqueness and supporting efficient sequential composition in the common case. In addition, the deterministic automata used in \KATch support additional ``negative''
operators that are useful for verification. Other papers based on NetKAT have explored use of the language in other settings such as distributed control planes~\cite{Beckett2016}, probabilistic networks~\cite{Foster2016,Smolka2017,Smolka2019}, and Kleene algebra over composable theories with unbounded state~\cite{Greenberg2022}. It would be interesting to extend the techniques developed in this paper to these richer settings. Another interesting direction for future work is to build a symbolic verifier for the guarded fragment of \NetKAT~\cite{Smolka2019a}.

\paragraph*{Network Verification}
% Katch results
% Missouri & 63.1
% Airtel & 21.2
% Sanet & 28.1
% Shentel & 11.7
% Layer42 & 1.6
% Telcove & 71.1
% Cogentco & 555.6
% Uunet & 37.9
% Deltacom & 192.2
% Compuserve & 6.8
% \begin{wrapfigure}{R}{0.50\textwidth}
%   \begin{tabular}{lrrr}
%      \toprule
%      Network & \multicolumn{2}{c}{\APK} & \KATch \\
%              & Init. & Query &  \\
%      \midrule
% Layer42 & 248.54 & 1.32 & 1.97\\
% Compuserve & 178.17 & 5.82 & 7.82\\
% Airtel & 86.24 & 6.76 & 9.84\\
% Shentel & 164.52 & 21.74 & 18.84\\
% Sanet & 91.61 & 60.83 & 42.11\\
% Uunet & 131.83 & 62.15 & 60.17\\
% Telcove & 77.38 & 118.11 & 109.34\\
% Missouri & 89.21 & 169.14 & 535.23\\
% Deltacom & 122.84 & 546.11 & 324.67\\
% Cogentco & 105.59 & 2710.65 & 648.66\\
% Kdl & 722.86 & 92776.23 & 16078.15\\
%      \bottomrule
%   \end{tabular}
%   \caption{Comparison with \APK{} on full reachability. Times in milliseconds, averaged over 100 runs, after 10 runs of JIT warmup.}
%   \label{tab:apk-comparison}
% \end{wrapfigure}
Early work by Xie et al.~\cite{Xie2005} proposed a unifying mathematical model for Internet routers and developed algorithms for analyzing network-wide reachability. Although the paper did not discuss an implementation, its elegant formal model has been extremely influential in the community and has served as the foundation for many follow-on efforts, including this work. The emergence of software-defined networking (SDN) led to a surge of interest in static data plane verification, including systems such as Header Space Analysis (HSA)~\cite{Kazemian2012}, Anteater~\cite{Mai2011}, VeriFlow~\cite{Khurshid2012}, Atomic Predicates (AP)~\cite{Yang2016}, APKeep~\cite{Zhang2020}. These systems all follow a common approach: they build a model of the network-wide forwarding behavior and then check whether given properties hold. However, they vary in the data structures and algorithms used to represent and analyze the network model. For instance, Anteater relies on first-order logic and SAT solvers, while VeriFlow uses prefix trees and custom graph-based algorithms. HSA, AP and APKeep are arguably the most related to our work as they use symbolic representations and BDDs respectively.

The primary difference between \KATch{} and these systems is that the latter implement specialized algorithms for network-wide analysis queries, while \KATch{} (and \Frenetic{}) solve the more general \NetKAT equivalence problem, into which network-wide analysis queries can be encoded.
This generic approach is based on principled automata-theoretic foundations, but one might expect it to be less efficient than specialized algorithms.
\Frenetic{} was found to be comparable to or faster than HSA \cite{Foster2015}, so by transitivity we expect \KATch{} to be faster than HSA.
However, the current state of the art is \APK{}, which is significantly faster than \Frenetic{} and HSA.
We include a comparison with \APK{} in \Cref{sec:eval} on reachability queries.
The results show that \KATch{} is competitive with \APK{} on these benchmarks, despite solving general NetKAT queries. As discussed, one should not draw strong conclusions from this comparison as \APK{} includes additional features; nevertheless, it offers a promising preliminary indication of NetKAT's scalability and performance.

% This generality comes at a cost, as the specialized algorithms can be more efficient for specific queries. For example, \APK{} is a state-of-the-art tool for network-wide reachability queries, and we include a comparison with \APK{} in \Cref{tab:apk-comparison}.

% However, in contrast to \Frenetic{} and \KATch{}, they rely on specialized algorithms for network-wide analysis, which differs from NetKAT's more .
% We therefore focused our evaluation on \Frenetic{}, because this is the most apples-to-apples comparison, as \Frenetic{} also solves the more general problem of \NetKAT{} equivalence, unlike these other tools.
%  We believe that \APK{} is the current state of the art, and include a comparison in \Cref{sec:eval}.

% For the interested reader, we provide a brief comparison of \KATch{} with \APK{} on full reachability queries in \Cref{tab:apk-comparison} (\APK{} initialization takes $70$-$700$ms, depending on the topology; this time is not included). On these benchmarks, \KATch{} is competitive with \APK{}, despite solving general NetKAT queries, not just reachability. It is important to note that \APK{} has support for prefix-matching, ACLs, NAT, incrementality, and other features for which \KATch{} does not have specialized support and which are not tested in this comparison.
% One should therefore not draw strong conclusions from this comparison; we include it as it is encouraging that reachability via equivalence can be competitive with a state-of-the-art tool.

Another line of work has explored how to lift verification from the data plane to the control plane. Batfish~\cite{Fogel2015,Brown2023}
proposed using symbolic simulation to analyze distributed control planes---i.e., generating all possible data planes that might be produced starting from a given control-plane configuration. Like \KATch, Batfish uses a BDD-based representation for data plane analysis. MineSweeper~\cite{Beckett2017} improves on Batfish using an SMT encoding of the converged states of the control plane that avoids having to explicitly simulate the underlying routing protocols. Recent work has focused on using techniques like modular reasoning~\cite{Tang2023,Thijm2023} abstract interpretation~\cite{Beckett2020}, and a form of symmetry reduction~\cite{Beckett2018} to further improve the scalability of control-plane verification.

\paragraph*{Automata-Theoretic Approach and Symbolic Automata}
Our work on \KATch builds on the large body of work on the automata-theoretic approach to verification and symbolic automata. The automata-theoretic approach was pioneered in the 1980s, with applications of temporal logics and model checking to hardware verification~\cite{Vardi1986}. BDDs, originally proposed by Lee \cite{Lee1959}, were further developed by Bryant~\cite{Bryant1986}, and used by McMillan for symbolic model checking~\cite{Burch1990}. BDD-based techniques were a success story of symbolic model checking of hardware in the 1990s---Chaki and Gurfinkel give an overview \cite{Chaki2018}.

An influential line of work by D'Antoni and Veanus developed techniques for representing and transforming finite automata where the transitions are not labeled with individual characters but with elements of a so-called effective Boolean algebra~\cite{DAntoni2014,DAntoni2017}. Shifting from a concrete to a symbolic representation requires generalizing classical algorithms, such as minimization and equivalence, but facilitates building automata that work over huge alphabets, such as Unicode.

Pous \cite{Pous2015} developed symbolic techniques for checking the equivalence of automata where transitions are specified using BDDs. The methods developed in this paper take Pous's work as a starting point but develop a non-trivial extension for \NetKAT. In particular, \SPPn{}s provide a compact representation for ``carry-on'' packets, which is a critical and unique aspect of \NetKAT's semantics. Bonchi and Pous \cite{Bonchi2013} explored the use of up-to techniques for checking equivalence of automata. In principle, one can view the invariants enforced by \KATch's term representations as a kind of up-to technique, but a full investigation of this idea requires more work.

Our backward algorithm for computing bisimulations can be seen as a variant of Moore's classic algorithm for computing the greatest bisimulation of classic automata.  Doenges et al.~\cite{Doenges2022}
proposed an analogous approach for checking equivalence of automata that model the behavior of P4 packet parsers~\cite{Bosshart2014}. However, Leapfrog's model is simpler than \NetKAT's and is based on classic finite automata. It achieves scalability due to a novel up-to technique that ``leaps''
over internal buffering transitions rather than symbolic representations.

\section*{Acknowledgements}

We wish to thank the PLDI reviewers for many helpful comments and suggestions, which improved the paper significantly. We are also grateful to Aaron Gember-Jacobson and Peng Zhang for making AP Keep available as open-source software, to Caleb Koch for early discussions on compact data structures for data plane verification, to the Cornell PLDG for helpful comments on early drafts, and to the EPFL DCSL group for providing a welcoming and supportive environment. This work was supported in part by ONR grant N68335-22-C-0411, DARPA grant W912CG-23-C-0032, ERC grant 101002697, and a Royal Society Wolfson fellowship.

\section*{Data Availability}

The current version of the implementation is available at \url{https://github.com/cornell-netlab/KATch}.  A snapshot of the code as of artifact evaluation (with a dependencies-included Docker image) is available on Zenodo \cite{Moeller2024Artifact}.

\bibliographystyle{ACM-Reference-Format}
\bibliography{refs}


\begin{thebibliography}{46}

%%% the Web address.

\ifx \showCODEN    \undefined \def \showCODEN     #1{\unskip}     \fi
\ifx \showDOI      \undefined \def \showDOI       #1{#1}\fi
\ifx \showISBNx    \undefined \def \showISBNx     #1{\unskip}     \fi
\ifx \showISBNxiii \undefined \def \showISBNxiii  #1{\unskip}     \fi
\ifx \showISSN     \undefined \def \showISSN      #1{\unskip}     \fi
\ifx \showLCCN     \undefined \def \showLCCN      #1{\unskip}     \fi
\ifx \shownote     \undefined \def \shownote      #1{#1}          \fi
\ifx \showarticletitle \undefined \def \showarticletitle #1{#1}   \fi
\ifx \showURL      \undefined \def \showURL       {\relax}        \fi
% invisible to TeX
\providecommand\bibfield[2]{#2}
\providecommand\bibinfo[2]{#2}
\providecommand\natexlab[1]{#1}
\providecommand\showeprint[2][]{arXiv:#2}

\bibitem[Anderson et~al\mbox{.}(2014)]%
        {Anderson2014}
\bibfield{author}{\bibinfo{person}{Carolyn~Jane Anderson},
  \bibinfo{person}{Nate Foster}, \bibinfo{person}{Arjun Guha},
  \bibinfo{person}{Jean-Baptiste Jeannin}, \bibinfo{person}{Dexter Kozen},
  \bibinfo{person}{Cole Schlesinger}, {and} \bibinfo{person}{David Walker}.}
  \bibinfo{year}{2014}\natexlab{}.
\newblock \showarticletitle{NetKAT: Semantic Foundations for Networks}. In
  \bibinfo{booktitle}{\emph{POPL}}.
\newblock
\urldef\tempurl%
\url{https://doi.org/10.1145/2535838.2535862}
\showDOI{\tempurl}


\bibitem[Antimirov(1996)]%
        {Antimirov1996}
\bibfield{author}{\bibinfo{person}{Valentin Antimirov}.}
  \bibinfo{year}{1996}\natexlab{}.
\newblock \showarticletitle{Partial derivatives of regular expressions and
  finite automaton constructions}.
\newblock \bibinfo{journal}{\emph{Theoretical Computer Science}}
  \bibinfo{volume}{155}, \bibinfo{number}{2} (\bibinfo{date}{Mar}
  \bibinfo{year}{1996}), \bibinfo{pages}{291--319}.
\newblock
\showISSN{0304-3975}
\urldef\tempurl%
\url{https://doi.org/10.1016/0304-3975(95)00182-4}
\showDOI{\tempurl}


\bibitem[Barnett et~al\mbox{.}(2005)]%
        {Barnett2005}
\bibfield{author}{\bibinfo{person}{Michael Barnett},
  \bibinfo{person}{Bor{-}Yuh~Evan Chang}, \bibinfo{person}{Robert DeLine},
  \bibinfo{person}{Bart Jacobs}, {and} \bibinfo{person}{K.~Rustan~M. Leino}.}
  \bibinfo{year}{2005}\natexlab{}.
\newblock \showarticletitle{Boogie: {A} Modular Reusable Verifier for
  Object-Oriented Programs}. In \bibinfo{booktitle}{\emph{FMCO}}.
\newblock
\urldef\tempurl%
\url{https://doi.org/10.1007/11804192\_17}
\showDOI{\tempurl}


\bibitem[Beckett et~al\mbox{.}(2017)]%
        {Beckett2017}
\bibfield{author}{\bibinfo{person}{Ryan Beckett}, \bibinfo{person}{Aarti
  Gupta}, \bibinfo{person}{Ratul Mahajan}, {and} \bibinfo{person}{David
  Walker}.} \bibinfo{year}{2017}\natexlab{}.
\newblock \showarticletitle{A General Approach to Network Configuration
  Verification}. In \bibinfo{booktitle}{\emph{SIGCOMM}}.
\newblock
\urldef\tempurl%
\url{https://doi.org/10.1145/3098822.3098834}
\showDOI{\tempurl}


\bibitem[Beckett et~al\mbox{.}(2018)]%
        {Beckett2018}
\bibfield{author}{\bibinfo{person}{Ryan Beckett}, \bibinfo{person}{Aarti
  Gupta}, \bibinfo{person}{Ratul Mahajan}, {and} \bibinfo{person}{David
  Walker}.} \bibinfo{year}{2018}\natexlab{}.
\newblock \showarticletitle{Control Plane Compression}. In
  \bibinfo{booktitle}{\emph{SIGCOMM}}.
\newblock
\urldef\tempurl%
\url{https://doi.org/10.1145/3230543.3230583}
\showDOI{\tempurl}


\bibitem[Beckett et~al\mbox{.}(2020)]%
        {Beckett2020}
\bibfield{author}{\bibinfo{person}{Ryan Beckett}, \bibinfo{person}{Aarti
  Gupta}, \bibinfo{person}{Ratul Mahajan}, {and} \bibinfo{person}{David
  Walker}.} \bibinfo{year}{2020}\natexlab{}.
\newblock \showarticletitle{Abstract Interpretation of Distributed Network
  Control Planes}. In \bibinfo{booktitle}{\emph{POPL}}.
\newblock
\urldef\tempurl%
\url{https://doi.org/10.1145/3371110}
\showDOI{\tempurl}


\bibitem[Beckett et~al\mbox{.}(2016)]%
        {Beckett2016}
\bibfield{author}{\bibinfo{person}{Ryan Beckett}, \bibinfo{person}{Ratul
  Mahajan}, \bibinfo{person}{Todd Millstein}, \bibinfo{person}{Jitendra
  Padhye}, {and} \bibinfo{person}{David Walker}.}
  \bibinfo{year}{2016}\natexlab{}.
\newblock \showarticletitle{Don't Mind the Gap: Bridging Network-Wide
  Objectives and Device-Level Configurations}. In
  \bibinfo{booktitle}{\emph{SIGCOMM}}.
\newblock
\urldef\tempurl%
\url{https://doi.org/10.1145/2934872.2934909}
\showDOI{\tempurl}


\bibitem[Bonchi and Pous(2013)]%
        {Bonchi2013}
\bibfield{author}{\bibinfo{person}{Filippo Bonchi} {and}
  \bibinfo{person}{Damien Pous}.} \bibinfo{year}{2013}\natexlab{}.
\newblock \showarticletitle{Checking NFA Equivalence with Bisimulations up to
  Congruence}. In \bibinfo{booktitle}{\emph{POPL}}.
\newblock
\urldef\tempurl%
\url{https://doi.org/10.1145/2429069.2429124}
\showDOI{\tempurl}


\bibitem[Bosshart et~al\mbox{.}(2014)]%
        {Bosshart2014}
\bibfield{author}{\bibinfo{person}{Pat Bosshart}, \bibinfo{person}{Dan Daly},
  \bibinfo{person}{Glen Gibb}, \bibinfo{person}{Martin Izzard},
  \bibinfo{person}{Nick McKeown}, \bibinfo{person}{Jennifer Rexford},
  \bibinfo{person}{Cole Schlesinger}, \bibinfo{person}{Dan Talayco},
  \bibinfo{person}{Amin Vahdat}, \bibinfo{person}{George Varghese}, {and}
  \bibinfo{person}{David Walker}.} \bibinfo{year}{2014}\natexlab{}.
\newblock \showarticletitle{P4: Programming Protocol-Independent Packet
  Processors}.
\newblock \bibinfo{journal}{\emph{SIGCOMM Computer Communication Review}}
  \bibinfo{volume}{44}, \bibinfo{number}{3} (\bibinfo{date}{Jul}
  \bibinfo{year}{2014}), \bibinfo{pages}{87–95}.
\newblock
\showISSN{0146-4833}
\urldef\tempurl%
\url{https://doi.org/10.1145/2656877.2656890}
\showDOI{\tempurl}


\bibitem[Brown et~al\mbox{.}(2023)]%
        {Brown2023}
\bibfield{author}{\bibinfo{person}{Matt Brown}, \bibinfo{person}{Ari Fogel},
  \bibinfo{person}{Daniel Halperin}, \bibinfo{person}{Victor Heorhiadi},
  \bibinfo{person}{Ratul Mahajan}, {and} \bibinfo{person}{Todd~D. Millstein}.}
  \bibinfo{year}{2023}\natexlab{}.
\newblock \showarticletitle{Lessons from the evolution of the Batfish
  configuration analysis tool}. In \bibinfo{booktitle}{\emph{SIGCOMM}}.
\newblock
\urldef\tempurl%
\url{https://doi.org/10.1145/3603269.3604866}
\showDOI{\tempurl}


\bibitem[Bryant(1986)]%
        {Bryant1986}
\bibfield{author}{\bibinfo{person}{Randal~E. Bryant}.}
  \bibinfo{year}{1986}\natexlab{}.
\newblock \showarticletitle{Graph-Based Algorithms for Boolean Function
  Manipulation}.
\newblock \bibinfo{journal}{\emph{IEEE Trans. Comput.}} \bibinfo{volume}{C-35},
  \bibinfo{number}{8} (\bibinfo{date}{Aug} \bibinfo{year}{1986}),
  \bibinfo{pages}{677--691}.
\newblock
\urldef\tempurl%
\url{https://doi.org/10.1109/TC.1986.1676819}
\showDOI{\tempurl}


\bibitem[Bryant(1992)]%
        {Bryant1992}
\bibfield{author}{\bibinfo{person}{Randal~E. Bryant}.}
  \bibinfo{year}{1992}\natexlab{}.
\newblock \showarticletitle{Symbolic Boolean Manipulation with Ordered
  Binary-Decision Diagrams}.
\newblock \bibinfo{journal}{\emph{Comput. Surveys}} \bibinfo{volume}{24},
  \bibinfo{number}{3} (\bibinfo{date}{Sep} \bibinfo{year}{1992}),
  \bibinfo{pages}{293–318}.
\newblock
\showISSN{0360-0300}
\urldef\tempurl%
\url{https://doi.org/10.1145/136035.136043}
\showDOI{\tempurl}


\bibitem[Brzozowski(1962)]%
        {Brzozowski1962}
\bibfield{author}{\bibinfo{person}{Janusz~A. Brzozowski}.}
  \bibinfo{year}{1962}\natexlab{}.
\newblock \showarticletitle{Canonical regular expressions and minimal state
  graphs for definite events}. In \bibinfo{booktitle}{\emph{Proceedings of the
  Symposium of Mathematical Theory of Automata}}.
\newblock


\bibitem[Burch et~al\mbox{.}(1990)]%
        {Burch1990}
\bibfield{author}{\bibinfo{person}{Jerry~R. Burch}, \bibinfo{person}{Edmund~M.
  Clarke}, \bibinfo{person}{Kenneth~L. McMillan}, \bibinfo{person}{David~L.
  Dill}, {and} \bibinfo{person}{L.~J. Hwang}.} \bibinfo{year}{1990}\natexlab{}.
\newblock \showarticletitle{Symbolic Model Checking: 10{\^{}}20 States and
  Beyond}. In \bibinfo{booktitle}{\emph{LICS}}.
\newblock
\urldef\tempurl%
\url{https://doi.org/10.1109/LICS.1990.113767}
\showDOI{\tempurl}


\bibitem[Chaki and Gurfinkel(2018)]%
        {Chaki2018}
\bibfield{author}{\bibinfo{person}{Sagar Chaki} {and} \bibinfo{person}{Arie
  Gurfinkel}.} \bibinfo{year}{2018}\natexlab{}.
\newblock \showarticletitle{BDD-Based Symbolic Model Checking}.
\newblock In \bibinfo{booktitle}{\emph{Handbook of Model Checking}},
  \bibfield{editor}{\bibinfo{person}{Edmund~M. Clarke},
  \bibinfo{person}{Thomas~A. Henzinger}, \bibinfo{person}{Helmut Veith}, {and}
  \bibinfo{person}{Roderick Bloem}} (Eds.). \bibinfo{publisher}{Springer
  International Publishing}, \bibinfo{pages}{219--245}.
\newblock
\showISBNx{978-3-319-10575-8}
\urldef\tempurl%
\url{https://doi.org/10.1007/978-3-319-10575-8_8}
\showDOI{\tempurl}


\bibitem[D'Antoni and Veanes(2014)]%
        {DAntoni2014}
\bibfield{author}{\bibinfo{person}{Loris D'Antoni} {and}
  \bibinfo{person}{Margus Veanes}.} \bibinfo{year}{2014}\natexlab{}.
\newblock \showarticletitle{Minimization of Symbolic Automata}. In
  \bibinfo{booktitle}{\emph{POPL}}.
\newblock
\urldef\tempurl%
\url{https://doi.org/10.1145/2535838.2535849}
\showDOI{\tempurl}


\bibitem[D'Antoni and Veanes(2017)]%
        {DAntoni2017}
\bibfield{author}{\bibinfo{person}{Loris D'Antoni} {and}
  \bibinfo{person}{Margus Veanes}.} \bibinfo{year}{2017}\natexlab{}.
\newblock \showarticletitle{Forward Bisimulations for Nondeterministic Symbolic
  Finite Automata}. In \bibinfo{booktitle}{\emph{TACAS}}.
\newblock
\urldef\tempurl%
\url{https://doi.org/10.1007/978-3-662-54577-5_30}
\showDOI{\tempurl}


\bibitem[Doenges et~al\mbox{.}(2022)]%
        {Doenges2022}
\bibfield{author}{\bibinfo{person}{Ryan Doenges}, \bibinfo{person}{Tobias
  Kapp\'{e}}, \bibinfo{person}{John Sarracino}, \bibinfo{person}{Nate Foster},
  {and} \bibinfo{person}{Greg Morrisett}.} \bibinfo{year}{2022}\natexlab{}.
\newblock \showarticletitle{Leapfrog: Certified Equivalence for Protocol
  Parsers}. In \bibinfo{booktitle}{\emph{PLDI}}.
\newblock
\urldef\tempurl%
\url{https://doi.org/10.1145/3519939.3523715}
\showDOI{\tempurl}


\bibitem[Fogel et~al\mbox{.}(2015)]%
        {Fogel2015}
\bibfield{author}{\bibinfo{person}{Ari Fogel}, \bibinfo{person}{Stanley Fung},
  \bibinfo{person}{Luis Pedrosa}, \bibinfo{person}{Meg Walraed-Sullivan},
  \bibinfo{person}{Ramesh Govindan}, \bibinfo{person}{Ratul Mahajan}, {and}
  \bibinfo{person}{Todd Millstein}.} \bibinfo{year}{2015}\natexlab{}.
\newblock \showarticletitle{A General Approach to Network Configuration
  Analysis}. In \bibinfo{booktitle}{\emph{NSDI}}.
\newblock


\bibitem[Foster et~al\mbox{.}(2011)]%
        {Foster2011}
\bibfield{author}{\bibinfo{person}{Nate Foster}, \bibinfo{person}{Rob
  Harrison}, \bibinfo{person}{Michael~J. Freedman},
  \bibinfo{person}{Christopher Monsanto}, \bibinfo{person}{Jennifer Rexford},
  \bibinfo{person}{Alec Story}, {and} \bibinfo{person}{David Walker}.}
  \bibinfo{year}{2011}\natexlab{}.
\newblock \showarticletitle{Frenetic: {A} Network Programming Language}. In
  \bibinfo{booktitle}{\emph{ICFP}}.
\newblock
\urldef\tempurl%
\url{https://doi.org/10.1145/2034773.2034812}
\showDOI{\tempurl}


\bibitem[Foster et~al\mbox{.}(2016)]%
        {Foster2016}
\bibfield{author}{\bibinfo{person}{Nate Foster}, \bibinfo{person}{Dexter
  Kozen}, \bibinfo{person}{Konstantinos Mamouras}, \bibinfo{person}{Mark
  Reitblatt}, {and} \bibinfo{person}{Alexandra Silva}.}
  \bibinfo{year}{2016}\natexlab{}.
\newblock \showarticletitle{Probabilistic NetKAT}. In
  \bibinfo{booktitle}{\emph{ESOP}}.
\newblock
\urldef\tempurl%
\url{https://doi.org/10.1007/978-3-662-49498-1_12}
\showDOI{\tempurl}


\bibitem[Foster et~al\mbox{.}(2015)]%
        {Foster2015}
\bibfield{author}{\bibinfo{person}{Nate Foster}, \bibinfo{person}{Dexter
  Kozen}, \bibinfo{person}{Mae Milano}, \bibinfo{person}{Alexandra Silva},
  {and} \bibinfo{person}{Laure Thompson}.} \bibinfo{year}{2015}\natexlab{}.
\newblock \showarticletitle{A Coalgebraic Decision Procedure for NetKAT}. In
  \bibinfo{booktitle}{\emph{POPL}}.
\newblock
\urldef\tempurl%
\url{https://doi.org/10.1145/2676726.2677011}
\showDOI{\tempurl}


\bibitem[Greenberg et~al\mbox{.}(2022)]%
        {Greenberg2022}
\bibfield{author}{\bibinfo{person}{Michael Greenberg}, \bibinfo{person}{Ryan
  Beckett}, {and} \bibinfo{person}{Eric Campbell}.}
  \bibinfo{year}{2022}\natexlab{}.
\newblock \showarticletitle{Kleene Algebra Modulo Theories: A Framework for
  Concrete KATs}. In \bibinfo{booktitle}{\emph{PLDI}}.
\newblock
\urldef\tempurl%
\url{https://doi.org/10.1145/3519939.3523722}
\showDOI{\tempurl}


\bibitem[Hooimeijer et~al\mbox{.}(2011)]%
        {Hooimeijer2011}
\bibfield{author}{\bibinfo{person}{Pieter Hooimeijer},
  \bibinfo{person}{Benjamin Livshits}, \bibinfo{person}{David Molnar},
  \bibinfo{person}{Prateek Saxena}, {and} \bibinfo{person}{Margus Veanes}.}
  \bibinfo{year}{2011}\natexlab{}.
\newblock \showarticletitle{Fast and Precise Sanitizer Analysis with BEK}. In
  \bibinfo{booktitle}{\emph{USENIX Conference on Security}}.
\newblock


\bibitem[Hopcroft and Karp(1971)]%
        {Hopcroft1971}
\bibfield{author}{\bibinfo{person}{John~E. Hopcroft} {and}
  \bibinfo{person}{Richard~M. Karp}.} \bibinfo{year}{1971}\natexlab{}.
\newblock \showarticletitle{A Linear Algorithm for Testing Equivalence of
  Finite Automata}.
\newblock


\bibitem[Kazemian et~al\mbox{.}(2012)]%
        {Kazemian2012}
\bibfield{author}{\bibinfo{person}{Peyman Kazemian}, \bibinfo{person}{George
  Varghese}, {and} \bibinfo{person}{Nick McKeown}.}
  \bibinfo{year}{2012}\natexlab{}.
\newblock \showarticletitle{Header Space Analysis: Static Checking for
  Networks}. In \bibinfo{booktitle}{\emph{NSDI}}.
\newblock


\bibitem[Khurshid et~al\mbox{.}(2012)]%
        {Khurshid2012}
\bibfield{author}{\bibinfo{person}{Ahmed Khurshid}, \bibinfo{person}{Wenxuan
  Zhou}, \bibinfo{person}{Matthew Caesar}, {and} \bibinfo{person}{P.~Brighten
  Godfrey}.} \bibinfo{year}{2012}\natexlab{}.
\newblock \showarticletitle{Veriflow: Verifying Network-Wide Invariants in Real
  Time}.
\newblock \bibinfo{journal}{\emph{SIGCOMM Computer Communications Review}}
  \bibinfo{volume}{42}, \bibinfo{number}{4} (\bibinfo{date}{Sep}
  \bibinfo{year}{2012}), \bibinfo{pages}{467–472}.
\newblock
\urldef\tempurl%
\url{https://doi.org/10.1145/2377677.2377766}
\showDOI{\tempurl}


\bibitem[Knight et~al\mbox{.}(2011)]%
        {Knight2011}
\bibfield{author}{\bibinfo{person}{Simon Knight}, \bibinfo{person}{Hung~X.
  Nguyen}, \bibinfo{person}{Nickolas Falkner}, \bibinfo{person}{Rhys Bowden},
  {and} \bibinfo{person}{Matthew Roughan}.} \bibinfo{year}{2011}\natexlab{}.
\newblock \showarticletitle{The Internet Topology Zoo}.
\newblock \bibinfo{journal}{\emph{IEEE Journal on Selected Areas in
  Communications}} \bibinfo{volume}{29}, \bibinfo{number}{9}
  (\bibinfo{date}{Oct} \bibinfo{year}{2011}), \bibinfo{pages}{1765 --1775}.
\newblock
\urldef\tempurl%
\url{https://doi.org/10.1109/JSAC.2011.111002}
\showDOI{\tempurl}


\bibitem[Kozen(1996)]%
        {Kozen1996}
\bibfield{author}{\bibinfo{person}{Dexter Kozen}.}
  \bibinfo{year}{1996}\natexlab{}.
\newblock \showarticletitle{Kleene Algebra with Tests and Commutativity
  Conditions}. In \bibinfo{booktitle}{\emph{TACAS}}.
\newblock
\urldef\tempurl%
\url{https://doi.org/10.1007/3-540-61042-1_35}
\showDOI{\tempurl}


\bibitem[Lee(1959)]%
        {Lee1959}
\bibfield{author}{\bibinfo{person}{C.~Y. Lee}.}
  \bibinfo{year}{1959}\natexlab{}.
\newblock \showarticletitle{Representation of switching circuits by
  binary-decision programs}.
\newblock \bibinfo{journal}{\emph{The Bell System Technical Journal}}
  \bibinfo{volume}{38}, \bibinfo{number}{4} (\bibinfo{date}{Jul}
  \bibinfo{year}{1959}), \bibinfo{pages}{985--999}.
\newblock
\urldef\tempurl%
\url{https://doi.org/10.1002/j.1538-7305.1959.tb01585.x}
\showDOI{\tempurl}


\bibitem[Leino and W{\"{u}}stholz(2014)]%
        {Leino2014}
\bibfield{author}{\bibinfo{person}{K.~Rustan~M. Leino} {and}
  \bibinfo{person}{Valentin W{\"{u}}stholz}.} \bibinfo{year}{2014}\natexlab{}.
\newblock \showarticletitle{The Dafny Integrated Development Environment}. In
  \bibinfo{booktitle}{\emph{F-IDE}}.
\newblock
\urldef\tempurl%
\url{https://doi.org/10.4204/EPTCS.149.2}
\showDOI{\tempurl}


\bibitem[Mai et~al\mbox{.}(2011)]%
        {Mai2011}
\bibfield{author}{\bibinfo{person}{Haohui Mai}, \bibinfo{person}{Ahmed
  Khurshid}, \bibinfo{person}{Rachit Agarwal}, \bibinfo{person}{Matthew
  Caesar}, \bibinfo{person}{P.~Brighten Godfrey}, {and}
  \bibinfo{person}{Samuel~Talmadge King}.} \bibinfo{year}{2011}\natexlab{}.
\newblock \showarticletitle{Debugging the Data Plane with Anteater}.
\newblock \bibinfo{journal}{\emph{SIGCOMM Computer Communications Review}}
  \bibinfo{volume}{41}, \bibinfo{number}{4} (\bibinfo{date}{Aug}
  \bibinfo{year}{2011}), \bibinfo{pages}{290–301}.
\newblock
\showISSN{0146-4833}
\urldef\tempurl%
\url{https://doi.org/10.1145/2043164.2018470}
\showDOI{\tempurl}


\bibitem[Moeller et~al\mbox{.}(2024)]%
        {Moeller2024}
\bibfield{author}{\bibinfo{person}{Mark Moeller}, \bibinfo{person}{Jules
  Jacobs}, \bibinfo{person}{Olivier~Savary Belanger}, \bibinfo{person}{David
  Darais}, \bibinfo{person}{Cole Schlesinger}, \bibinfo{person}{Steffen
  Smolka}, \bibinfo{person}{Nate Foster}, {and} \bibinfo{person}{Alexandra
  Silva}.} \bibinfo{year}{2024}\natexlab{}.
\newblock \bibinfo{title}{KATch: A Fast Symbolic Verifier for NetKAT}.
\newblock
\newblock
\showeprint[arxiv]{2404.04760}~[cs.PL]
\urldef\tempurl%
\url{https://arxiv.org/pdf/2404.04760.pdf}
\showURL{%
\tempurl}


\bibitem[Moore(1956)]%
        {Moore1956}
\bibfield{author}{\bibinfo{person}{Edward~F. Moore}.}
  \bibinfo{year}{1956}\natexlab{}.
\newblock \showarticletitle{Gedanken-Experiments on Sequential Machines}.
\newblock In \bibinfo{booktitle}{\emph{Automata Studies. (AM-34), Volume 34}}.
\newblock
\urldef\tempurl%
\url{https://doi.org/10.1515/9781400882618-006}
\showDOI{\tempurl}


\bibitem[Pous(2015)]%
        {Pous2015}
\bibfield{author}{\bibinfo{person}{Damien Pous}.}
  \bibinfo{year}{2015}\natexlab{}.
\newblock \showarticletitle{Symbolic Algorithms for Language Equivalence and
  Kleene Algebra with Tests}. In \bibinfo{booktitle}{\emph{POPL}}.
\newblock
\urldef\tempurl%
\url{https://doi.org/10.1145/2775051.2677007}
\showDOI{\tempurl}


\bibitem[Smolka et~al\mbox{.}(2015)]%
        {Smolka2015}
\bibfield{author}{\bibinfo{person}{Steffen Smolka}, \bibinfo{person}{Spiridon
  Eliopoulos}, \bibinfo{person}{Nate Foster}, {and} \bibinfo{person}{Arjun
  Guha}.} \bibinfo{year}{2015}\natexlab{}.
\newblock \showarticletitle{A Fast Compiler for NetKAT}. In
  \bibinfo{booktitle}{\emph{ICFP}}.
\newblock
\urldef\tempurl%
\url{https://doi.org/10.1145/2784731.2784761}
\showDOI{\tempurl}


\bibitem[Smolka et~al\mbox{.}(2019a)]%
        {Smolka2019a}
\bibfield{author}{\bibinfo{person}{Steffen Smolka}, \bibinfo{person}{Nate
  Foster}, \bibinfo{person}{Justin Hsu}, \bibinfo{person}{Tobias Kapp\'{e}},
  \bibinfo{person}{Dexter Kozen}, {and} \bibinfo{person}{Alexandra Silva}.}
  \bibinfo{year}{2019}\natexlab{a}.
\newblock \showarticletitle{Guarded Kleene Algebra with Tests: Verification of
  Uninterpreted Programs in Nearly Linear Time}. In
  \bibinfo{booktitle}{\emph{POPL}}.
\newblock
\urldef\tempurl%
\url{https://doi.org/10.1145/3371129}
\showDOI{\tempurl}


\bibitem[Smolka et~al\mbox{.}(2019b)]%
        {Smolka2019}
\bibfield{author}{\bibinfo{person}{Steffen Smolka}, \bibinfo{person}{Praveen
  Kumar}, \bibinfo{person}{Nate Foster}, \bibinfo{person}{Justin Hsu},
  \bibinfo{person}{Dexter Kozen}, {and} \bibinfo{person}{Alexandra Silva}.}
  \bibinfo{year}{2019}\natexlab{b}.
\newblock \showarticletitle{Scalable Verification of Probabilistic Networks}.
  In \bibinfo{booktitle}{\emph{PLDI}}.
\newblock
\urldef\tempurl%
\url{https://doi.org/10.1145/3314221.3314639}
\showDOI{\tempurl}


\bibitem[Smolka et~al\mbox{.}(2017)]%
        {Smolka2017}
\bibfield{author}{\bibinfo{person}{Steffen Smolka}, \bibinfo{person}{Praveen
  Kumar}, \bibinfo{person}{Nate Foster}, \bibinfo{person}{Dexter Kozen}, {and}
  \bibinfo{person}{Alexandra Silva}.} \bibinfo{year}{2017}\natexlab{}.
\newblock \showarticletitle{Cantor Meets Scott: Semantic Foundations for
  Probabilistic Networks}. In \bibinfo{booktitle}{\emph{POPL}}.
\newblock
\urldef\tempurl%
\url{https://doi.org/10.1145/3093333.3009843}
\showDOI{\tempurl}


\bibitem[Tang et~al\mbox{.}(2023)]%
        {Tang2023}
\bibfield{author}{\bibinfo{person}{Alan Tang}, \bibinfo{person}{Ryan Beckett},
  \bibinfo{person}{Steven Benaloh}, \bibinfo{person}{Karthick Jayaraman},
  \bibinfo{person}{Tejas Patil}, \bibinfo{person}{Todd~D. Millstein}, {and}
  \bibinfo{person}{George Varghese}.} \bibinfo{year}{2023}\natexlab{}.
\newblock \showarticletitle{Lightyear: Using Modularity to Scale {BGP} Control
  Plane Verification}. In \bibinfo{booktitle}{\emph{SIGCOMM}}.
\newblock
\urldef\tempurl%
\url{https://doi.org/10.1145/3603269.3604842}
\showDOI{\tempurl}


\bibitem[Thijm et~al\mbox{.}(2023)]%
        {Thijm2023}
\bibfield{author}{\bibinfo{person}{Timothy~Alberdingk Thijm},
  \bibinfo{person}{Ryan Beckett}, \bibinfo{person}{Aarti Gupta}, {and}
  \bibinfo{person}{David Walker}.} \bibinfo{year}{2023}\natexlab{}.
\newblock \showarticletitle{Modular Control Plane Verification via Temporal
  Invariants}. In \bibinfo{booktitle}{\emph{PLDI}}.
\newblock
\urldef\tempurl%
\url{https://doi.org/10.1145/3591222}
\showDOI{\tempurl}


\bibitem[Torlak and Bod{\'{\i}}k(2013)]%
        {Torlak2013}
\bibfield{author}{\bibinfo{person}{Emina Torlak} {and}
  \bibinfo{person}{Rastislav Bod{\'{\i}}k}.} \bibinfo{year}{2013}\natexlab{}.
\newblock \showarticletitle{Growing Solver-aided Languages with {R}osette}. In
  \bibinfo{booktitle}{\emph{Onward! (SPLASH)}}.
\newblock
\urldef\tempurl%
\url{https://doi.org/10.1145/2509578.2509586}
\showDOI{\tempurl}


\bibitem[Vardi and Wolper(1986)]%
        {Vardi1986}
\bibfield{author}{\bibinfo{person}{Moshe~Y. Vardi} {and}
  \bibinfo{person}{Pierre Wolper}.} \bibinfo{year}{1986}\natexlab{}.
\newblock \showarticletitle{An Automata-Theoretic Approach to Automatic Program
  Verification (Preliminary Report)}. In \bibinfo{booktitle}{\emph{LICS}}.
\newblock


\bibitem[Xie et~al\mbox{.}(2005)]%
        {Xie2005}
\bibfield{author}{\bibinfo{person}{Geoffry~G. Xie}, \bibinfo{person}{Jibin
  Zhan}, \bibinfo{person}{David~A. Maltz}, \bibinfo{person}{Hui Zhang},
  \bibinfo{person}{Albert Greenberg}, \bibinfo{person}{Gisli Hjalmtysson},
  {and} \bibinfo{person}{Jennifer Rexford}.} \bibinfo{year}{2005}\natexlab{}.
\newblock \showarticletitle{On Static Reachability Analysis of {IP} Networks}.
  In \bibinfo{booktitle}{\emph{INFOCOMM}}.
\newblock
\urldef\tempurl%
\url{https://doi.org/10.1109/INFCOM.2005.1498492}
\showDOI{\tempurl}


\bibitem[Yang and Lam(2016)]%
        {Yang2016}
\bibfield{author}{\bibinfo{person}{Hongkun Yang} {and}
  \bibinfo{person}{Simon~S. Lam}.} \bibinfo{year}{2016}\natexlab{}.
\newblock \showarticletitle{Real-Time Verification of Network Properties Using
  Atomic Predicates}.
\newblock \bibinfo{journal}{\emph{IEEE/ACM Transactions on Networking}}
  \bibinfo{volume}{24}, \bibinfo{number}{2} (\bibinfo{date}{Apr}
  \bibinfo{year}{2016}), \bibinfo{pages}{887–900}.
\newblock
\urldef\tempurl%
\url{https://doi.org/10.1109/TNET.2015.2398197}
\showDOI{\tempurl}


\bibitem[Zhang et~al\mbox{.}(2020)]%
        {Zhang2020}
\bibfield{author}{\bibinfo{person}{Peng Zhang}, \bibinfo{person}{Xu Liu},
  \bibinfo{person}{Hongkun Yang}, \bibinfo{person}{Ning Kang},
  \bibinfo{person}{Zhengchang Gu}, {and} \bibinfo{person}{Hao Li}.}
  \bibinfo{year}{2020}\natexlab{}.
\newblock \showarticletitle{{APKeep}: Realtime Verification for Real Networks}.
  In \bibinfo{booktitle}{\emph{NSDI}}.
\newblock


\end{thebibliography}

\ifthenelse{\boolean{isExtendedVersion}}{%
\appendix

\section{Assumptions on Values and Fields}
\subsection{Set of Fields}
We assume packet header fields come from a set $F$, and their values are from a set $V$. Both $F$ and $V$ must have a total order, $\sqsubset$. The implementation makes frequent use of finite maps, especially with $V$ as keys. We use the symbol ($\Mapsto$) to construct finite map types. For example, a concrete packet $\pk$ could naively be represented using the type $F \Mapsto V$, which we take to mean that $\pk = \{f_0 \mapsto v_0, \ldots, f_n \mapsto v_n\}$.  All of our maps have only one binding for each key. We often collect the keys as a set with $\keys(\{f_0 \mapsto v_0, \ldots, f_n \mapsto v_n\}) \triangleq \{f_0,\ldots,f_n\}$.

\subsection{Set of Values}\label{sec:values}
We need one more assumption about the value space: it is larger than the values actually mentioned in any expressions. That we do not assume a fixed finite value space is a boon for compositionality: when \KATch determines that two \NetKAT expressions are equivalent, then they are equivalent under every value set containing the values in the two expressions. On the other hand, if $V$ is small, say $V=\{0,1\}$, then the test $f\test 0$ is semantically equivalent to $f\testNE 1$, but \KATch will not equate these two (unless both are preceded by, e.g. $(f\test 0 + f\test 1)$).

If one wants to reason about a small $V$ using \KATch, one can do any \emph{one} of the following to avoid false negatives:
\begin{enumerate}
\item Not use negative tests.
\item Only check equivalences which do not refer to every value.
\item Prefix both sides of equations with tests representing the full set of values, e.g. for $V=\{0,1\}$, we run the query $(f\test 0 + f\test 1) \cdot e_1 \equiv (f\test 0 + f\test 1) \cdot e_2$. In this case, $f \test 0$ is equivalent to $f\testNE 1$ when appearing in $e_1$ or $e_2$.
\end{enumerate}

\section{Operations on Symbolic Packets}\label{app:sp_ops}

Recall the definition of SP from \Cref{sec:sympk}:
\begin{align*}
  p \in \SP \Coloneqq \bot \mid \top \mid
  \underbrace{\SP(f, \ \{ \ldots, v_i \mapsto q_i, \ldots \},\ q)}_{\equiv \ \ \sum_i f \test v_i \cdot q_i\ +\ (\prod_i f \testNE v_i)\cdot q}
\end{align*}

We take the semantics of an SP to be the semantics of its associated \NetKAT expression, as shown above.
For a packet $\pk\in\Pk$ and $p\in\SP$, note that either $\Sem{p}(\pk) = \emptyset$ or $\Sem{p}(\pk) = \{\pk\}$.
Accordingly, in the notation we treat $\Sem{p}$ as a set of packets and write $\pk\in\Sem{p}$
for $\pk\in\Sem{p}(\pk)$.

We ensure that every SP ever constructed is canonical (i.e., it is Ordered and Reduced). This invariant is established by ensuring that:
\begin{itemize}
    \item The $\SP$ constructor produces canonical-form SPs if its arguments are
        canonical-form.
    \item The other operations only construct SPs using the $\SP$ constructor.
\end{itemize}

\paragraph{Canonicalization}

The smart constructor for SPs has type
\[
  \spsc\colon (F \times (V \Mapsto \SP) \times \SP) \to \SP
\]
and is defined as follows:
\[
  \begin{array}{ll}
  \quad \spsc(f,b,d) \triangleq
    \text{ if } b' = \emptyset \text{ then } d \text{ else } \SP(f, b', d)\\
  \qquad \text{where } b' \triangleq \{v \mapsto p \mid v \mapsto p \in b, p \neq d\}
  \end{array}
\]

\paragraph{Operations}

We have the following 0-ary operations:
\[p \in \SP \Coloneqq \top \mid \bot \mid (f \test v_0) \mid (f \testNE v_0) \]
The $\top$ and $\bot$ are translated directly to the corresponding SP objects.
The $f \test v_0$ and $f \testNE v_0$ are translated as follows:
\begin{align*}
    (f \test v_0) &\triangleq \SP(f, \{v_0 \mapsto \one\}, \zero)\\
    (f \testNE v_0) &\triangleq \SP(f, \{v_0 \mapsto \zero\}, \one)
\end{align*}

Other SPs are built from applying the following operations:
\[
  \begin{array}{cc}
    \nf{+}, \nf{\cdot}, \nf{\cap}, \nf{\oplus}, \nf{-}\ :\ \SP \times \SP \to \SP \qquad
    \nf{\neg}, \nf{\star}\ :\ \SP \to \SP\\
    \NKexists,\NKforall\ :\ F \to \SP \to \SP
  \end{array}
\]

We define the base cases of the binary operations on SPs using the classic truth table definitions:
\[
  \begin{array}{ll}
  \top \plusop \top \triangleq \top \quad \top \cdotop \top \triangleq \top \quad \top \capop \top \triangleq \top \quad \top \oplusop \top \triangleq \bot \quad \top \minusop \top \triangleq \bot\\% \quad \negop \top \triangleq \bot\\
  \top \plusop \bot \triangleq \top \quad \top \cdotop \bot \triangleq \bot \quad \top \capop \bot \triangleq \bot \quad \top \oplusop \bot \triangleq \top \quad \top \minusop \bot \triangleq \top\\% \quad \negop \bot \triangleq \top\\
  \bot \plusop \top \triangleq \top \quad \bot \cdotop \top \triangleq \bot \quad \bot \capop \top \triangleq \bot \quad \bot \oplusop \top \triangleq \top \quad \bot \minusop \top \triangleq \bot\\
  \bot \plusop \bot \triangleq \bot \quad \bot \cdotop \bot \triangleq \bot \quad \bot \capop \bot \triangleq \bot \quad \bot \oplusop \bot \triangleq \bot \quad \bot \minusop \bot \triangleq \bot\\
  \end{array}
\]
The inductive case when the field $f$ matches both SPs can be expressed for all five operations generically, using $\pmop$ to stand for the operation:
\[
  \begin{array}{ll}
  \quad \SP(f, b_p, d_p)\pmop \SP(f, b_q, d_q) \triangleq  \spsc(f,b'\!,d_p \pmop  d_q)\\
    \quad\quad \text{ where }\  b' = \{v \mapsto b_p(v;d_p) \pmop b_q(v;d_q) \mid v\in \keys(b_p \cup b_q)\},\text{ and}\\
    \quad\quad \phantom{\text{ where }}\  b(v;d) \triangleq \text{if }v \in \keys(b) \text{ then } b[v] \text{ else }d\\
  \end{array}
\]
When the fields do not match, we use a rule we call \emph{expansion} to temporarily create an SP with a matching field.
\[
  \begin{array}{ll}
  \quad p \equiv \SP(f, \emptyset, p) \text {\qquad if \qquad} p \in \{\top, \bot, \SP(f', b, d)\} \text{ where } f \sqsubset f'\\
  \end{array}
\]
Thus when we have different fields in the arguments to an operation, we simply apply expansion to the SP whose field is greater, reducing to the inductive case above.

    %% Definition of the $\SP$ operations. The inductive case is identical for all operations (indicated by $\pmop$), and applies when both $\SP$s test the same field. Expansion inserts a trivial $\SP$ node to reduce the remaining cases to the inductive case.
    % The formal definitions are in \Cref{app:sp_ops}.

The remaining operations are defined as follows:
\[
  \begin{array}{ll}
  %% \quad \negop \SP(f, b, d) \triangleq \spsc(f, b', \negop d)\\
    %% \quad\quad \text{ where }\  b' = \{v\mapsto \negop p \mid v\mapsto p \in b\}\\
%%
  %% \quad \negop p \triangleq \one \minusop p\\
  \quad p^{\starop} \triangleq \top\\
  \end{array}
\]

%% Finally the definitions for $\NKforall$ and $\NKexists$ are:
\begin{align*}
\NKforall\ f_1\ p &\triangleq \begin{cases}
        p       &\text{if } p = \bot \vee p = \top\\
        \nf{\bigcap}_{i\leq n+1} p_i
                &\text{if } p = \SP(f_1, \{v_0\mapsto p_0, \ldots, v_n\mapsto p_n\}, p_{n+1})\\
        \spsc(f_2, \{v_i \mapsto \NKforall\ f_1\ p_i \mid i \leq n\}, \NKforall\ f_1\ p_{n+1})
                &\text{if } p = \SP(f_2, \{v_0\mapsto p_0,\ldots, v_n\mapsto p_n\}, p_{n+1}), f_1 \neq f_2\\
  \end{cases}\displaybreak[0]\\
\NKexists\ f_1\ p &\triangleq \begin{cases}
        p     &\text{if } p = \bot \vee p = \top\\
        \nf{\sum}_{i\leq n+1} p_i
              &\text{if } p = \SP(f_1, \{v_0\mapsto p_0,\ldots, v_n\mapsto p_n\}, p_{n+1})\\
        \spsc(f_2, \{v_i \mapsto \NKexists\ f_1\ p_i \mid i \leq n\}, \NKexists\ f_1\ p_{n+1})
              &\text{if } p = \SP(f_2, \{v_0\mapsto p_0,\ldots, v_n\mapsto p_n\}, p_{n+1}), f_1 \neq f_2\\
    \end{cases} \displaybreak[0]\\
\end{align*}
\subsection{Correctness of SP operations}

The goal of this section is to show that our representation of symbolic packets
is canonical and correct, in the sense that two SPs are semantically equivalent
if and only if they are syntactically equal.

First, we formally define a predicate on SPs that they are reduced and ordered (as described in \Cref{sec:sympk}).
Then we will prove the that this property implies they are canonical in the above sense.

\begin{definition}[Reduced and Ordered]
  An SP $p$ is \emph{reduced and ordered for a field $f\in F$}
  if it satisfies:
  \begin{align*}
    &\ROf\colon \SP \to 2\\
    &\ROf (p) \triangleq \begin{cases}
      \top & \text{if }p = \top \vee p = \bot\ \vee\\
      &\phantom{\text{if }}p = \SP(f', \{ v_0 \mapsto p_0, \ldots, v_n \mapsto p_n
      \}, p_{n+1}), \text{ for } n \geq 0\ \wedge \\
           & \phantom{\text{if }p = } f\sqsubset f'\  \wedge\\
           & \phantom{\text{if }p = }
           \forall i \in \{0,\ldots, n+1\}\colon \ROx{f'}(p_i)\ \wedge\\
           & \phantom{\text{if }p = }
           \forall i \in \{0,\ldots, n\}\colon p_i \neq p_{n+1}\\
      \bot & \text{otherwise}
    \end{cases}
  \end{align*}
  Moreover an SP $p$ is \emph{reduced and ordered} if it satisfies:
  \begin{align*}
    &\RO\colon \SP \to 2\\
    &\RO\ (p) \triangleq \exists f \in F\colon \ROf\ (p)
  \end{align*}
\end{definition}

\begin{lemma}\label{lem:const-f}
  For an SP $p$ such that $\ROf(p)$, then for any concrete packet $\pk$, we have
  $\pk\in\pSem{p}$ implies $\forall v\in V\colon \pk[f\mut v]\in\pSem{p}$ and
  $\pk\notin\pSem{p}$ implies $\forall v\in V\colon \pk[f\mut v]\notin\pSem{p}$.
\end{lemma}
\begin{proof}
By straightforward induction on $p$. The lemma holds trivially for $p=\top$ and
  $p=\bot$. For the remaining cases, $\ROf(p)$ means that $p$ does not test $f$,
  so the same path must be taken on $\pk[f\mut v]$ for each $v\in V$.
\end{proof}

Now we show that our smart constructor only produces SPs which are $\RO$.
\begin{lemma}\label{lem:romaint}
  If $f \sqsubset f'$ and $\ROx{f'} (p_0),\ldots, \ROx{f'} (p_{n+1})$, then $\ROf\ (\spsc\ (f, \{ v_0 \mapsto p_0, \ldots, v_n \mapsto p_n \}, p_{n+1}))$.
\end{lemma}
\begin{proof}
  After considering the assumption, we only need to argue that no child $p_i$ in $\spsc\ (f, \{ v_0 \mapsto p_0, \ldots, v_n \mapsto p_n \}, p_{n+1})$ is equal to the default case. This is true by assumption in the if-statement true branch of $\spsc$ and true by construction in the false branch.
\end{proof}

\begin{corollary}
  If $\RO(p_1)$ and $\RO(p_2)$, then all of the following hold:
  \begin{enumerate}
    \item $\RO(p_1 \nf{+} p_2)$,
    \item $\RO(p_1 \nf{\cdot} p_2)$,
    \item $\RO(p_1 \nf{\cap} p_2)$,
    \item $\RO(p_1 \nf{-} p_2)$,
    \item $\RO(p_1 \nf{\oplus} p_2)$,
    \item $\RO(p_1^{\nf{\star}})$,
    %% \item $\RO(\nf{\neg} p_1)$,
    \item $\RO(\NKforall\ f\ p_1)$, and
    \item $\RO(\NKexists\ f\ p_1)$.
  \end{enumerate}
\end{corollary}
\begin{proof}
  By structural induction on $p_1, p_2$. If $p_1,p_2\in\{\zero,\one\}$ the statement holds trivially by the base cases of all the operators. For the inductive case, we consider that the inductive cases use the smart constructor and apply \Cref{lem:romaint} to conclude the result is $\RO$. We need to be slightly careful in the use of the expansion rule, which temporarily constructs a non-canonical SP (e.g., $\SP(f_1, \emptyset, p_2)$ for $p_2$, which is not $\RO$ because the test branch map is empty). This construction still works because we apply the inductive case after expansion, which calls the smart constructor---the premises of \Cref{lem:romaint} are still met for this call by the assumptions for the original $p_1, p_2$.
\end{proof}

As a result of this corollary and the fact that $\top$, $\bot$, and the SPPs for $f\test v$ and $f\testNE v$ are all $RO$, all the SPs generated in KATch are $RO$.

Next, the smart constructor produces symbolic packets which are semantically equivalent to the plain datatype constructor:
\begin{lemma}\label{lem:spsc}
  $\pSem{\spsc(f, \{ v_0 \mapsto p_0, \ldots, v_n \mapsto p_n \}, p_{n+1})}
   = \pSem{\SP(f, \{ v_0 \mapsto p_0, \ldots, v_n \mapsto p_n \}, p_{n+1})}$
\end{lemma}
\begin{proof}
The case that we return $d$ holds trivially (note $d = p_{n+1}$), since if $p_0 = \ldots = p_n = p_{n+1}$ then also $\pSem{p_0} = \ldots = \pSem{p_{n+1}} = \pSem{\SP(f, \{ v_0 \mapsto p_0, \ldots, v_n \mapsto p_n \}, p_{n+1})}$.

    In the other case we only delete branches. This is fine because branches we delete are syntactically the same as the default case. Any packets that would have gone to the deleted branch $p_i$ are handled equivalently by the default (i.e., $p_{n+1}$) case because it is deleted precisely if $p_i=p_{n+1}$.
\end{proof}

Next we will show that two SPs in canonical form are syntactically equal if and only if they
are semantically equal.

%% \begin{remark} Note that to apply $\Sem\cdot$ to an SP we are relying on the
%% straightforward embedding of SP in NetKAT, and in particular the correspondence
%% between NK form and SP datatype given at the beginning of \cref{app:sp_ops}.

%% Furthermore, for an SP $p$, the set of packets it stands for is the
%% \emph{image}: $\Sem{p}(\Pk)$. Abusing notation somewhat, when it is clear that
%% $p$ is an SP, we elide the application of $\Pk$ and use $\Sem{p}$ to stand for the set of packets represented.
%% \end{remark}

\begin{theorem}[RO form is unique]
  For all $p_1, p_2\in \SP$, if $\RO (p_1)$ and $\RO (p_2)$, then $p_1 = p_2 \iff \pSem{p_1} = \pSem{p_2}$.
\end{theorem}
\begin{proof}
  Suppose $\RO (p_1)$ and $\RO (p_2)$.  We will show
  $p_1 = p_2 \iff \pSem{p_1} = \pSem{p_2}$.

  That $p_1 = p_2 \Rightarrow \pSem{p_1} = \pSem{p_2}$ is immediate because $\pSem{\cdot}$ is a function. For the other direction, we will show $p_1\neq p_2 \Rightarrow \pSem{p_1}\neq\pSem{p_2}$. We show this property by induction on $p_1, p_2$ structurally. So let $p_1, p_2\in\SP$ such that $\RO(p_1),\RO(p_2)$, and $p_1\neq p_2$. Considering symmetry, there is only one base case, i.e. that $p_1 = \top$ and $p_2 = \bot$. Then $\pSem{p_1} = \Pk \neq \emptyset = \pSem{p_2}$.

  There are four inductive cases after considering symmetry:
  \begin{enumerate}
    \item Let $p_1 = \top$ and $p_2 = \SP(f, \{\ldots, v_i\mapsto p_i,\ldots\}, p_{n+1})$.
      It suffices to find a packet $\pk\notin\Sem{p_2}$ (since $\Sem{p_1}=\Pk$). From $\RO(p_2)$, there is a $p_i \neq p_{n+1}$ (syntactically!). Apply the induction hypothesis to get that $\pSem{p_i} \neq \pSem{p_{n+1}}$. This means that $\pSem{p_i} \neq \Pk$ or $\pSem{p_{n+1}} \neq \Pk$. In the first case, $\pSem{p_i} \neq \Pk$. Pick a packet $\pk\notin\pSem{p_i}$, then since $\ROx{f'}(p_i)$ for $f'\sqsupset f$, $p_i$ cannot test $f$, so it is also the case that $\pk[f\mut v_i] \notin \pSem{p_i}$. This in turn means $\pk\notin\pSem{p_2}$ and therefore that $\pSem{p_2} \neq \Pk$. On the other hand if instead $\pSem{p_{n+1}} \neq \Pk$, then we pick $\pk\notin\pSem{p_{n+1}}$. By the assumption in \Cref{sec:values}, then we can pick a value $v_{n+1}\notin \{v_0,\ldots,v_n\}$. Then by $\ROx{f'}(p_{n+1})$ for $f'\sqsupset f$ and \Cref{lem:const-f} we have $\pk[f\mut v_{n+1}]\notin \pSem{p_{n+1}}$ and thus $\pk[f\mut v_{n+1}]\notin \pSem{p_2}$, as needed.

    \item Let $p_1 = \bot$ and $p_2 = \SP(f, \{\ldots,v_i\mapsto p_i,\ldots\}, p_{n+1})$. We need to show $\pSem{p_2} \neq \emptyset$. By $\RO(p_2)$ there is a $p_i\neq p_{n+1}$. Apply the induction hypothesis to get $\pSem{p_i} \neq \pSem{p_{n+1}}$. Now this means either $\pSem{p_i}\neq\emptyset$ or $\pSem{p_{n+1}}\neq\emptyset$. These two subcases are similar to those in the previous case.

    \item Let
            $p_1 = \SP(f, \{v_{1,0} \mapsto p_{1,0},\ldots, v_{1,n} \mapsto p_{1,n}\}, p_{1,n+1})$ and
            $p_2 = \SP(f, \{v_{2,0} \mapsto p_{2,0},\ldots, v_{2,m} \mapsto p_{2,m}\}, p_{2,m+1})$.
          We know that $p_1 \neq p_2$. There are two subcases for how this can
          be true:
          \begin{itemize}
            \item We have $v_{1,i} \notin \{v_{2,0},\ldots,v_{2,m}\}$ (w.l.o.g.).
            From $\RO(p_1)$, we know that $p_{1,i} \neq p_{1,n+1}$. Apply the
            induction hypothesis to get that $\pSem{p_{1,i}} \neq \pSem{p_{1,n+1}}$.
            In particular, let $\pk\in\pSem{p_{1,i}\oplus p_{1,n+1}}$. Then
              choose $v_{1,n+1}\notin\{v_{1,0},\ldots,v_{1,n},
              v_{2,0},\ldots,v_{2,m}\}$ and observe
              that $\pk[f\mut v_{1,i}]\in\pSem{p_1} \iff \pk[f\mut
              v_{1,n+1}]\notin\pSem{p_1}$. But from $v_{1,i} \notin \{v_{2,0},\ldots,v_{2,m}\}$ we know both of these packets will go to the default case of $p_2$. And $p_{2,m+1}$ is
              $\ROx{f'}$, which means both packets are in $\pSem{p_2}$ or not in
              $\pSem{p_2}$ because they differ only by $f$. We conclude $\pSem{p_1}\neq\pSem{p_2}$.

            \item Otherwise we can renumber $p_2$ so that
            $v_{1,i} = v_{2,i}$ for all $i \in \{0,\ldots,n\}$ (noting that $(n=m)$ or else we could resolve in the previous case). Moreover there must be some $i\in \{0,\ldots,n+1\}$ for which
            $p_{1,i}\neq p_{2,i}$. Applying the inductive hypothesis, we get
            that $\pSem{p_{1,i}}\neq \pSem{p_{2,i}}$. Let
            $\pk\in\pSem{p_{1,i}}\oplus \pSem{p_{2,i}}$. If $i \leq n$, then
            $\pk[f\mut v_i]\in \pSem{p_1}\oplus\pSem{p_2}$, and otherwise if
            $i=n+1$, we get $v_{n+1}\notin\{v_0,\ldots,v_n\}$ from the
            assumption in \Cref{sec:values} and we have
            $\pk[f\mut v_{n+1}]\in \pSem{p_1}\oplus\pSem{p_2}$. In either case we
            see that $\pSem{p_1}\neq\pSem{p_2}$ as needed.
          \end{itemize}

    \item Let
            $p_1 = \SP(f_1, \{v_{1,0} \mapsto p_{1,0},\ldots, v_{1,n} \mapsto p_{1,n}\}, p_{1,n+1})$,
            $p_2 = \SP(f_2, \{v_{2,0} \mapsto p_{2,0},\ldots, v_{2,m} \mapsto p_{2,m}\}, p_{2,m+1})$, and
            without loss of generality assume $f_1 \sqsubset f_2$.
          Assume $\RO(p_1)$.
          Then there is at least one $i\leq n$ for which $p_{1,i}\neq p_{1,n+1}$.
          By the induction hypothesis, $\pSem{p_{1,i}}\neq \pSem{p_{1,n+1}}$.
          Let $\pk\pSem{p_{1,i}}\oplus\pSem{p_{1,n+1}}$.
          Let $v_{n+1}\notin\{v_0,\ldots,v_n\}$.
          Now consider the two packets $\pk[f_1\mut v_i]$ and $\pk[f_1\mut v_{n+1}]$.
          Clearly by construction one is in $\pSem{p_1}$ and the other is not. But they differ only by $f_1$ and $p_2$ cannot test $f_1$ because $f_1\sqsubset f_2$ and $\RO(p_2)$, so they must either both be in $\pSem{p_2}$ or both not be in $\pSem{p_2}$.
          Thus one of the two packets witnesses $\pSem{p_1}\neq\pSem{p_2}$.

  \end{enumerate}
\end{proof}

Finally, we show that the remaining operations defined above have the expected meanings semantically.

\begin{theorem}\label{thm:soundspops}
  For SPs $p_1, p_2$, all of the following hold.
  \begin{enumerate}
    \item $\pSem{p_1 \nf{+} p_2} = \pSem{p_1} \cup \pSem{p_2}.$
    \item $\pSem{p_1 \nf{\cdot} p_2} = \pSem{p_1} \cdot \pSem{p_2}.$
    \item $\pSem{p_1 \nf{\cap} p_2} = \pSem{p_1} \cap \pSem{p_2}.$
    \item $\pSem{p_1 \nf{-} p_2} = \pSem{p_1} - \pSem{p_2}.$
    \item $\pSem{p_1 \nf{\oplus} p_2} = \pSem{p_1} \oplus \pSem{p_2}.$
    \item $\pSem{p^{\nf{\star}}} = \pSem{p}^\star$
    \item $\pSem{\NKforall\ f\ p} = \{\pk\in\Pk \mid \forall v\in V\colon \pk[f\mut v] \in \pSem{p}\}$
    \item $\pSem{\NKexists\ f\ p} = \{\pk\in\Pk \mid \exists v\in V\colon \pk[f\mut v] \in \pSem{p}\}$
  \end{enumerate}
\end{theorem}
\begin{proof} We again proceed by structural induction on $p_1, p_2$. We only show the cases for $\nf{+}$ and $\NKforall$ because the other cases are similar or reduce directly to the other operators.
  \begin{enumerate}
    \item The base cases are trivial. For the inductive cases, let
        $p_1 = \SP(f, \{\ldots, v_{1,j} \mapsto p_{1,j}, \ldots\},
        p_{1,n+1})$, and let $p_2 = \SP(f, \{\ldots, v_{2,k} \mapsto p_{2,k},\ldots\}, p_{2,m+1})$.
        Let $\pk\in\Pk$. Considering $\pk_f = v_j$, we select a $p_j$: if $v_j \in \{v_{1,0},\ldots, v_{1,n}\}$, then $p_j = p_{1,j}$, or if $v_j\notin\{v_{1,0},\ldots, v_{1,n}\}$, then $p_j = p_{1,n+1}.$ Similarly (treating the default cases uniformly with nondefault cases), we can select the child $p_k$ from $\{p_{2,0},\ldots, p_{2,m+1}\}.$ From this:
        \begin{align*}
          \pk \in \pSem{p_1 \nf{+} p_2} &\iff \pk\in\pSem{p_j \nf{+} p_k} & \text{Definition of $\nf{+}$; Case analysis for $\pk_f$}\\
                                  &\iff \pk\in\pSem{p_j} \cup \pSem{p_k} &\text{Induction hypothesis}\\
                                  &\iff \pk\in\pSem{p_1} \cup \pSem{p_2} &\text{Definitions of $p_1,p_2$}
        \end{align*}

      On the other hand, let $p_1 = \SP(f_1, \{\ldots, v_{1,j} \mapsto p_{1,j},\ldots\}, p_{1,n+1})$, $p_2 = \SP(f_2, \{\ldots, v_{2,k} \mapsto p_{2,k}, \ldots\}, p_{2,m+1})$,  and w.l.o.g., suppose $f_1 \sqsubset f_2$. We see that the result is really by reduction to the previous case, since $p_1$ and $\SP(f_1,\emptyset,p_2)$ \emph{do} match on their top level field ($f_1$). Clearly, $\pSem{p_2} = \pSem{\SP(f_1,\emptyset,p_2)}$, so that this method produces the correct result by the previous argument.
    \setcounter{enumi}{6}
    \item Let $S =\{\pk\in\Pk \mid \forall v\in V\colon \pk[f\mut v] \in \pSem{p}\}$.
      We need to show $\pSem{\NKforall\ f\ p} = S$:
        \begin{itemize}
            \item Base cases: let $p = \top$ or $p = \bot$. Then trivially $\pk\in S$ iff $\pk\in\pSem{\NKforall\ f\ p} = \pSem{p}$.
            \item Let $p = \SP(f, \{v_0\mapsto p_0, \ldots, v_n\mapsto p_n\}, p_{n+1})$. Then $\pk\in S$ iff $\pk\in \pSem{p_i}$ for $i\leq 0 \leq n+1$, since the children (incl. default case) cover all values $v\in V$. But this is the same as $\pk\in \nf{\bigcap}_{i\leq n+1} p_i$.
            \item Let $p = \SP(f', \{v_0\mapsto p_0, \ldots, v_n\mapsto p_n\}, p_{n+1})$, with $f\neq f'$. Then $\pk\in \pSem{p}$ iff $\pk\in\pSem{\NKforall\ f\  p_i}$ for some child $p_i$ of $p$, depending on the value of $\pk_{f_2}$. By the induction hypothesis, this is equivalent to $\pk[f \mut v]\in \pSem{p_i}$ for all $v\in V$. But then we know $\pk[f \mut v]_{f'} = \pk_{f'}$, so this is equivalent to $\pk[f\mut v]\in\pSem{p}$ for all $v\in V$, i.e. that $\pk\in S$, as needed.
        \end{itemize}
  \end{enumerate}
\end{proof}

\section{Operations on Symbolic Packet Programs}\label{app:spp}
In the implementation, there is also datatype defined for SPPs. We define its
operations in this section.

Recall the definition of \SPPn{}:
\begin{align*}
    p \in \SPP \Coloneqq\ &\bot \mid \top \mid
      \underbrace{
        \SPP(f, \{ \ldots, v_i \mapsto \{ \ldots, w_{ij} \mapsto q_{ij}, \ldots \}, \ldots \}, \{ \ldots, w_i \mapsto q_i, \ldots \}, q)
      }_{
        \equiv\ \ \sum_i f \test v_i \cdot \sum_j f \mut w_{ij} \cdot q_{ij}\ +\ \
        (\prod_i f \testNE v_i) \cdot (\sum_i f \mut w_i\ +\ (\prod_i f \testNE w_i) \cdot q)
      }
\end{align*}

As with SPs, we take the semantics of an SPP to be the semantics of its associated \NetKAT expression, shown here.
As before, our $\SPP$ tuple here matches the \SPPn{} object type in the
implementation, and we also only construct \SPPs which are
canonical, in the sense of \Cref{sec:symtrans}. We formalize this
canonicalization property in \Cref{sec:spp-correct}.

First of all, there is a straightforward embedding from SP into SPP:
\[
  \begin{array}{ll}
  \quad\SPPid\colon \SP \to \SPP\\
  \SPPid (\zero) = \zero \qquad \SPPid (\one) = \one\\
  \SPPid(\SP(f, b, d)) = \SPP(f, \{v_i\mapsto\{v_i\mapsto p_i\} \mid v_i\mapsto p_i \in b\}, \emptyset, \SPPid (d))\\
  \end{array}
\]
%% &\SPPid\ p = \begin{cases}
    %% p & \text{if }p = \top \vee p = \bot\\
    %% \SPP(f, \{\ldots, v_i \mapsto \{v_i \mapsto p_i\}, \ldots\}, \emptyset, p_d)
        %% &\text{if }p = \SP(f, \{\ldots, v_i \mapsto p_i, \ldots\}, p_d)
%% \end{cases}

Observe that by applying the embedding to the SP definition for $f\test v$ and $f\testNE v$, we obtain the SPP forms for them:
\begin{align*}
  (f\test v) &=  \SPP(f, \{v\mapsto\{v\mapsto \one\}\}, \emptyset, \zero) \\
  (f\testNE v) &= \SPP(f, \{v\mapsto\emptyset\}, \emptyset, \one)
\end{align*}

In addition to this embedding, we only need one additional constructor:
\begin{align*}
  (f \mut v) &\triangleq \SPP(f, \emptyset, \{v \mapsto \one\}, \zero)
\end{align*}

The remaining SPPs we will construct from the operations given below.

\subsection{Canonicalization}\label{app:sppsc}

Once again, we define a smart constructor to help maintain the invariant that
all SPPs in memory are canonical:
\begin{align*}
  \sppsc\colon (F\times (V\Mapsto(V\Mapsto\SPP)) \times (V\Mapsto\SPP) \times \SPP) \to \SPP
\end{align*}

Here is the definition:

\begin{align*}
\quad\sppsc(f,b,m,d) &\triangleq\
    \text{if } m' = \emptyset \land b'' = \emptyset \text{ then } d \text{ else } \SPP(f, b'', m', d)\\
    \text{ where}&\\
    m' &= \{ v \sto p \mid v \sto p \in m \land p \neq \bot \}\\
    b' &= \{ v \sto \{ w \sto p \mid w \sto p \in m_i \land p \neq \bot \} \mid v \sto m_i \in b \} \cup \{ v \sto m' \mid v \sto \bot \in m' \land v \notin b \}\\
    b'' &= \{ v \sto m_i \mid v \sto m_i \in b' \land m_i \neq (\text{if } v \in \keys(m') \lor d = \bot \text{ then } m' \text{ else } m' \cup \{ v \sto d \}) \}\\
\end{align*}

\framebox{\parbox{\linewidth}{\textbf{Erratum. } An earlier version of this section contained a mistake in the above definition, on the line assigning $b'$. On the righthand side, it read $v\mapsto\zero$ instead of $v\mapsto m'$. In that version, it was possible to create inputs to the smart constructor whose semantics would be altered by canonicalization.  The previous proof of \Cref{lem:sppscsem} was therefore also wrong; it erroneously treated a case as simpler than it was. The statement and proof of the lemma are updated to match this version. However, the results of the main algorithms were not affected because the inputs which cause the bug are guaranteed not to occur in the way SPPs are built from \NetKAT expressions.}}

\subsection{Operations}

In the rest of this section we give the definitions for the remaining operations used to
build SPPs.
\begin{align*}
    \nf{+}, \nf{\cdot}, \nf{\cap}, \nf{\oplus}, \nf{-}\ :\ \SPP \times \SPP \to \SPP \qquad
    \nf{\star}\ :\ \SPP \to \SPP \qquad
\end{align*}

The base cases are identical to the corresponding operations on SPs, so we omit them here.
Similar to SPs, the inductive case for $\nf{+}, \nf{\cap}, \nf{\oplus},$ and $\nf{-}$ can be expressed at once. We do so here using $\pmop$. The notation $b(v;d)$ means the child $b(v)$ if $v \in \keys(b)$, or by default $d$ otherwise.
\[
  \begin{array}{ll}
    \quad \SPP(f, b_p, m_p, d_p)\pmop \SPP(f, b_q, m_q, d_q) \triangleq  \sppsc(f,b'\!,m_p \mathop{\vec{\pm}} m_q,d_p \pmop  d_q)\\
    \whereindent \text{ where }\  b' = \{v \sto \get{p}{v} \mathop{\vec{\pm}} \get{q}{v} \mid v\in \keys(b_p \cup b_q \cup m_p \cup m_q) \}\\
    \quad m_1 \mathop{\vec{\pm}} m_2 \triangleq \{v \,\sto\, m_1(v;\bot) \pmop m_2(v;\bot) \mid v\in \keys(m_1 \cup m_2)\},\\
    \quad \get{p}{v} \triangleq b_p(v; \text{ if } v \in m_q \lor d_p = \bot \text{ then } m_p \text{ else } m_p \cup \{ v \sto d_p \}) \text{\quad (similarly for $q(v)$)}\\
  \end{array}
\]
Sequencing ($\nf{\cdot}$) is more subtle and needs to be defined separately (using $\vec{\Sigma} \triangleq \text{$n$-ary sum w.r.t. $\vec{+}$}$ as defined above):
\[
  \begin{array}{ll}
    \quad \SPP(f, b_p, m_p, d_p)\cdotop \SPP(f, b_q, m_q, d_q) \triangleq  \sppsc(f,b'\!,m_A \mathop{\vec{+}} m_B, d_p \cdotop  d_q)\\
    \whereindent \text{ where }\  b' = \{v \sto \unionmaps{v'\! \sto p' \in \get{p}{v}}{ \{ w'\! \sto p' \nf{\cdot} q' \mid w'\! \sto q' \in \get{q}{v'}\} } \mid v\in \keys(b_p \cup b_q \cup m_p \cup m_q \cup m_A)\}\\
    \whereindent \phantom{\text{ where }}\  m_A = \unionmaps{v'\! \sto p' \in m_p}{ \{ w'\! \sto p' \nf{\cdot} q' \mid w'\! \sto q' \in \get{q}{v'}\} },
      \quad m_B = \{ w'\! \sto d_p \nf{\cdot} q' \mid w'\! \sto q' \in m_q \} \\
  \end{array}
\]

The expansion rule works similarly to the SP case. Namely, when the fields do not match, we expand the SPP with the greater field by:
\[
    \quad p \equiv \SPP(f, \emptyset, \emptyset, p) \text {\qquad if \qquad} p \in \{\top, \bot, \SPP(f'\!, b, m, d)\} \text{ where } f \sqsubset f'\\
\]

Finally star is defined by:
\[
  \begin{array}{ll}
p^{\nf{\star}} \triangleq \nf{\Sigma}_{i=0,1,\ldots} p^{i}\\
    \whereindent\text{ where } \quad p^0 \triangleq \top \\
    \whereindent\phantom{\text{ where }} \quad p^n \triangleq p \nf{\cdot} p^{n-1} \quad \text{ for }n \geq 1
  \end{array}
\]
In the implementation we perform repeated concatenations until the SPP no longer changes (i.e., the fixpoint is reached).

\subsection{Push and Pull}

First we define two helper functions which ``flatten'' SPPs into SPs:
\begin{align*}
    \mathsf{fwd} :\ \SPP \to \SP \qquad\qquad \mathsf{bwd} :\ \SPP \to \SP
\end{align*}

The idea of $\mathsf{fwd}$ is to collect a symbolic packet representing the set of output packets
for an SPP when the input packets can be any packet:
\[
\begin{array}{ll}
 \quad\mathsf{fwd}(\one) \triangleq \one\\
 \quad\mathsf{fwd}(\zero) \triangleq \zero\\
 \quad\mathsf{fwd}(\SPP(f, b, m, d)) \triangleq \spsc(f, b_A \vecplus b_B \vecplus b_C \vecplus b_D, \mathsf{fwd}(d))\\
 \quad\whereindent \text{ where }\ b_A = \{ v' \sto \mathsf{fwd}(p') \mid v' \sto p' \in m \}\\
 \quad\whereindent \phantom{\text{ where }}\ b_B = \unionmaps{v' \sto m' \in b}{ \{ w' \sto \mathsf{fwd}(p') \mid w' \sto p' \in m' \} }\\
 \quad\whereindent \phantom{\text{ where }}\  b_C = \{ v' \sto \mathsf{fwd}( d) \mid v' \in \keys(b_B) \setminus (\keys(b) \cup \keys(m)) \}\\
 \quad\whereindent \phantom{\text{ where }}\ b_D = \{ v' \sto \bot \mid v' \in \keys(b) \setminus \keys(b_B) \}
\end{array}
\]

Conversely, $\mathsf{bwd}$ gives the set of input packets which, when input to the given SPP, produce some nonempty set of packets as output:
\[
\begin{array}{ll}
 \quad\mathsf{bwd}(\one) \triangleq \one\\
 \quad\mathsf{bwd}(\zero) \triangleq \zero\\
 \quad\mathsf{bwd}(\SPP(f, b, m, d)) \triangleq \spsc(f, b_A \vecplus b_B, d' \plusop \mathsf{bwd}(d))\\
 \quad\whereindent \text{ where }\  d' = \nf{\Sigma}_{v' \sto p' \in m} \mathsf{bwd}(p')\\
 \quad\whereindent \phantom{\text{ where }}\  b_A = \{ v' \sto \nf{\Sigma}_{w' \sto p' \in m'} \mathsf{bwd}(p') \mid v' \sto m' \in b \}\\
 \quad\whereindent \phantom{\text{ where }}\  b_B = \{ v \sto d' \mid v \in \keys(m) \setminus \keys(b) \}
\end{array}
\]

Recall the \sf{push} and \sf{pull} operations which are described in \Cref{sec:symtrans}:
\begin{align*}
    \mathsf{push} :\ \SP \to \SPP \to \SP \qquad\qquad \mathsf{pull} :\ \SPP \to \SP \to \SP
\end{align*}

We conclude the section with their definitions:
\begin{align*}
  &\mathsf{push}\ p\ s \triangleq \mathsf{fwd}((\SPPid\ p) \cdotop s),\\
  &\mathsf{pull}\ s\ p \triangleq \mathsf{bwd}(s \cdotop (\SPPid\ p)),
 \end{align*}

\subsection{Correctness of SPP operations}\label{sec:spp-correct}

Finally, we give proofs that the SPP operations maintain a canonical form and have the expected
semantic meanings.

\begin{definition}[Reduced and Ordered]
  An SPP $p$ is \emph{reduced and ordered for a field $f\in F$} if it satisfies:
  \begin{align*}
    &\ROsppf\colon \SPP \to 2\\
    &\ROsppf (\zero) \triangleq \one \qquad\qquad \ROsppf (\one) \triangleq \one\\
  \end{align*}
  $\ROsppf (\SPP(f',b,m,d)) \triangleq \one$ if all of the following hold:
  \begin{enumerate}
    \item $b\neq\emptyset \vee m\neq\emptyset$
    \item $f \sqsubset f'$
    \item $\forall (v_i\mapsto m_i)\in b.\ (\forall (v_j\mapsto p_j)\in m_i.\ \ROsppx{f'}(p_j) \wedge p_j\neq \zero)$\\
      $\wedge ((m_i\neq m \wedge (v_i\in\keys(m)\vee d=\zero))\vee(m_i\neq(m\cup \{v_i\mapsto d\})\wedge v_i\notin\keys(m)\wedge d\neq\zero))$
    \item $(\forall v\mapsto p_i) \in m.\ \ROsppx{f'}(p_i) \wedge p_i\neq \zero$
    \item $\ROsppx{f'}(d)$
  \end{enumerate}
  $\ROsppf (\SPP(f',b,m,d)) \triangleq \zero$ otherwise.

  More generally, an SPP $p$ is \emph{reduced and ordered} if it satisfies:
  \begin{align*}
    &\ROspp\colon \SPP \to 2\\
    &\ROspp\ (p) \triangleq \exists f \in F\colon \ROsppf\ (p)
  \end{align*}
\end{definition}

Using this definition, we first establish that the smart constructor only produces reduced and ordered SPPs, provided the children in the test and mutation branches are already reduced and ordered.

\begin{lemma}\label{lem:sppscro}
  Let $f\in F, b\in V\Mapsto(V\Mapsto\SPP), m\in V\Mapsto\SPP, d\in\SPP$. If all of the following hold:
  \begin{enumerate}[(a)]
    \item $f \sqsubset f'$,
    \item $\forall v_i\mapsto m_i\in b.\ (\forall v_j\mapsto p_j\in m_i.\ \ROsppx{f'}(p_j))$\\
    \item $\forall v_i\mapsto p_i \in m.\ \ROsppx{f'}(p_i)$
    \item $\ROsppx{f'}(d)$
  \end{enumerate}
      Then $\ROsppf\ (\sppsc(f, b, m, d))$.
\end{lemma}
\begin{proof}
  Clearly if $\sppsc(f, b, m, d) = \zero$ or $\sppsc(f, b, m, d) = \one$ then the result holds trivially.
  So suppose $\SPP(f, b'', m', d) = \sppsc(f, b, m, d)$. Note the variable names are chosen to match those in the definition of $\sppsc$). We establish the criteria of $\ROspp$ one at a time, following the list in the definition.
  \begin{enumerate}
    \item Holds because both $b''$ and $m'$ being empty is prevented by the outermost if statement in $\sppsc$.
    \item By assumption (a).
    \item Let $(v_i\mapsto m_i)\in b''$ and let $v_j \mapsto p_j\in m_i$. We have $\ROsppx{f'}(p_j)$ by assumption (b). In addition, $p_j \neq \zero$ because such bindings are filtered out in the construction of $b'$, and the construction of $b''$ only removes more bindings (not adding any). The remaining condition matches precisely what is checked by the construction of $b''$.
    \item Let $(v_i\mapsto p_i)\in m'$. That we have $\ROsppx{f'}(p_i)$ holds by assumption (c). That $p_i\neq\zero$ is true because such bindings are removed by the construction of $m'$ from $m$.
    \item By assumption (d).
  \end{enumerate}
\end{proof}

Next, it follows that all of the operations maintain the property of being reduced and ordered.

\begin{corollary}\label{lem:sppromaint}
  If for SPPs $p_1,p_2$, we have $\ROspp(p_1)$ and $\ROspp(p_2)$, then all of the following hold:
  \begin{enumerate}
    \item $\ROspp(p_1 \nf{+} p_2)$,
    \item $\ROspp(p_1 \nf{\cdot} p_2)$,
    \item $\ROspp(p_1 \nf{\cap} p_2)$,
    \item $\ROspp(p_1 \nf{-} p_2)$,
    \item $\ROspp(p_1 \nf{\oplus} p_2)$,
    \item $\ROspp(p_1^{\nf{\star}})$,
  \end{enumerate}
\end{corollary}
\begin{proof}
  Because we only use the smart constructor to create SPPs in the definitions of all of the operations,
  we only need to show that the premises of \Cref{lem:sppscro} are satisfied when the smart constructor is used.
  Indeed, these premises are guaranteed by $\ROspp(p_1)$, and $\ROspp(p_2)$, and comparisons of $f, f'$ in the definitions of the operations. We only need to be slightly careful in the cases where the smart constructor is \emph{not} used (since the expansion form is not RO), but each of these is a direct reduction to a case where the smart constructor is used. Note that the default case of the expansion form is still RO by assumption, and the expansion rule itself maintains the field order.
\end{proof}

The main benefit of canonicalization is that, after establishing reduced an ordered form, we can check semantic equivalence by checking syntactic equivalence, which we show next.

\begin{theorem}
  For all $p_1, p_2\in \SP$, if $\ROspp (p_1)$ and $\ROspp (p_2)$, then $p_1 = p_2 \iff \pSem{p_1} = \pSem{p_2}$.
\end{theorem}
\begin{proof}
  Suppose $\ROspp(p_1)$ and $\ROspp(p_2)$.
  ($\Rightarrow$) is immediate because $\Sem{\cdot}$ is a function. ($\Leftarrow$). Suppose $p_1 \neq p_2$ (syntactically). We will show $\Sem{p_1} \neq \Sem{p_2}$ by induction on the structure of $p_1, p_2$:
  \begin{enumerate}
    \item Let $p_1=\zero$, and let $p_2=\one$. Then $\Sem{p_1}=\emptyset$ and $\pk\in\Sem{p_2}(\pk)$ for all $\pk\in\Pk$.
    \item Let $p_1=\zero$, and let $p_2=\SPP(f,b,m,d)$.
      We only need to show that $\Sem{p_2}$ is nonempty. First of all, if $d\neq\zero$, we are done, since by induction this means $\Sem{d}\neq\emptyset$, and moreover $\pkp\in\Sem{d}(\pk)$ means that $\pkp\in\Sem{p_2}(\pk[f\mut v])$ for any $v\notin\keys(b)\cup\keys(m)$. So assume $d=\zero$. Either $b\neq\emptyset$ or $m\neq\emptyset$. If $m\neq\emptyset$ we are done, because $v\mapsto p \in m$ means that $p\neq\zero$, and thus by the induction hypothesis that there are $\pk,\pkp\in\Pk$ for which $\pkp\in\Sem{p}$. It follows that $\pkp[f\mut v]\in\Sem{p_2}(\pk)$. So assume $m = \emptyset$. Then there is some $(v_i\mapsto m_i)\in b$. Now because $d=\zero$ and $m=\emptyset$, item (3) of the definition of $\ROsppf$ says that $m_i\neq\emptyset$. Thus there is some $v_j\mapsto p_j\in m_i$ for which $p_j\neq\zero$. So there are packets $\pk,\pkp$ such that $\pkp\in\Sem{p_j}(\pk)$, and thus $\pkp[f\mut v_j]\in\Sem{p_2}(\pk[f\mut v_i])$.
    \item Let $p_1=\one$, and let $p_2=\SPP(f,b,m,d)$. We need to show $p_2$ modifies $(\pkp\in\Sem{p_2}(\pk)$ for $\pkp\neq\pk$) or drops a packet ($\Sem{p_2}(\pk)=\emptyset$).
      If $d\neq\one$, then by the induction hypothesis, either $\pkp\in\Sem{d}(\pk)$ for $\pkp\neq\pk$ or $\Sem{d}(\pk)=\emptyset$. Then for $\pk[f\mut v]$ for $v\notin\keys(b)\cup\keys(m)$, we have either $\pkp[f\mut v]\in\Sem{p_2}(\pk[f\mut v])$ with $\pkp[f\mut v]\neq \pk[f\mut v]$ or $\Sem{p_2}(\pk[f\mut v])=\emptyset$, respectively. So assume $d=\one$. If $(v_i\mapsto p_i)\in m$, then since $p_i\neq\zero$, we have $\pk[f\mut v_i]\in\Sem{p_2}(\pk[f\mut v_j])$ for any $v_j\neq v_i$. So assume $m=\emptyset$. Since then $b\neq\emptyset$,  there is some $(v_i\mapsto m_i)\in b$. From item (3) of the definition of $\ROsppf$, we know $m_i\neq\{v_i\mapsto\one\}$. There are several subcases based on what $m_i$ can be:
      \begin{itemize}
        \item If $m_i=\emptyset$, we are done since then $\Sem{p_2}(\pk[f\mut v_i])=\emptyset$.
        \item If $m_i=\{v_j\mapsto\one,\ldots\}$, with $v_j\neq v_i$. Then $\pk[f\mut v_j]\in\Sem{p_2}\pk[f\mut v_i]$ (for any $\pk$).
        \item If $m_i = \{v_j \mapsto p_j,\ldots\}$ with $p_j\neq\one$ (but possibly $v_i=v_j$), then we have by the induction hypothesis either a packet $\pk$ such that $\Sem{p_j}(\pk)=\emptyset$ (in which case $\Sem{p_2}(\pk[f\mut v_i])=\emptyset$), or some $\pkp[f\mut v_j]\neq\pk[f\mut v_j]$ such that $\pkp\in\Sem{p_j}(\pk)$ (in which case $\pkp[v\mut v_j]\in\Sem{p_2}(\pk[f\mut v_i]$).
      \end{itemize}
    \item Let $p_1=\SPP(f_1,b_1,m_1,d_1)$, and let $p_2=\SPP(f_2,b_2,m_2,d_2)$. We proceed by the cases in which $p_1$ and $p_2$ may differ syntactically, and show why each leads to a difference in semantics:
      \begin{itemize}
        \item Suppose $f_1\sqsubset f_2$, and suppose $v_i\mapsto p_i\in m_1$. Then $p_i\neq\zero$, meaning there are packets $\pk,\pkp$ for which $\pkp\in\Sem{p_i}(\pk)$. Pick $v_j\notin\keys(b_1)\cup\keys(m_1)$. Then
          $\pkp[f_1\mut v_i]\in\Sem{p_1}(\pk[f_1\mut v_j])$, but $p_2$ cannot modify $f_1$ at all, so we must have $\pkp[f_1\mut v_i]\notin\Sem{p_2}(\pk[f_1\mut v_j])$. On the other hand suppose $m_1=\emptyset$. Then there must be $v_i\mapsto m_i\in b_1$. Now we have two cases, in item (3) of $\ROspp (p_1)$:
          \begin{itemize}
            \item Let $d_1=\zero$. Then $m_i\neq\emptyset$. Let $v_j\mapsto p_j\in m_i$. Note $p_j\neq\zero$.
              Then for $\pk,\pkp$ such that $\pkp\in\Sem{p_j}(\pk)$, we have $\pkp[f_1\mut v_j]\in\Sem{p_1}(\pk[f_1\mut v_i])$. If $v_i\neq v_j$, we are done, since $p_2$ cannot modify $f_1$. Otherwise let $v_j'\notin\keys(b_1)\cup\keys(m_1)$. Then because $d_1=\zero$, $\Sem{p_1}(\pk[f_1\mut v_i])\neq\emptyset=\Sem{p_1}(\pk[f_1\mut v_j'])$. But $p_2$ cannot test $f_1$, so we must have $\Sem{p_2}(\pk[f_1\mut v_i])=\Sem{p_2}(\pk[f_1\mut v_j'])$.
            \item Let $d_1\neq\zero$. Then $m_i\neq\{v_i\mapsto d_1\}$. This means for any $v_j\notin\keys(b_1)\cup\keys(m_1)$, there must be some $\pk,\pkp$ such that
              $\pkp[f_1\mut v_i] \in \Sem{p_1}(\pk[f_1\mut v_i]) \iff \pkp[f_1\mut v_j]\notin\Sem{p_1}(\pk[f_1\mut v_j])$. But $p_2$ cannot test $f_1$, so we are done.
          \end{itemize}
          For the remaining cases, we let $f = f_1 = f_2$.
        \item Suppose $d_1\neq d_2$.
          Then by the induction hypothesis there is a packet $\pk$ for which $\Sem{d_1}(\pk)\neq\Sem{d_2}(\pk)$. We just pick a value $v\notin\keys(b_1)\cup\keys(m_1)\cup\keys(m_2)\cup\keys(b_2)$, and we get $\Sem{p_1}(\pk[f\mut v])\neq\Sem{p_2}(\pk[f\mut v])$.
        \item Suppose $\keys(m_1)\neq\keys(m_2)$.
          Then (w.l.o.g.) let $v_i\mapsto p_i\in m_1$ and $v_i\notin\keys(m_2)$, and choose $v\notin\keys(b_1)\cup\keys(m_1)\cup\keys(b_2)\cup\keys(m_2)$. Then since $p_i\neq\zero$, there are $\pk,\pkp\in\Pk$ such that $\pkp\in\Sem{p_i}(\pk)$. This means $\pkp[f\mut v_i]\in\Sem{p_1}(\pk[f\mut v])$, but $\pkp[f\mut v_i]\notin\Sem{p_2}(\pk[f\mut v])$.
        \item Suppose $v\mapsto p_i\in m_1$, while $v\mapsto p_j\in m_2$, and $p_i\neq p_j$.
          By the induction hypothesis, there is a $\pk\in\Pk$ for which $\Sem{p_i}(\pk)\neq\Sem{p_j}(\pk)$. Pick $v\notin\keys(b_1)\cup\keys(b_2)\cup\keys(m_1)\cup\keys(m_2)$. Then $\Sem{p_1}(\pk[f\mut v]) \neq \Sem{p_2}(\pk[f\mut v])$.
        \item Assume $d_1=d_2$ and $m_1=m_2$, but suppose $\keys(b_1)\neq\keys(b_2)$.
          Then (w.l.o.g.) let $v_i\mapsto m_i\in b_1$ and $v_i\notin\keys(b_2)$.
          Now there are two cases depending on whether $v_i\in\keys(m_1)\vee d_1=\zero$, in which case we know $m_i\neq m_1$. But since $m_1 = m_2$, this means $m_i\neq m_2$. Following the reasoning of the previous two cases (different keys or different SPP children), we find $\Sem{p_1}(\pk[f\mut v_i])\neq\Sem{p_2}(\pk[f\mut v_i])$.
          On the other hand, $v_i\notin\keys(m_1)\wedge d_1\neq\zero$, in which case $m_i\neq(m_1\cup\{v_i\mapsto d_1\})$. Since $m_1 = m_2$ and $d_1 = d_2$, this means that also $m_i\neq(m_2\cup\{v_i\mapsto d_2\})$.
          There are 4 subcases for how this inequality can happen:
          \begin{itemize}
            \item Suppose $v_i\notin\keys(m_i)$. Then there are $\pk,\pkp$ for which $\pkp\in\Sem{d_2}(\pk)$ meaning $\pkp[f\mut v_i]\in\Sem{p_2}(\pk[f\mut v_i])$, but $\pkp[f\mut v_i]\notin\Sem{p_1}(\pk[f\mut v_i])$.
            \item Otherwise $v_i\mapsto p\in m_i$ for $p\neq d_2$. By the induction hypothesis, there is a $\pk\in\Pk$ for which $\Sem{p}(\pk)\neq\Sem{d_2}(\pk)$. But then $\Sem{p_1}(\pk[f\mut v_i])\neq\Sem{p_2}(\pk[f\mut v_i])$.
          \item The other two subcases are different keys or different SPP children for $m_i$ and $m_2$, but then the counterexample packets are constructed similarly to the previous cases.
          \end{itemize}
        \item Suppose $v\mapsto m_i\in b_1$, while $v\mapsto m_j\in b_2$, and $\keys(m_i)\neq\keys(m_j)$.
          Then (w.l.o.g.) let $v'\mapsto p\in m_i$, but $v'\notin\keys(m_j)$. From $\ROspp(p_1)$, we know $p\neq\zero$. So there are packets $\pk,\pkp\in\Pk$ such that $\pkp\in\Sem{p}(\pk)$. But then $\pkp[f\mut v']\in\Sem{p_1}(\pk[f\mut v])$ but
          $\pkp[f\mut v']\notin\Sem{p_2}(\pk[f\mut v])$ because $p_2$ does not have an assignment to $v'$ in the $v$ branch for $f$.
        \item Suppose $v\mapsto m_i\in b_1$, while $v\mapsto m_j\in b_2$, and
          $v'\mapsto p_i\in m_i$, while $v'\mapsto p_j\in m_j$ for $p_i\neq p_j$.
          By the induction hypothesis, for some packet $\pk\in\Pk$, we have $\Sem{p_i}(\pk)\neq\Sem{p_j}(\pk)$.
          But then $\Sem{p_1}(\pk[f\mut v])\neq\Sem{p_2}(\pk[f\mut v])$.
      \end{itemize}
  \end{enumerate}
\end{proof}

The rest of the results have to do with the semantics of SPPs. First we justify the ``expansion'' construction used in several operations to simplify the construction:
\begin{lemma}\label{lem:expansion}
Suppose $f_1,f_2\in F$, with $f_1 \sqsubset f_2$, and $p = \SPP(f_2, b, m, d)$. Then
$\Sem{\SPP(f_1, \emptyset, \emptyset, p))} = \Sem{p}$.
\end{lemma}
\begin{proof}
  For any $\pk\in\Pk$, clearly $\pk_{f_1}\notin\keys(\emptyset)$ (i.e., we neither test or set $f_1$), and it follows that $\pkp\in\Sem{p}(\pk)$ iff $\pkp\in{\Sem{\SPP(f_1,\emptyset,\emptyset, p)}}(\pk)$.
\end{proof}

Next, we show the SPP smart constructor does not change the semantics, compared to the plain datatype constructor:

\begin{lemma}[Semantic equivalence of SPP smart constructor]\label{lem:sppscsem}
  For all $f\in F, b\in(V\Mapsto (V\Mapsto\SPP)), m\in(V\Mapsto\SPP), d\in\SPP$, we have that
  $\Sem{\SPP(f,b,m,d)} = \Sem{\sppsc(f,b,m,d)}$.
\end{lemma}
\begin{proof}
  We argue in three steps, referring to the local variables in the $\sppsc$ definition:

  \begin{itemize}
    \item First, $\Sem{\SPP(f,b,m,d)} = \Sem{\SPP(f,b',m',d)}$

      ($\subseteq$).
      Let $\pk,\pkp\in\Pk$, and suppose $\pkp\in\Sem{\SPP(f,b,m,d)}(\pk)$.
      Let $v = \pk_f$ and let $v' = \pkp_f$.
      Otherwise there are three cases:
      \begin{enumerate}[(a)]
        \item Suppose $v \in \keys(b)$. Clearly since $\pkp\in\Sem{\SPP(f,b,m,d)}(\pk)$, we know
        $b[v][v'] \neq \bot$. This means that $b'[v][v'] = b[v][v']$, and thus that
        $\pkp \in\Sem{\SPP(f,b',m',d)}(\pk)$.
        \item Suppose $v \in \keys(m)\setminus \keys(b)$. Then $v' \mapsto s\in m$ and
                    $\pkp\in\Sem{s}(\pk[f\mut v'])$.
              Thus $s\neq\bot$, and we so must have $v' \mapsto s \in m'$, and therefore
              that $\pkp\in\Sem{\SPP(f,b',m',d)}(\pk)$.
        \item Finally if $v \notin \keys(b) \cup \keys(m)$, then $\pkp\in\Sem{d}(\pk)$ (in which case we are done), or $\pkp\in\Sem{s}(\pk[f\mut v'])$ for some $v'\mapsto s\in m$. By the same reason in the previous case,
        $v'\mapsto s \in m'$, and thus $\pkp\in\Sem{\SPP(f,b',m',d)}(\pk)$.
      \end{enumerate}

      ($\supseteq$).
      Let $\pk,\pkp\in\Pk$, and suppose $\pkp\in\Sem{\SPP(f,b',m',d)}(\pk)$.
      Let $v = \pk_f$ and let $v' = \pkp_f$.
      Again there are three cases:
      \begin{enumerate}[(a)]
      \item Suppose $v \in \keys(b')$. There are two subcases:
        \begin{enumerate}[(i)]
        \item Suppose $v'\in\keys(b'[v])$ and $b'[v][v']\neq\bot$
            (i.e. $v$ is in the left addend of $b'$), then by construction
            $b[v][v']\neq\bot$ and thus $\pkp\in\Sem{\SPP(f,b,m,d)}(\pk)$.
        \item Otherwise it must be that $v \mapsto \bot \in m'$ and $v'\in\keys(m')$.
            This, in turn, means $v'\mapsto s \in m$ and $s\neq\bot$.
            We conclude $\pkp\in{\Sem{\SPP(f,b,m,d)}}(\pk)$.
        \end{enumerate}
      \item Suppose $v\in\keys(m')\setminus\keys(b')$. Then $v'\in\keys(m')$,
            and the same reasoning as Case (a)(ii) applies.
      \item Finally if $v\notin\keys(b')\cup\keys(m')$, then either $v'\in\keys(m')$,
          in which case the reasoning of Case (a)(ii) applies, or $\pkp\in\Sem{d}(\pk)$,
          which also means $\pkp\in{\Sem{\SPP(f,b,m,d)}}(\pk)$.
      \end{enumerate}

    %% We see that $m'$ removes any bindings that send a value $v$ to $\bot$. Since this binding would
    %% cause a packet $\pk$ with $\pk_f = v$ to be dropped instead of possibly being transformed by
    %% the default case, we replicate this behavior by adding binding for $v$ in $b'$, which gets sent to
    %% $\emptyset$, dropping the packet as needed.

    %% The other change to get from $b$ to $b'$ is to remove bindings to $\bot$ in all the maps, but this does not change the semantics since applicable packets are dropped in either case.

    \item Next, $\Sem{\SPP(f,b',m',d)} = \Sem{\SPP(f,b'', m', d)}$

    Note that $b'' \subseteq b'$ since $b''$ is constructed by removing mappings from $b'$.

    ($\subseteq$).
      Let $\pk,\pkp\in\Pk$, and suppose $\pkp\in\Sem{\SPP(f,b',m',d)}(\pk)$.
      Let $v = \pk_f$ and let $v' = \pkp_f$.
      There are two cases:
      \begin{enumerate}[(a)]
      \item Let $v\in\keys(b'')$. We know $b''[v][v'] = b'[v][v']$,
                    so $\pkp\in\Sem{\SPP(f,b'',m',d)}(\pk)$.
      \item Let $v\in\keys(b')\setminus\keys(b'')$.
            Let $m_i = b'[v]$. Note that $\pkp\in\Sem{m_i[v']}(\pk)$. There are two subcases:
            \begin{enumerate}[(i)]
                \item Suppose $v\in\keys(m')$ or $d =\bot$. Then from the fact that
                $v$ was removed, we know $m_i = m'$, which in turn means $v'\in\keys(m')$,
                and thus that $\pkp\in\Sem{\SPP(f,b'',m',d)}$ as needed.
                \item Otherwise, we know $m_i = m' \cup \{v\mapsto d\}$. So either $v' = v$,
                in which case that means $\pkp\in\Sem{d}(\pk)$, or $v'\neq v$, in which case
                $\pkp\in\Sem{m'[v']}(\pk[f\mut v'])$. In both cases, we have
                $\pkp\in\Sem{\SPP(f,b'',m',d)}$.
            \end{enumerate}
      \item Let $v\notin\keys(b')$. Then either $\pkp\in\Sem{d}(\pk)$ or $v'\mapsto s \in m'$ and $\pkp\in\Sem{s}(\pk[f\mut v'])$. In either case, we have $\pkp\in\Sem{\SPP(f,b'',m',d)}(\pk)$ as needed.

      \end{enumerate}

    ($\supseteq$).
      Let $\pk,\pkp\in\Pk$, and suppose $\pkp\in\Sem{\SPP(f,b'',m',d)}(\pk)$.
      Let $v = \pk_f$ and let $v' = \pkp_f$.
      There are two cases:
      \begin{enumerate}[(a)]
      \item Let $v\in\keys(b'')$. Then $\pkp\in\Sem{b''[v][v']}(\pk)$,
      and since $b''\subseteq b'$, we have $\pkp\in\Sem{b'[v][v']}(\pk)$,
      and thus $\pkp\in\Sem{\SPP(f,b',m',d)}(\pk)$.
      \item Let $v\in\keys(b')\setminus\keys(b'')$.
            Then there are 3 subcases:
            \begin{enumerate}[(i)]
                \item Suppose $v\in\keys(m')$. So we must have
                $v'\in\keys(m')$ also ($v\in\keys(m')$ means $d$ is ignored).
                From $v\in\keys(b')\setminus\keys(b'')$, we know $b'[v] = m'$
                and thus that $\pkp\in\Sem{b'[v][v']}(\pk)$, giving
                $\pkp\in\Sem{\SPP(f,b',m',d)}$.
                \item Suppose $d = \bot$. We know $v'\in\keys(m')$ because $\Sem{d} =\emptyset$.
                From here, reasoning follows the previous case that 
                $\pkp\in\Sem{\SPP(f,b',m',d)}$.
                \item Otherwise, we know $b'[v] = m' \cup \{v\mapsto d\}$.
            \end{enumerate}
      \item Let $v\notin\keys(b')$. Then either $\pkp\in\Sem{d}(\pk)$
      or $v'\mapsto s \in m'$ and $\pkp\in\Sem{s}(\pk[f\mut v'])$. In either case,
      because $v\notin\keys(b')$, we have $\pkp\in\Sem{\SPP(f,b',m',d)}(\pk)$ as needed.
      \end{enumerate}

    \item Finally, if $b''=\emptyset$ and $m'=\emptyset$, then $\Sem{\SPP(f,b'',m',d)} = \Sem{d}$. This is exactly \Cref{lem:expansion}.
  \end{itemize}

  Putting the steps together, we have shown $\Sem{\SPP(f,b,m,d)} = \Sem{\sppsc(f,b,m,d)}$, as required.
\end{proof}

\begin{lemma}[Correctess of $\SPPid$]\label{lem:sppid}
Let $p$ be an $\SP$. For any $\pk,\pkp\in\Pk$, then $\pkp\in\Sem{\SPPid\ p}(\pk)$ iff $\pkp=\pk$ and $\pkp\in\Sem{p}$.
\end{lemma}
\begin{proof}
Straightforward induction on $p$.
\end{proof}

The next lemma makes the meaning of $\get{p}{v}$ more precise:

\begin{lemma}[Behavior of $\get{p}{v}$]\label{lem:getspp}
    Let $p = \SPP(f, b, m, d), v\in V$, and $\pk\in\Pk$. Then:
    \[ \Sem{p}(\pk[f\mut v]) = \bigcup_{(v_i\mapsto p_i)\in \get{p}{v}} \Sem{p_i}(\pk[f\mut v_i])\]
\end{lemma}
\begin{proof}
  Let $\pk,\pkp\in\Pk$.
  We proceed by the four cases of $\get{p}{v}$ that corresponding to $p = \SPP(f,b,m,d)$. In each case, we will show $\pkp\in\Sem{p}(\pk[f\mut v])$ iff $\pkp\in\Sem{p_i}(\pk[f\mut v_i])$, observing that the union is actually disjoint.
  \begin{itemize}
    \item Let $v\in \keys(b)$. Then each $(v_i\mapsto p_i)$ is a value and child guarded by $(f\test v)$, so we have $\pkp\in\Sem{p}(\pk[f\mut v])$ iff $\pkp\in\Sem{p_i}(\pk[f\mut v_i])$.
    \item Let $v\in \keys(m)\setminus\keys(b)$. Let $v\mapsto p_i\in m$. Since $v\notin\keys(b)$, $\pk[f\mut v]$ passes the $(f\testNE v)$
      that guard the assignments in $m$. The significance of $v\in\keys(m)$ is that we know the default case
      $d$ has negative guards $(f\testNE v)$ for these values, so $d$ cannot be applied to this packet.
      Once again $\pkp\in\Sem{p}(\pk[f\mut v])$ iff $\pkp\in\Sem{p_i}(\pk[f\mut v_i])$
    \item Let $v\notin \keys(b)\cup\keys(m)$, but $d=\bot$. This is similar to the previous case in that the
    $p_i$ in $m$ are precisely the reachable children, but now $\Sem{p}(\pk[f\mut v]) = \Sem{d} = \Sem{\bot} = \emptyset$.
    \item Let $v\notin\keys(b)\cup\keys(m)$, and $d\neq\bot$. We apply each child $v_i \mapsto p_i$ from $m$, but we also have the mapping $v \mapsto d$ (i.e. $v_i = v$ and $p_i = d$). This is appropriate because $v\notin\keys(b)\cup\keys(m)$ ensures we pass all the guards on $d$, and we have $v_i = v$ because no update to the packet is performed before applying $d$.
  \end{itemize}
\end{proof}

We need one more tiny lemma describing the behavior of $\vecplus$ that we will reuse.
\begin{lemma}\label{lem:vecplus}
For maps $m_1,m_2\colon V\Mapsto \SPP$, provided that $\Sem{p_i \plusop p_j} = \Sem{p_i} \cup \Sem{p_j}$ for all $p_i,p_j$ in the values of $m_1,m_2$, then we have $\pkp\in\Sem{p}(\pk)$ for some binding $v\mapsto p \in m_1 \vecplus m_2$ iff $\pkp\in\Sem{p_i}(\pk)$ for a binding $v\mapsto p_i \in m_1$ or $\pkp\in\Sem{p_j}(\pk)$ for a binding $v\mapsto p_j \in m_1$.
\end{lemma}
\begin{proof}
  The claim is true by inspection of the definition of $\vec{+}$: the only nontrivial case of that analysis is that there are bindings for $v$ in both $m_1$ and $m_2$, but in that case we have $p = p_i \nf{+} p_j$, and we apply the additional assumption to get the lemma as stated.
\end{proof}

We are ready to show that the operations on SPPs have their expected meanings:

\begin{theorem}\label{thm:sppcorrect}
  For SPPs $p_1, p_2$, all of the following hold.
  \begin{enumerate}
    \item $\Sem{p_1 \nf{+} p_2} = \Sem{p_1} \cup \Sem{p_2}.$
    \item $\Sem{p_1 \nf{\cdot} p_2}(\pk) = \{\pkp\in\Pk \mid \pkpp\in\Sem{p_1}(\pk), \pkp\in\Sem{p_2}(\pkpp)\}.$
    \item $\Sem{p_1 \nf{\cap} p_2} = \Sem{p_1} \cap \Sem{p_2}.$
    \item $\Sem{p_1 \nf{-} p_2} = \Sem{p_1} - \Sem{p_2}.$
    \item $\Sem{p_1 \nf{\oplus} p_2} = \Sem{p_1} \oplus \Sem{p_2}.$
    \item $\Sem{p^{\nf{\star}}} = \bigcup_{n\geq 0} \Sem{p^n}.$
  \end{enumerate}
\end{theorem}
\begin{proof} We again proceed by structural induction on $p_1, p_2$. We only show the cases for $\nf{+}$ and  $\nf{\cdot}$ because the other cases are similar or reduce directly to the other operators.
  In each case, we consider arbitrary packets $\pk,\pkp\in\Pk$.

  Now we proceed according to the cases in the theorem statement.
  \begin{enumerate}
    \item We need to show $\pkp\in\Sem{p_1\nf{+}p_2}(\pk)$ iff $\pkp\in\Sem{p_1}(\pk)$ or $\pkp\in\Sem{p_2}(\pk)$. The base cases (where both $p_1$ or $p_2$ are $\top$ or $\bot$) are trivial. The cases where only one of $p_1,p_2$ is $\zero$ or $\one$ reduce to the general case by expansion. We start with the case that the two SPPs have the same top level field, i.e., let $p_1 =\SPP(f, b_1, m_1, d_1), p_2 = \SPP(f, b_2, m_2, d_2)$. Note that we rely on \Cref{lem:sppscsem} to reason as if plain SPP constructor is called instead of the smart constructor. Evidently $\pkp\in\Sem{p_1\nf{+}p_2}(\pk)$ iff $\pkp$ is in one of the children found in the branches, mutation, or default cases. Which case is applicable depends only on $\pk_f$. So we handle these three by case, letting $v = \pk_f$:
      \begin{itemize}
        \item Now we look at the test branch. In this case, $\pkp\in\Sem{p_1\nf{+}p_2}(\pk)$ iff there is a binding $v\mapsto p$ in $\{v \mapsto \get{p_1}{v} \vecplus \get{p_2}{v} \mid v\in\keys(b_1)\cup\keys(b_2)\cup\keys(m_1)\cup\keys(m_2)\}$.
          By \Cref{lem:vecplus}\footnote{The assumption about the values of $m_1,m_2$ is satisfied by the induction hypothesis.} and \Cref{lem:getspp}, this is the same as saying $\pkp\in\Sem{p_i}(\pk[f\mut v_i])$ for a binding $v_i\mapsto p_i\in \get{p_1}{v}$ or $\pkp\in\Sem{p_j}(\pk[f\mut v_j])$, for a binding $v_j\mapsto p_j\in\get{p_2}{v}$. Combining these, this is equivalent to $\pkp\in\Sem{p_1}(\pk) \cup \Sem{p_2}(\pk)$ as needed.
        \item For the mutation branch, $v\notin\keys(b_1)\cup\keys(b_2)\cup\keys(m_1)\cup\keys(m_2)$. In this case, $\pkp\in\Sem{p_1\nf{+}p_2}(\pk)$ whenever there is a binding $(v'\mapsto p)\in (m_1 \vecplus m_2)$ such that $\pkp\in\Sem{p}(\pk)$. By \Cref{lem:vecplus}, this is equivalent to saying that either $\pkp\in\Sem{p_i}(\pk)$ for some binding $(v'\mapsto p_i)\in m_1$, or that $\pkp\in\Sem{p_j}(\pk)$ for some $(v'\mapsto p_j)\in m_2$. But now we are done, because $\pkp\in\Sem{p_i}(\pk)$ or $\pkp\in\Sem{p_j}(\pk)$ iff $\pkp[f\mut v']\in\Sem{p_1}(\pk[f\mut v])\cup\Sem{p_j}(\pk[f\mut v])$.
        \item Now consider the default case, where again $v\notin \keys(b_1)\cup\keys(b_2)\cup\keys(m_1)\cup\keys(m_2)$. By the induction hypothesis, $\pkp\in\Sem{d_1\nf{+}d_2}(\pk)$ iff $\pkp\in\Sem{d_1}(\pk) \cup \Sem{d_2}(\pk)$. Then since $v\notin\keys(b_1)\cup\keys(m_1)$ and $v\notin\keys(b_2)\cup\keys(m_2)$, this is true iff $\pkp\in\Sem{p_1}\cup\Sem{p_2}$.
      \end{itemize}

      The other cases, namely that the top level field is different, reduce naturally to the above case by applying \Cref{lem:expansion}.

    \item We need to show $\pkp\in\Sem{p_1\nf{\cdot}p_2}(\pk)$ iff there is a $\pkpp\in\Pk$ such that $\pkpp\in\Sem{p_1}(\pk)$ and $\pkp\in\Sem{p_2}(\pkpp)$. The base cases (where either $p_1$ or $p_2$ are $\top$ or $\bot$) are trivial. We proceed by cases, letting $v=\pk_f$:
      \begin{itemize}
        \item Test branches case. Thus $v\in\keys(b_1)\cup\keys(b_2)\cup\keys(m_1)\cup\keys(m_2)$,
          or $v\in b_2[v']$ for some $v'\in\keys(b_2)$. So there is a binding for $v$ in (what in the definition is called) $b'$.
          Applying \Cref{lem:vecplus}\footnote{The assumption of the lemma is satisfied by part (1).} inductively to this map, we get that $\pkp\in\Sem{p\nf{\cdot}p'}(\pk)$ means
          there must be corresponding bindings
          $(v_2\mapsto p)\in\get{p_1}{v}$ and $(v_3\mapsto p')\in\get{p_2}{v_2}$.
          Letting $\pkpp=\pk[f\mut v_2]$, this is equivalent to: $\pk[f\mut v_2]\in\Sem{p_1}$ and $\pkp\in\Sem{p_2}(\pkpp)$, completing the case.
        \item For the mutation branch case, $v\notin\keys(b_1)\cup\keys(b_2)\cup\keys(m_1)\cup\keys(m_2)$
          Then $\pkp\in\Sem{p_1\nf{\cdot}p_2}(\pk)$ is equivalent to $\pkp=\pk[f\mut v]$ for some binding for
          $v\in m_A \vecplus m_B$. We apply \Cref{lem:vecplus}, and divide into two subcases. Namely $\pkp\in\Sem{p_1\nf{\cdot}p_2}(\pk)$ iff we come from:
            \begin{itemize}
              \item[($m_A$)] We apply \Cref{lem:vecplus} again, inductively, to the folded sum in $m_A$, finding that
                there are bindings $(v_2\mapsto p)\in m_1$ and
                $(v_3\mapsto p')\in\get{p_2}{v_2}$ such that $\pkp\in\Sem{p\nf{\cdot}p'}(\pk)$. Applying the induction hypothesis, there are $\pk,\pkpp,\pkp$ such that $\pkpp\in\Sem{p}(\pk)$ and $\pkp\in\Sem{p'}(\pkpp)$. This is equivalent to saying $\pkpp[f\mut v_2]\in\Sem{p_1}(\pk)[f\mut v]$ and $\pkp[f\mut v_3]\in\Sem{p_2}(\pkpp)$ as needed.
              \item[($m_B$)] Then $\pkp\in\Sem{d_1\nf{\cdot}p}(\pk)$ for some binding $(v'\mapsto p)\in m_2$. We apply the induction hypothesis to see that this means there is a $\pkpp\in\Sem{d_1}(\pk)$ for which $\pkp\in\Sem{p}(\pkpp)$. In turn
                this is equivalent to: $\pkpp[f\mut v]\in\Sem{p_1}(\pk[f\mut v])$ and $\pkp[f\mut v']\in\Sem{p_2}(\pkpp[f\mut v])$ as needed.
            \end{itemize}
          \item We are in the default case when $v\notin\keys(b_1)\cup\keys(b_2)\cup\keys(m_1)\cup\keys(m_2)\cup\keys(m_A\vecplus m_B)$.
          By the induction hypothesis, $\pkp\in\Sem{d_1\nf{\cdot}d_2}(\pk)$ iff there is a $\pkpp$ such that
          $\pkpp\in\Sem{d_1}(\pk)$ and $\pkp\in\Sem{d_2}(\pkpp)$, but from $v\notin\keys(b_1)\cup\keys(m_1)$, this is the same as $\pkpp\in\Sem{p_1}(\pk)$ and from $v\notin\keys(b_2)\cup\keys(m_2)$, this is $\pkp\in\Sem{p_2}(\pkpp)$, completing the proof.
      \end{itemize}
  \end{enumerate}
\end{proof}

Now we justify the definitions of $\fwd$ and $\bwd$.

\begin{theorem}\label{thm:fwdbwdcorrect}
  For SPPs $p_1, p_2$, all of the following hold.
  \begin{enumerate}
    \item $\Sem{\mathsf{fwd}\ s} = \{\pkp\in\Pk \mid \exists \pk\in\Pk.\ \pkp\in\Sem{s}(\pk)\}.$
    \item $\Sem{\mathsf{bwd}\ s} = \{\pkp\in\Pk \mid \exists \pk\in\Pk.\ \pk\in\Sem{s}(\pkp)\}.$
  \end{enumerate}
\end{theorem}
\begin{proof}
    We proof each case separately:
    \begin{enumerate}
    %----------------------- Fwd --------------------------
    \item The base cases, $s=\one$ and $s=\zero$, are trivial. So let $s = \SPP(f, b, m, d)$. Let $\pkp\in\Pk$, let $v' = \pk_f$, and let $B = \bigcup_{v\in \keys(b)} \keys(b[v])$. We will show that $\pkp\in\Sem{\fwd\ s}$ iff there is a packet $\pk\in\Pk$ such that $\pkp\in\Sem{s}(\pk)$. For each such packet $\pk$, we will reserve $v$ to refer to $\pk_f$. There are eight cases for the behavior of $\fwd$ depending on whether $v'$ is in each of the three sets $B$, $\keys(b)$, and $\keys(m)$. Note for reference that $\keys(b_B) = B$, and $\keys(b_A) = \keys(m)$.

      \begin{enumerate}
        \item Suppose $v'\in B \setminus (\keys(b) \cup \keys(m))$.
          By applying \Cref{lem:vecplus} to the definition of $\fwd$, we see that $\pkp\in\Sem{\fwd\ s}$ iff $v'\in\keys(b_B)$ corresponding to some binding $v' \mapsto \fwd(p')\in b_B$ for which $\pkp\in\Sem{p'}$, or we have a binding $v'\mapsto \fwd(d)\in b_C$. In the first case, $v'\mapsto \fwd(p')\in b_B$ iff there is a binding $v'\mapsto p'\in b[v]$ for some $v$. Applying the induction hypothesis, this is true iff $\pkp\in\Sem{p'}$ for packet $\pk=\pkp[f\mut v]$, and $\pkp\in\Sem{s}(\pk)$ as needed. In the second case, our binding is $v'\mapsto\fwd(d)\in b_C$. By the induction hypothesis, $\pkp\in\Sem{\fwd(d)}$. Because we know that $v'\notin\keys(b)\cup\keys(m)$, this also means that for $\pk=\pkp$ we have $\pkp\in\Sem{s}(\pk)$ as needed.
        \item Suppose $v'\in (B \cap \keys(m)) \setminus \keys(b)$.
          Applying \Cref{lem:vecplus}, we have $\pkp\in\Sem{\fwd\ s}$ iff $v'\in\keys(b_A)$ or $v'\in\keys(b_B)$.
          In the first case, there is a binding $v'\mapsto \fwd(p')\in m$. By induction hypothesis, that means we can choose $\pk=\pkp$ and have both $\pkp\in\Sem{p'}$ and $\pkp\in\Sem{s}(\pk)$.
          In the other case, then the binding in $b_B$ works just as in the previous case.
        \item Suppose $v'\in \keys(m) \setminus (B \cup \keys(b))$.
          In this case, $\pkp\in\Sem{\fwd\ s}$ iff there is an applicable binding in $b_A$. The argument matches the argument for $b_A$ in the previous case.
        \item Suppose $v'\in (B \cap \keys(b)) \setminus \keys(m)$.
          Then $\pkp\in\Sem{\fwd\ s}$ iff there is an applicable binding in $b_B$. The argument matches (a) and (b).
        \item Suppose $v'\in B\cap\keys(b)\cap\keys(m)$. This case is effectively the same as case (b).
        \item Suppose $v'\in (\keys(b)\cup\keys(m)) \setminus B$.  This case is effectively the same as case (c).
        \item Suppose $v'\in \keys(b) \setminus (B \cup \keys(m))$.
          Observe that it cannot be the case that $\pkp\in\Sem{\fwd\ s}$, since the only bindings available
          are $v'\mapsto\bot\in b_D$. But we cannot have $\pkp\in\Sem{s}(\pk)$ either, since any packet $\pk$ which matches a key in $b$ or $m$ will have $\pk_f$ updated to a value in $B$ or $m$, but $v'\notin B\cup m$. In addition, a $\pk$ with $\pk_f\notin b \cup m$ cannot produce $\pkp$ either, a since in that case no update occurs (so $v' = v$), but we know $v'\in b$.
        \item Suppose $v'\notin \keys(b) \cup \keys(m) \cup B$.
          Clearly $\pkp\in\Sem{\fwd\ s}$ iff $\pkp\in\Sem{\fwd\ d}$. By the induction hypothesis, that means $\pkp\in\Sem{d}(\pk)$ for some $\pk$. In turn this means $\pkp\in\Sem{s}(\pk)$ as needed.
      \end{enumerate}
    %----------------------- Bwd --------------------------
    \item The base cases are trivial. Let $s = \SPP(f, b, m, d)$. As in the proof above for $\fwd$, we let $\pkp\in\Pk$, let $v' = \pk_f$, and let $B = \bigcup_{v\in \keys(b)} \keys(b[v])$. We will show that $\pkp\in\Sem{\fwd\ s}$ iff there is a packet $\pk\in\Pk$ such that $\pkp\in\Sem{s}(\pk)$. For each such packet $\pk$, we will reserve $v$ to refer to $\pk_f$.
        Note for reference that $\keys(b_A) = B$ and $\keys(b_B) = \keys(m) \setminus \keys(b)$.
        \begin{enumerate}
            \item Suppose $v'\in B \setminus (\keys(b) \cup \keys(m))$. Then $\pkp\in\Sem{\bwd\ p\ s}$ iff there is
            a binding $v' \mapsto m' \in b$ for which there is a binding $v' \mapsto \nf{\sum}_{w'\mapsto p'\in m'} \bwd(p') \in b_A$. By the induction hypothesis (applied inductively to the sum), this is equivalent to there being a packet $\pkpp = \pkp[f\mut w']$ such that there is an output packet $\pk\in\Sem{p'}(\pkpp)$. Now remember that $p'$ is a child of $s$, so this means that $\pk\in\Sem{s}(\pkp)$ as we wanted.
            \item Suppose $v'\in (B \cap \keys(m)) \setminus \keys(b)$. Then $\pkp\in\Sem{\bwd\ p\ s}$ iff there is a binding in $b_A$, as in the previous case, or there is a binding in $v\mapsto d'\in b_B$. This in turn is true iff there is
            a binding $v''\mapsto p'$ such that $\pkp\in\Sem{\bwd\ p'}$. Indeed, from the fact that $v'\notin\keys(b)$, and by the induction hypothesis, this means that the child $p'$ is reachable for $\pkp$ and moreover produces $\pkp[f\mut v'']\in\Sem{s}(\pkp)$.
            \item Suppose $v'\in \keys(m) \setminus (B \cup \keys(b))$. Then $\pkp\in\Sem{\bwd\ p\ s}$ iff there is a binding in $b_B$, and the argument follows the previous case.
            \item Suppose $v'\in (B \cap \keys(b))$. Then $\pkp\in\Sem{\bwd\ p\ s}$ iff there is a binding in $b_A$, and the argument follows the argument in (a).
            %%\item Suppose $v'\in B\cap\keys(b)\cap\keys(m)$. Then $\pkp\in\Sem{\bwd\ p\ s}$ iff
            % (Included a with previous case)
            \item Suppose $v'\in (\keys(b)\cup\keys(m)) \setminus B$. In this case
            $v'$ cannot be in $b_A \vecplus b_B$, so $\pkp\in\Sem{\bwd\ p\ s}$ iff $\pkp\in\Sem{d' \nf{+} \bwd(d)}$. By the induction hypothesis, this is equivalent to either $\pkp\in\Sem{d'}$ or $\Sem{d}(\pk)$ for some $\pk$. In the later case, $\pkp\in\Sem{s}(\pk)$ and we are done. In the former case, there is a binding $v''\mapsto p'\in d'$ for which $\pkp[f\mut v'']\in\Sem{s}(\pkp)$, completing the case.
            %%\item Suppose $v'\in \keys(b) \setminus (B \cup \keys(m))$. Then $\pkp\in\Sem{\bwd\ p\ s}$ iff
            \item Suppose $v'\notin \keys(m) \cup B$. Then $\pkp\in\Sem{\bwd\ p\ s}$ iff $\pkp\in\Sem{d' \nf{+} \bwd(d)}$, and the argument follows the previous case.
        \end{enumerate}
    \end{enumerate}
\end{proof}

We conclude by demonstrating the correctness of $\mathsf{push}$ and $\mathsf{pull}$.

\begin{corollary}
  For $\mathsf{push}$ and $\mathsf{pull}$, the following hold:
  \begin{enumerate}
    \item $\pkp \in \Sem{\mathsf{push}\ p\ s} \iff \exists \pk\in\pSem{p}.\ \pkp\in\Sem{s}(\pk).$
    \item $\pkp \in \Sem{\mathsf{pull}\ s\ p} \iff \exists \alpha \in \pSem{p}.\ \alpha \in \Sem{s}(\pkp).$
  \end{enumerate}
\end{corollary}
\begin{proof}
    Each part follows from the respective definitions, and applying \Cref{thm:fwdbwdcorrect}, \Cref{lem:sppid}, and \Cref{thm:sppcorrect}.
\end{proof}

}{} % End "isExtendedVersion" if-statement

\end{document}